\documentclass[11pt]{article}

\usepackage[letterpaper,margin=1.0in]{geometry}

\usepackage{amsmath, amssymb, amsthm, amsfonts}
\usepackage{thmtools,thm-restate}
\usepackage{bbm}

\usepackage{ifthen}

\usepackage{cite}

\usepackage{graphicx}
\usepackage{xcolor}
\usepackage{algorithm}
\usepackage[noend]{algpseudocode}

\usepackage{dsfont}
\newcommand{\1}{\mathds{1}}

\usepackage[textsize=tiny]{todonotes}

\usepackage{framed}
\usepackage[framemethod=tikz]{mdframed}
\usepackage[bottom]{footmisc}
\usepackage{enumitem}
\setitemize{noitemsep,topsep=3pt,parsep=3pt,partopsep=3pt}
\usepackage[font=small]{caption}
\usepackage{xspace}

\definecolor{darkgreen}{rgb}{0,0.5,0}
\definecolor{darkblue}{rgb}{0,0,0.6}
\usepackage{hyperref}
\hypersetup{
    unicode=false,          
    colorlinks=true,        
    linkcolor=darkblue,          
    citecolor=darkgreen,        
    filecolor=magenta,      
    urlcolor=cyan           
}
\usepackage[capitalize, nameinlink]{cleveref}
\Crefname{theorem}{Theorem}{Theorems}
\Crefname{lemma}{Lemma}{Lemmas}
\Crefname{claim}{Claim}{Claims}
\Crefname{remark}{Remark}{Remarks}
\Crefname{observation}{Observation}{Observations}

\newtheorem{theorem}{Theorem}[section]
\newtheorem{lemma}[theorem]{Lemma}
\newtheorem{meta-theorem}[theorem]{Meta-Theorem}
\newtheorem{claim}[theorem]{Claim}

\newtheorem{corollary}[theorem]{Corollary}

\newtheorem{definition}{Definition}[section]

\algnewcommand\algorithmicswitch{\textbf{switch}}
\algnewcommand\algorithmiccase{\textbf{case}}

\algdef{SE}[SWITCH]{Switch}{EndSwitch}[1]{\algorithmicswitch\ #1\ \algorithmicdo}{\algorithmicend\ \algorithmicswitch}%
\algdef{SE}[CASE]{Case}{EndCase}[1]{\algorithmiccase\ #1}{\algorithmicend\ \algorithmiccase}%
\algtext*{EndSwitch}%
\algtext*{EndCase}%

\newcommand{\eps}{\varepsilon}

\newcommand{\CONGEST}{$\mathsf{CONGEST}$\xspace}
\newcommand{\LOCAL}{$\mathsf{LOCAL}$\xspace}

\newcommand{\local}{\LOCAL}

\newcommand{\poly}{\operatorname{\text{{\rm poly}}}}

\renewcommand{\tilde}{\widetilde}

\DeclareMathOperator{\E}{\mathbb{E}}
\newcommand{\Labels}{\Sigma}

\newcommand{\utility}{\mathbf{u}}
\newcommand{\cost}{\mathbf{c}}

\newcommand{\fC}{\mathcal{C}}

\newcommand{\fE}{\mathcal{E}}

\newcommand{\fS}{\mathcal{S}}

\newcommand{\separation}{\text{s}}

\usepackage{appendix}
\newcommand{\FullOrShort}{full}
\ifthenelse{\equal{\FullOrShort}{full}}{
	
  \newcommand{\fullOnly}[1]{#1}
  \newcommand{\shortOnly}[1]{}
	\renewcommand{\baselinestretch}{1.00}
  
  }{
    \renewcommand{\baselinestretch}{1.00}
    \newcommand{\fullOnly}[1]{}
    \newcommand{\IncludePictures}[1]{}
   
}
\begin{document}
\date{}
\title{Near-Optimal Deterministic Network Decomposition \\ and Ruling Set, and Improved MIS
}

\author{
  Mohsen Ghaffari \\
  \small{MIT}\\
  \small{ghaffari@mit.edu}
  \and
  Christoph Grunau \\
  \small{ETH Zurich}\\
  \small{cgrunau@inf.ethz.ch}
}
\maketitle

\begin{abstract}
This paper improves and in two cases nearly settles, up to logarithmically lower order factors, the deterministic complexity of some of the most central problems in distributed graph algorithms, which have been studied for over three decades:
\begin{itemize}
    \medskip
    \item \textbf{Near-Optimal Network Decomposition}: We present a deterministic distributed algorithm that computes a network decomposition in $\tilde{O}(\log^2 n)$ rounds with $O(\log n)$ diameter and $O(\log n)$ colors. This round complexity is near-optimal in the following sense: even given an ideal network decomposition, using it (in the standard way) requires round complexity equal to the product of diameter and number of colors, and that is known to be $\tilde{\Omega}(\log^2 n)$. We find this near-optimality remarkable, considering the rarity of optimal deterministic distributed algorithms and that for network decomposition, even the first polylogarithmic round algorithm was achieved only recently, by Rozhon and Ghaffari [STOC 2020], after three decades.
    \medskip
    
    \item \textbf{Near-Optimal Ruling Set}: We present a deterministic distributed algorithm that computes an $O(\log\log n)$ ruling set---i.e., an independent set such that each node is within its $O(\log\log n)$ distance---in $\tilde{O}(\log n)$ rounds. This is an exponential improvement on the $O(\log n)$ ruling set of Awerbuch, Goldberg, Luby, and Plotkin [FOCS'89], while almost matching their $O(\log n)$ round complexity. Our result's round complexity nearly matches the $\tilde{\Omega}(\log n)$ lower bound of Balliu, Brandt, Kuhn, and Olivetti [STOC 2022] that holds for any $\poly(\log\log n)$ ruling set.

    \medskip
    \item \textbf{Improved Maximal Independent Set (MIS)}: We present a deterministic distributed algorithm for computing an MIS in $\tilde{O}(\log^{5/3} n)$ rounds. This improves on the $\tilde{O}(\log^2 n)$ complexity achieved by Ghaffari and Grunau [STOC 2023] and breaks the log-squared barrier necessary for any method based on network decomposition. By known reductions, the $\tilde{O}(\log^{5/3} n)$ round complexity also applies to deterministic algorithms for maximal matching, $\Delta+1$ vertex coloring, and $(2\Delta-1)$ edge coloring. Also, via the shattering technique, the improvement spreads also to randomized complexities of these problems, e.g., the new state-of-the-art randomized complexity of $\Delta+1$ vertex coloring is now $\tilde{O}((\log\log n)^{5/3})$.
\end{itemize}
\end{abstract}

\setcounter{page}{0}
\thispagestyle{empty}

\newpage
{\small
\tableofcontents
\setcounter{page}{0}
\thispagestyle{empty}
}
  
\setcounter{page}{0}
\thispagestyle{empty}

\newpage
\section{Introduction}
We present improved deterministic algorithms for three of the central problems in the area of distributed graph algorithms, which have been studied for over three decades: \textit{Network Decomposition} (ND), \textit{Ruling set} (RS), and \textit{Maximal Independent Set} (MIS). 
For the network decomposition and ruling set problems, our algorithms almost settle the round complexity up to logarithmically smaller factors. For MIS, our algorithm breaks a natural barrier of previous methods. 

\subsection{Context}
\paragraph{Distributed model.} We work with the standard synchronous message-passing model of distributed computing, often referred to as the \local model, due to Linial~\cite{linial1987LOCAL}. The network is abstracted as an $n$-node undirected graph $G=(V, E)$, where each node represents one processor and an edge between two nodes indicates that those two processors can communicate directly. Each processor has a unique $b$-bit identifier, where we typically assume $b=O(\log n)$. Initially, nodes do not know the network topology $G$, except for potentially knowing some global parameters such as a polynomial upper bound $N\in [n, \poly(n)]$ on $n$. Computations and communications take place in synchronous rounds. Per round, each process/node can send one message to each of its neighbors, receives their messages, and then performs arbitrary computations. In the \local model, the message sizes are not bounded. The model variant where message sizes are limited to $O(\log n)$-bits is called \CONGEST~\cite{peleg00}. At the end of the computation, each node should know its own output, e.g., in the MIS problem, each node should know whether it is in the computed independent set.

\paragraph{Deterministic distributed algorithms.} Devising efficient deterministic distributed algorithms is known to be challenging. Linial, in his celebrated work~\cite{linial1987LOCAL,linial92}, wrote ``\textit{Getting a deterministic polylog-time algorithm for MIS seems hard, though simple randomized distributed algorithms are known.}" The latter referred to $O(\log n)$ round randomized MIS algorithms of Luby~\cite{luby86abbrv} and Alon, Babai, and Itai~\cite{alon86}. For a long time, the best known deterministic complexity was $2^{O(\sqrt{\log n})}$~\cite{awerbuch89, panconesi-srinivasan}. The first $\poly(\log n)$ round deterministic MIS algorithm was presented by Rozhon and Ghaffari~\cite{rozhonghaffari20}. The key ingredient was a deterministic $\poly(\log n)$ round algorithm for network decomposition. There have been improvements since then, some \cite{GGR20,elkin2022deterministic} based on refinements of \cite{rozhonghaffari20}, and others via different methods \cite{GhaffariK21,ghaffari2023netdecomp,faour2022local,ghaffarigrunau2023fasterMIS}. Our work contributes to, and partially closes, this line with new technical ingredients that nearly settle the complexity for network decomposition and ruling set, and break a barrier for MIS. 

\subsection{Network Decomposition}
Network decomposition was defined in the pioneering work of Awerbuch, Goldberg, Luby, and Plotkin~\cite{awerbuch89} as a general and powerful tool for devising distributed algorithms, especially deterministic ones (though it also has wide applications in randomized algorithms). Given an $n$-node graph $G=(V, E)$, a \textit{Network Decomposition} (ND) of it with \textit{$c(n)$ colors} and \textit{$d(n)$ diameter} is a partitioning of vertices into sets $V_1, V_2, \dots, V_{c(n)}$ such that, in the subgraph $G[V_i]$ induced by each \textit{color class }$V_i$, each connected component has diameter at most $d(n)$. 

\paragraph{General recipe for using network decomposition.} If we are given an ND with small $c(n)$ and $d(n)$, it vastly simplifies the task of building a distributed algorithm, often boiling down the problem to (a few iterations of) those in low-diameter graphs, where global coordination is easy. For instance, we can compute an MIS easily in $O(c(n)\cdot d(n))$ rounds: we have $c(n)$ stages, one for each color class. In stage $i$, we compute an MIS $S_i$ of the (non-eliminated) nodes in $V_i$ in $O(d(n))$ rounds, by leveraging the diameter of each component of $G[V_i]$, and we then add $S_i$ to the output MIS and mark as \textit{eliminated} all neighbors of $S_i$ in color classes $V_{i+1}$, $V_{i+2}$, \dots, $V_{c(n)}$. 

More generally, it is simple and known that one can use this ND-based approach as a general recipe to turn \textit{any sequential local algorithm} (formally, an algorithm in the $\mathsf{SLOCAL}$ model, and informally including any greedy algorithm where we decide about different elements sequentially, each based on the decisions made already in its constant-hop neighborhood) into \textit{a distributed local algorithm}, in the $\mathsf{LOCAL}$ model, with $O(c(n)\cdot d(n))$ round complexity. See \cite{ghaffari2017complexity} for details.

\paragraph{Existential bounds.} For any $n$-node graph, there exists a network decomposition with $c(n)=O(\log n)$ colors and $d(n)=O(\log n)$ diameter~\cite{Awerbuch-Peleg1990,linial93}. These are often regarded as the ideal bounds. Once given, the time to use such a decomposition---in the standard way, as done above---is $O(c(n)\cdot d(n))=O(\log^2 n)$. It is also known that this is roughly the right bound, in the sense if we have $c(n)$ color and $d(n)$ diameter network decomposition for all $n$-node graphs, then $c(n)\cdot d(n) =\tilde{\Omega}(\log^2 n)$: on one hand, since there are graphs with chromatic number $\Omega(\log^2 n)$ and girth $\Omega(\log n/\log\log n)$\cite{bollobas2004extremal}, if $d(n)=o(\log n/\log \log n)$, we must have $c(n)=\Omega(\log^2 n)$. The reason is that if $d(n)$ is below the girth, each component of each color class induces a tree and thus can be $2$ colored, which means that the chromatic number is at most $2c(n)$.
On the other hand, Linial and Saks~\cite[Theorem 3.2]{linial93} showed a graph where if $c(n)=o(\log n/\log \log n)$, then we must have $d(n)=\Omega(\log^2 n)$. Hence, the worst-case bound is $c(n)\cdot d(n)=\Omega((\frac{\log n}{\log\log n})^2)$.

\paragraph{Distributed computation and known results.} The early work of Awerbuch et al.~\cite{awerbuch89} and its later sharpening by Panconesi and Srinivasan~\cite{panconesi-srinivasan} gave a \textit{deterministic} distributed algorithm that computes a decomposition with $2^{O(\sqrt{\log n})}$ colors and $2^{O(\sqrt{\log n})}$ diameter in $2^{O(\sqrt{\log n})}$ rounds. Linial and Saks\cite{linial93} gave an $O(\log^2 n)$ round \textit{randomized} algorithm for computing an ideal decomposition with $O(\log n)$ colors and $O(\log n)$ diameter.\footnote{One can make their construction run \textit{artificially} in $O(\log n)$ rounds, by building the colors in parallel, but this does not seem to have any benefits once we come to using the decomposition. Moreover, the diameter in their work was based on a weaker definition, where distances are measured in the original graph rather than in the induced subgraph, but that was ironed out by Elkin and Neiman~\cite{elkin16_decomp} via an algorithm of Miller, Peng, and Xu\cite{miller2013parallel}.} It took over three decades, until the work of Rozhon and Ghaffari\cite{rozhonghaffari20}, to achieve a \textit{deterministic} network decomposition algorithm with polylogarithmic bounds and construction time. The bounds have been improved in a number of steps~\cite{GGR20,elkin2022deterministic,ghaffari2023netdecomp}. The state of the art algorithm gives a decomposition with $O(\log n)$ colors and $\tilde{O}(\log n)$ diameter in $\tilde{O}(\log^3 n)$ rounds~\cite{ghaffari2023netdecomp}. 

\paragraph{Our contribution.} We present a deterministic distributed algorithm that almost settles this long line of research, by constructing a network decomposition with $O(\log n)$ colors and $O(\log n)$ diameter in the round complexity of $\tilde{O}(\log^2 n)$. 

\begin{theorem}
    \label{thm:NDmain}
    There is a deterministic distributed algorithm that in any $n$-node graph $G=(V, E)$ computes a network decomposition with $O(\log n)$ colors and $O(\log n)$ diameter in $\tilde{O}(\log^2 n)$ rounds of the \local model.
\end{theorem}
 We view this complexity as reaching a \textit{near-optimality}, in the following particular sense: One may construct such a decomposition faster, but the overall complexity of the construction and usage will remain $\tilde{\Omega}(\log^2 n)$, since just using the decomposition (in the standard way) requires $c(n)\cdot d(n)=\tilde{\Omega}(\log^2 n)$ rounds in the worst case. 

\Cref{thm:NDmain} leads to $\tilde{O}(\log^2 n)$-round deterministic distributed algorithms for a wide range of other problems. In particular, $\tilde{O}(\log^2 n)$ is now the state-of-the-art complexity for all problems in the class $\mathsf{SLOCAL}(1)$, in the terminology of Ghaffari, Kuhn, and Maus~\cite{ghaffari2017complexity}, which roughly denotes problems solvable by sequential algorithms where each decision depends on its $O(1)$ neighborhood. Here are two example problems: (I) \textit{greedy coloring}\cite{gavoille2009complexity}: properly color vertices subject to the constraint that a node can have color $i$ only if colors $1$ to $i-1$ are present among its neighbors. (II) \textit{maximal set without forbidden subgraphs}: Given a constant size forbidden subgraph $H$, compute a maximal set $S$ of vertices of our graph $G=(V, E)$ such that the induced subgraph $G[S]$ does not have a copy of $H$. For both problems, this $\tilde{O}(\log^2 n)$ deterministic complexity almost matches the state-of-the-art randomized complexity $O(\log^2 n)$. The best current randomized algorithms are also based on network decomposition. This $\tilde{O}(\log^2 n)$ complexity is (nearly) the best that one can achieve via the generic network decomposition recipe. 
 
 \paragraph{General distributed derandomizations and their slow down.} Network decomposition is also central to general derandomizations of distributed algorithms. For any locally checkable problem (roughly meaning a problem where the solution can be verified in constant rounds)\cite{naor95}, any $r(n)$-round randomized algorithm can be transformed into a deterministic algorithm with $O(r(n)\cdot c(n)\cdot d(n) + r(n) \cdot t(n))$ round complexity\cite{ghaffari2017complexity,ghaffari2018derandomizing}, assuming there is a $t(n)$-round algorithm for computing a network decomposition with $c(n)$ colors and $d(n)$ diameter. This is roughly based on the general recipe of using network decompositions along with a method of conditional expectations. \Cref{thm:NDmain} gives the best known efficiency for such a general derandomization of distributed graph algorithms. It achieves a slow-down factor of $\tilde{O}(\log^2 n)$, which is nearly the best possible via (the standard usage of) network decomposition. It is also known that an $\Omega(\log n/\log\log n)$ slow down is necessary, in the worst case, for such a general derandomization: e.g., the sinkless orientation problem has tight randomized complexity $\Theta(\log\log n)$ \cite{chang16,ghaffari2017orinetation, brandt} and tight deterministic complexity $\Theta(\log n)$\cite{chang16,ghaffari2017orinetation}.
 
\subsection{Ruling Set}
Ruling Set is a generalization/relaxation of Maximal Independent Set, defined by Awerbuch et al.\cite{awerbuch89} to achieve similar functionalities while admitting more efficient computation. One canonical way of using MIS in distributed computation is as a collection of non-adjacent \textit{centers} such that each node of the graph has one of them nearby, namely in its direct neighborhood. Ruling sets are defined in this vein, as follows: A $\beta$-ruling set is an independent set $S\subseteq V$ such that for each $v\in V$, we have $dist(S, v)\leq \beta$. In particular, an MIS is a $1$-ruling set.\footnote{Some literature uses the more general definition of $(\alpha,\beta)$-ruling set as a set $S\subset V$ where every two nodes in $S$ have distance at least $\alpha$, and each node $v\in V$ has a node in $S$ within its $\beta$-hops. In this terminology, a (2,1)-ruling set is simply an MIS. One can require distance at least $\alpha$ in $G$ with independence in the graph power $G^{\alpha-1}$. Thus, up to small rounding imprecisions, it suffices for us to fix $\alpha=2$ and use the simpler notation $\beta$-ruling set.} Naturally, we are primarily interested in very small values of $\beta$.

\paragraph{Known results.} The work of Awerbuch et al.\cite{awerbuch89} showed a deterministic algorithm for computing an $O(\log n)$ ruling set in $O(\log n)$ rounds. Before our work, there was no known (polylogarithmic-time) deterministic algorithm for smaller $\beta$ values, except the trivial case of the MIS problem itself, which required higher polylogarithmic complexity, as we discuss in the next subsection. 
In contrast, there is a randomized algorithm that computes an $O(\log\log n)$ ruling set in $O(\log\log n)$ rounds\cite{SEW13}, and this is frequently used, e.g., in the shattering frameworks~\cite{barenboim_symmbreaking, gmis}. See \cite{gmis, balliu2022distributed} for tradeoff results, and a survey of other work focusing on graphs with small maximum degree $\Delta$.

\paragraph{Our contribution.} We present a deterministic algorithm that computes an $O(\log\log n)$ ruling set in $\tilde{O}(\log n)$ rounds. This is an exponentially better $\beta$ parameter compared to Awerbuch et al.\cite{awerbuch89}, in essentially the same complexity.

\begin{theorem}
    \label{thm:RSmain}
    There is a deterministic distributed algorithm that in any $n$-node graph $G=(V, E)$ computes an $O(\log\log n)$ ruling set in $\tilde{O}(\log n)$ rounds of the \LOCAL model.
\end{theorem}

The round complexity of this algorithm almost matches a lower bound of Balliu et al.\cite{balliu2022distributed}: They showed that deterministically computing an $O(\log\log n)$ ruling set requires $\Omega(\frac{\log n}{\log^2\log n})$ rounds and more generally, for any $\beta\leq \Theta(\sqrt{\frac{\log n}{\log\log n}})$, any deterministic $\beta$-ruling set algorithm requires at least $\Omega(\frac{\log n}{\beta \log\log n})$ rounds. 

\subsection{Maximal Independent Set}
\paragraph{Known results.} MIS is widely used in distributed algorithms and its complexity has been the subject of much research. The best known randomized algorithm remains at $O(\log n)$\cite{luby86}, and a lower bound of $\Omega(\sqrt{\log n/\log\log n})$ is known that applies also to randomized algorithms\cite{kuhn16_jacm}. For deterministic algorithms, the lower bound is much higher at $\Omega(\log n/\log\log n)$, by a celebrated work of Balliu et al.\cite{balliu2019LB}. As mentioned before, for around three decades, the best known deterministic complexity remained at $2^{O(\sqrt{\log n})}$~\cite{awerbuch89,panconesi-srinivasan}. That changed to $\poly(\log n)$ with the work of \cite{rozhonghaffari20}, and was improved gradually afterward\cite{GGR20,faour2022local,ghaffarigrunau2023fasterMIS}. The state-of-the-art round complexity is $\tilde{O}(\log^2 n)$\cite{ghaffarigrunau2023fasterMIS}.

\paragraph{Our contribuition.} We present an algorithm that computes MIS in $\tilde{O}(\log^{5/3} n)$ rounds. Furthermore, given known reductions\cite{linial1987LOCAL}, this $\tilde{O}(\log^{5/3} n)$ complexity applies to many other problems such as maximal matching and certain colorings:

\begin{theorem}
    \label{thm:MISmain}
    There is a deterministic distributed algorithm that in any $n$-node graph $G=(V, E)$ computes an MIS in $\tilde{O}(\log^{5/3} n)$ rounds of the \LOCAL model. This complexity also applies for maximal matching, $(deg+1)$-list vertex coloring, and $(2deg-1)$-list edge coloring\footnote{In the $(deg+1)$-list vertex coloring, each node $v$ having a prescribed list $L_{v}$ of colors, of size $|L_{v}|\geq deg(v)+1$ from which it should choose its color. In the $(2deg-1)$-list edge coloring, each edge $e=\{v, u\}$ has a prescribed list $L_{e}$ of colors, of size $|L_{e}|\geq deg(v)+deg(u)-1$ from which it should choose its color. Both problems reduce to MIS by a reduction of Luby~\cite{luby86, linial92}.}.
\end{theorem}

This $\tilde{O}(\log^{5/3} n)$ complexity goes polynomially below the $O(\log^2 n)$ barrier of the method via network decomposition. Moreover, via known results based on the shattering framework\cite{barenboim_symmbreaking,gmis, chang2018optimal}, the improvement also spreads to randomized algorithms, e.g., for coloring: 

\begin{corollary}[\Cref{thm:MISmain}+\cite{chang2018optimal}]
There is a randomized distributed algorithm, in the \local model, that computes a $\Delta+1$ vertex coloring in $\tilde{O}((\log\log n)^{5/3})$ rounds. The same holds also for the $2\Delta-1$ edge coloring problem.
\end{corollary}


\subsection{An overview of the new ingredients}
Our work builds on top of several recently developed techniques, e.g., the local rounding method for distributed derandomization of algorithms analyzed with pairwise independence~\cite{GhaffariK21,faour2022local}, derandomizations\cite{ghaffari2023netdecomp} of MPX-style network decomposition\cite{miller2013parallel}, and efficient computation of MIS by interweaving fast decomposition constructions with local rounding~\cite{ghaffarigrunau2023fasterMIS}. Our improvements involve three new major ingredients, along with smaller ideas throughout. We provide a brief overview of these main novelties here. 

The first ingredient, which we discuss in \Cref{subsubsec:sampling}, is a stronger deterministic ``\textit{sampling}" mechanism. It allows us to deterministically select a sparse subset of vertices, such that it has the ``\textit{right proportion}" of intersection with almost all neighborhoods. This sampling plays a key role in all our results, and in particular, it leads almost directly to our near-optimal ruling set algorithm (via smaller manipulations). Our network decomposition and maximal independent set results both involve two additional ingredients. 

The second ingredient, which we discuss in \Cref{subsubsec:newMPXderand}, is a novel method for derandomizing network decomposition constructions based on \textit{head start} values (in the style of the randomized algorithm of Miller, Peng, and Xu\cite{miller2013parallel}). This ingredient together with our stronger sampling result gives a novel method to build a network decomposition. However, with just these, the round complexity remains $\tilde{O}(\log^3 n)$. 

The third ingredient, which we discuss in \Cref{subsubsec:recursions}, is a careful recursive computation structure for network decomposition and maximal independent set. For instance, in network decomposition, this ingredient effectively constructs different colors of network decomposition simultaneously, sharing a lot of computation work among them. This recursion also leverages the computational structure provided by the new network decomposition derandomization of the second ingredient. 

\subsubsection{Stronger deterministic sampling, with simultaneous upper and lower bounds}
\label{subsubsec:sampling}
A key ingredient for our results is a stronger method for deterministically computing a sparse ``\textit{sample}" subset, which approximately mimics guarantees of randomized sampling with probability $p$. Consider for instance the setup of a bipartite graph $H=(U \sqcup V, E)$ with the promise that $\min_{u\in U} |\Gamma_H(u)| \geq \Omega(1/p)^{2}$. We would like to choose a subset $V'\subseteq V$ so that, for nearly all vertices $u\in U$, we have $\frac{|\Gamma_H(u) \cap V'|}{|\Gamma_H(u)|} \in  [p^{1.01},p]$. Our new deterministic sampling achieves this, in roughly $\tilde{O}(\log(1/p))$ rounds modulo $O(\log^* n)$ factors. The fraction, or more generally the importance-weighted fraction, of nodes $u$ that do not satisfy the condition is $\poly(p)$. See \Cref{chapter:sampling} for the precise statement and proofs. We comment that a similar sampling scheme was a key ingredient in the $\tilde{O}(\log^2 n)$-round MIS algorithm of Ghaffari and Grunau~\cite{ghaffarigrunau2023fasterMIS}, but with one major caveat: they could ensure a \textit{lower bound} for most $\frac{|\Gamma_H(u) \cap V'|}{|\Gamma_H(u)|}$ but not a simultaneous \textit{upper bound}. Instead, they worked with one global upper bound of $\frac{|V'|}{|V|}\leq p$. Achieving the lower and upper bounds simultaneously, for all but $\poly(p)$ fraction of neighborhoods, and especially in roughly $\tilde{O}(\log(1/p))$ rounds is a key source of strength in our sampling result. This is evidenced by how it almost directly leads to our near-optimal ruling set algorithm, which we present in \Cref{sec:rulingSet}. The sampling builds in part on the local rounding method for distributedly derandomizing pairwise analysis\cite{GhaffariK21,faour2022local}, along with a recursive pipelining/bootstrapping approach that approximately enforces the upper and lower bounds simultaneously.

\subsubsection{A novel derandomization approach to MPX-style network decomposition}
\label{subsubsec:newMPXderand}
The $\tilde{O}(\log^3 n)$ round deterministic network decomposition algorithm of \cite{ghaffari2023netdecomp} is based on derandomizing the randomized clustering algorithm of Miller, Peng, and Xu (MPX) \cite{miller2013parallel}. We first review the randomized MPX algorithm. Then, we sketch the derandomization of MPX given in \cite{ghaffari2023netdecomp}. Afterward, we sketch our new derandomization approach.

\paragraph{MPX Clustering.}
The randomized clustering of Miller, Peng, and Xu (MPX) computes a partition of the vertices into clusters of diameter $O(\log n)$ \cite{miller2013parallel}, such that in expectation at most half of nodes have neighbors in other clusters. By removing such nodes, one can also get non-adjacent low-diameter clustering for half of the nodes. The algorithm computes the clustering in two steps:

\begin{enumerate}
    \item \textbf{Random Head Starts:} Each vertex $v$ is assigned a random head start $h(v)$ drawn from a geometric distribution with parameter $1/2$.
    \item \textbf{Partition based on Head Starts:} Each node $u$ joins the cluster of the node $v$ that minimizes $d(u, v) - h(v)$, favoring nodes with larger identifier. This means nodes prefer to join clusters with closer centers and larger head starts.
\end{enumerate}
The resulting partition has strong-diameter at most $\max_{v \in V(G)} h(v)$, and therefore $O(\log n)$ with high probability. More generally, if one assigns each vertex $v$ a random head start $h(v)$ from a geometric distribution with parameter $p = n^{-1/k}$, then the resulting partition has diameter $O(k)$, with high probability. For $k \leq \log^{1-\eps}(n)$, the resulting partition has the property that each node neighbors at most $n^{O(1/k)}$ clusters, with high probability. Ghaffari and Grunau provided a deterministic algorithm that computes a partition that almost matches these guarantees: It has diameter $O(k)$ and each node neighbors at most $n^{\tilde{O}(1/k)}$ different clusters. Their algorithm computes the partition in $\tilde{O}(k^2 \log n)$ rounds \cite{ghaffarigrunau2023fasterMIS}. They used this partition, for $k = \sqrt{\log n}$, to derandomize each of the $O(\log n)$ rounds of Luby's randomized MIS algorithm in $\tilde{O}(\log n)$ rounds, leading to a deterministic MIS algorithm with round complexity $\tilde{O}(\log^2 n)$ rounds.
\paragraph{Deterministic head start computation of \cite{ghaffari2023netdecomp}}
The random head start $h(v)$ assigned to a given node can be determined by repeatedly flipping a coin until the first time it comes up tails; $h(v)$ is the total number of heads before the first tail. The algorithm of \cite{ghaffari2023netdecomp} derandomizes these coin flips one by one. It first determines for all nodes whether the first coin flip of each of them comes up heads or tails. Next, it determines the second coin flip of all nodes whose first coin was heads, and so on. Each of the $O(\log n)$ derandomization steps boils down, informally, to the following hitting set problem: Given a bipartite graph $H = (U \sqcup V,E)$, compute a subset $V' \subseteq V$ with $|V'|\leq |V|/2$ such that at most a $\frac{1}{100\log(n)}$-fraction of the nodes in $U$ are bad. Here we call a node $u\in U$ \textit{bad} if the following holds: $u$ has at least $\poly(\log \log n)$ neighbors but none of them is in $V'$. They gave an algorithm to compute such a subset in $\poly(\log \log n)$ rounds in $H$. The input graph $H$ to the hitting set instance is a virtual graph on top of the actual graph $G$, and simulating each round in $H$ takes $O(\log n)$ rounds. Therefore, each of the $O(\log n)$ derandomization steps takes $\tilde{O}(\log n)$ rounds, leading to a complexity of $\tilde{O}(\log^2 n)$ rounds.  
\paragraph{Our new deterministic head start computation}
We provide a new way to compute the head starts using our stronger deterministic sampling result, leading to an alternative MPX derandomization with the same round complexity. The maximum head start assigned to any node is $O(\log n)$, with high probability. Therefore, each head start can be represented in binary using $\approx \log\log(n)$ bits. Our derandomization fixes the bits in the binary representation one by one, starting with the most significant bit. This reduces the number of derandomization steps from $O(\log n)$ to $O(\log\log n)$. However, determining each bit in the binary representation boils down to a more complicated deterministic sampling problem. For example, determining the most significant bit roughly boils down to the following sampling problem: Given a bipartite graph $H = (U \sqcup V,E)$, compute a subset $V' \subseteq V$ such that at most a $\frac{1}{100\log\log(n)}$-fraction of the nodes in $U$ are bad, where we call a node $u\in U$ \textit{bad} if one the following holds: (1) $u$ has at least ${n^{0.5+\eps}}$ neighbors and none of them are included in $V'$, or (2) $u$ has ${n^{0.5+\eps}}$ neighbors in $V'$. Here, the reader can imagine $\eps$ as a tiny positive constant, though in our technical results, we will make it and use it as a subconstant. We can compute such a subset $V'$ in $\tilde{O}(\log n)$ rounds in $H$ using our new deterministic sampling result. The input graph $H$ is again a virtual graph on top of the actual graph $G$, and simulating each round in $H$ takes $O(\log n)$ rounds in $G$. Therefore, determining the most significant bit takes $\tilde{O}(\log^2 n)$ rounds. 

\subsubsection{Simultaneous (recursive) constructions sharing work}
The standard way to build a network decomposition is as follows: one clusters at least half of the yet uncolored nodes into non-adjacent clusters with diameter $d(n)$. Then, one assigns all of the clustered nodes the same new color. After $O(\log n)$ iterations, each node is assigned a color, and therefore one obtains a $(c(n),d(n))$-network decomposition with $c(n) = O(\log n)$. All of the previous $\poly(\log n)$-round deterministic network decomposition algorithms follow this template. For example, the MPX derandomization of \cite{ghaffari2023netdecomp} presented in the previous section leads to a deterministic algorithm that clusters at least half of the nodes into non-adjacent clusters of diameter $O(\log n)$ running in $\tilde{O}(\log^2 n)$ rounds. Plugging this algorithm into the template above therefore indeed leads to a network decomposition algorithm with round complexity $\tilde{O}(\log^3 n)$. 

Our new derandomization of MPX can similarly be used to obtain a new network decomposition algorithm with the same round complexity of $\tilde{O}(\log^3 n)$.  However, our new derandomization scheme opens the door to further improve this complexity to $\tilde{O}(\log^2 n)$. 

Recall that our derandomization fixes the $\approx \log \log n$ bits of the binary representation of the head starts one by one, starting with the most significant bit. Determining the most significant bit already takes $\tilde{O}(\log^2 n)$ rounds. However, each subsequent bit can be determined faster and faster, and the least significant bit can be determined in just $\tilde{O}(\log n)$ rounds. Assigning each node one of $c(n) = O(\log n)$ colors boils down to computing $c(n)$ different head starts for each node, one for each color. Let $v$ be a node and $h_1(v),\ldots,h_{c(n)}(v)$ such that $h_i(v)$ is the head start assigned to $v$ for determining which of the yet uncolored nodes receive color $i$. Instead of computing $h_i(v)$ from scratch, we manage to share and reuse much of computational work among these: in particular, in our approach, $h_i(v)$ is partially determined by the previous head starts assigned to $v$. For instance, the most significant bit of $h_i(v)$ actually coincides with the most significant bit of $h_1(v)$. This in particular leverages the property of our deterministic sampling mechanism that the fraction of \textit{bad nodes} in the samplings used for the earlier bits is considerably lower. In general, across the $c(n)$ iterations, we recompute the $j$-th most significant bit only $2^j$ times in total. Computing the $j$-th most significant bit takes $\tilde{O}(\log^2(n)/2^j)$ rounds, and therefore we spend in total $\tilde{O}(\log^2 n)$ rounds for the $2^j$ times we recompute the $j$-th most significant bit. As there are $O(\log\log n)$ bits, we spend $\tilde{O}(\log^2 n)$ rounds for computing the $O(\log n)$ different head start values for each node. This ultimately leads to a network decomposition with an improved round complexity of $\tilde{O}(\log^2 n)$. 

The above describes the process for our network decomposition algorithm. We use a similar scheme of sharing work, and in particular the recursive structure of the computations, in our algorithm for the maximal independent set problem.

\label{subsubsec:recursions}

\section{Preliminaries}
\paragraph{Notations:} Throughout, we work with the base graph $G=(V, E)$, with $n=|V|$ nodes, and we assume that a polynomial upper bound $N\in [n, \poly(n)]$ is known to all nodes. We use the notation $\Gamma_{G}(v)$ to denote the set of neighbors of $v$ in graph $G$, and $\Gamma^{+}_{G}(v) = \{v\} \cup \Gamma_{G}(v)$. For a set $S\subseteq V$, we define $\Gamma_{G}(S) = \bigcup_{s\in S} \Gamma_{G}(s)$. Sometimes, when it is clear from the context, we may omit the subscript $G$. Also, we usually reserve $G$ for the original base graph, and discuss other auxiliary graphs with other names, e.g., $H$. For a subset $S\subseteq V$, we denote by $G[S]$ the subgraph induced by $S$, i.e., the subgraph with vertex set $S$ and edge set equal to edges of $E$ that have both endpoints in $S$. We use the notation $E(G[S])$ to refer to the edge set of this induced subgraph.

\subsection{Clustering and head start functions}
We formally define clusterings and their various properties, along with head start functions in MPX-style decomposition (mentioned previously in \Cref{subsubsec:newMPXderand}) and how they give rise to clustering.

\begin{definition}[Clustering, Partition, Cluster Degree of a Vertex, and Cluster Diameter]
\label{def:clusterings}
A subset $C \subseteq V(G)$ is called a cluster. The (strong)-diameter of a cluster is defined as $\max_{u,v \in C} d_{G[C]}(u,v)$.
A clustering $\fC$ is a set of disjoint clusters. We refer to $\fC$ as a partition if $\bigcup_{C \in \fC} C = V(G)$. The diameter of a clustering $\fC$ is defined as the maximum diameter of all its clusters.
For a node $u \in V(G)$ and a clustering $\fC$, we define the degree of $u$ with respect to $\fC$ as $deg_\fC(u) = |\{C \in \fC \colon d(C,u) \leq 1\}|$. We sometimes refer to $\max_{u \in V(G)} \deg_\fC(u)$ as the degree of this clustering.
\end{definition}

\begin{restatable}[Head start function and its Badness $bad_{h,d}(u)$]{definition}{bad}
\label{def:bad}
Consider a head start function $h \colon V(G) \mapsto \mathbb{N}$, $d \in \mathbb{N}$ and $u \in V(G)$. For any $d' \in \mathbb{N}$, we define 

\[S_{d'}(u) = \{v \in V(G) \colon dist_G(u,v) = d'\}.\] 
That is, $S_{d'}(u)$ consists of all nodes in $G$ of distance exactly $d'$ to $u$. We define the related \textit{badness}
\fullOnly{
\[bad_{h,d}(u) = \max_{d' \in \{0,1,\ldots,d\}} \left| \left\{v \in S_{d'}(u) \colon h(v) = \max_{v' \in S_{d'}(u)} h(v') \right\}  \right| .\]}
\shortOnly{
\begin{align*}
    bad_{h,d}(u) &= \\ &\max_{d' \in \{0,1,\ldots,d\}} \left| \left\{v \in S_{d'}(u) \colon h(v) = \max_{v' \in S_{d'}(u)} h(v') \right\}  \right|.
\end{align*}
}
 \end{restatable}

The next lemma shows how a given head start function with small headstart values and small badness leads to a partition with small cluster diameters and small cluster degrees:
\begin{lemma}[Partition given Head Starts]
\label{lem:partition_given_head}
Let $D \in \mathbb{N}$ and consider a head start function $h \colon V(G) \mapsto \mathbb{N}_0$ with $\max_{v \in V(G)} h(v) \leq D$. There exists a distributed algorithm that computes in $O(D)$ rounds a partition $\fC$ with diameter at most $5D$ such that for any node $u \in V(G)$, the degree of $u$ with respect to $\fC$ is at most $10D \cdot bad_{h,10D}(u)$.
\end{lemma}
\begin{proof}[Proof Sketch]
Put each node $u$ in the cluster centered at node $v$ where $v$ minimizes $dist(u,v) - 5h(v)$. Ties in this measure are broken arbitrarily but consistently, e.g., always toward the node $v$ with the smallest identifier. 

Since $h(u)\geq 0$, we have $dist(u,u) - 5h(u) \leq 0$. Hence, for each node $u$, its distance to its cluster center $v$---which minimizes $dist(u,v) - 5h(v)$--- satisfies $dist(u,v) \leq 5h(v) \leq 5D$. Since ties are broken consistently, each cluster has strong-diameter at most $5D$. The reason is that the node $v'$ one step closer to $v$ on the shortest path from $u$ to $v$ must have also chosen $v$ as its cluster center. 

For each node $u$, we know that its degree $deg_\fC(u) = |\{C \in \fC \colon d(C,u) \leq 1\}|$ in this clustering is at most $10D \cdot bad_{h,10D}(u)$. The reason is that for each particular distance value $d' \in [0, 10D]$, there are at most $bad_{h,10D}(u)$ centers at distance exactly $d'$ from $u$ whose cluster can be within distance $1$ of $u$. The reason for this is that, since $h$ has been multiplied by $5$, only nodes $v' \in \left\{v \in S_{d'}(u) \colon h(v) = \max_{v' \in S_{d'}(u)} h(v') \right\}$ can be centers at distance exactly $d'$ from $u$ who are the centers of clusters within distance $1$ of $u$. Those that are not in this set and have a distance exactly $d'$ from $u$ will have $dist(u, v')-5h(v')$ at least $5$ units smaller than $d'-h(v)$, and thus they cannot be directly neighboring $u$ or including it (the center of the cluster of $u$ provides a better minimizer for the neighbors of $u$).
\end{proof}

\subsection{Local rounding}
\label{subsec:prelim-localRounding}
We use the local rounding framework of Faour et al.~\cite{faour2022local}. Their rounding framework works via computing a particular weighted defective coloring of the vertices, which allows the vertices of the same color to round their values simultaneously, with a limited loss in some objective functions that can be written as a summation of functions, each of which depends on only two nearby nodes. Next, we provide a related definition and then state their black-box local rounding lemma.

\begin{restatable}{definition}{DEFutilcost} 
\label{def:utility_cost}
[Pairwise Utility and Cost Functions]
Let $G = (V_G,E_G)$ be a graph. For any label assignment $\vec{x}: V_G \rightarrow \Labels$, a pairwise utility function is defined as $\sum_{u \in V_G} \utility(u, \vec{x}) + \sum_{e \in E_G} \utility(e, \vec{x})$, where for a vertex $u$, the function $\utility(u, \vec{x})$ is an arbitrary function that depends only on the label of $u$, and for each edge $e=\{u, v\}$, the function $\utility(e, \vec{x})$ is an arbitrary function that depends only on the labels of $v$ and $u$. These functions can be different for different vertices $u$ and also for different edges $e$. A pairwise cost function is defined similarly. 

For a probabilistic/fractional assignment of labels to vertices $V_G$, where vertex $v$ assumes each label in $\Sigma$ with a given probability, the utility and costs are defined as the expected values of the utility and cost functions, if we randomly draw integral labels for the vertices from their corresponding distributions (and independently, though of course each term in the summation depends only on the labels of two vertices and thus pairwise independence suffices).
\end{restatable}

\begin{restatable}{lemma}{FaouretalRounding}
  \label{lemma:rounding}
  [cf. Lemma 2.5 of \cite{faour2022local}] Let $G=(V_G,E_G)$ be a multigraph, which is equipped with utility and cost functions $\utility$ and $\cost$ and with a fractional label assignment $\lambda$ such that for every label $\alpha\in \Labels$ and every $v\in V_G$, $\lambda_\alpha(v)\geq \lambda_{\min}$ for some given value $\lambda_{\min}\in(0,1]$. Further assume that $G$ is equipped with a proper $\zeta$-vertex coloring. If $\utility(\lambda)-\cost(\lambda)\geq 0.1\utility(\lambda)$ and if each node knows the utility and cost functions of its incident edges, there is a deterministic $O\big( \log^2\big(\frac{1}{\lambda_{\min}}\big)+ \log\big(\frac{1}{\lambda_{\min}}\big)\cdot \log^* \zeta\big)$-round distributed message passing algorithm on $G$, in the \local model, that computes an integral label assignment $\ell$ for which
  \[
  \utility(\ell)-\cost(\ell) \geq 0.9 \big(\utility(\lambda) - \cost(\lambda)\big).
  \]
\end{restatable}

\section{Deterministic Sampling}
\label{chapter:sampling}
In this section, we prove our main deterministic sampling result, as stated below, which will be used later as a key ingredient in our ruling set, network decomposition, and MIS algorithms.  
\begin{restatable}{theorem}{samplingmain}
	\label{thm:sampling main}
	There exists an absolute constant $c$ such that the following holds.
	Let $H$ be a bipartite graph with a given bipartition $V(H) = U \sqcup V$, where each vertex in $V$ has a unique identifier from $\{1,2,\ldots,N\}$ for some $N \in \mathbb{N}$. Let $\Delta_U \in \mathbb{R}_{\geq 0}$ with $\Delta_U \geq \max(c,\max_{u \in U}|\Gamma_H(u)|)$. Let $imp \in \mathbb{R}^U_{\geq 0}$ and $IMP \in \mathbb{R}_{\geq 0}$ with $\sum_{u \in U} imp(u) \leq IMP$. 
	
	There exists a deterministic distributed \LOCAL algorithm that computes in $\tilde{O}(\log(\Delta_U)\log^*(n))$ rounds a subset \( V^{sub} \subseteq V \)  such that for $TR(\Delta_U) := \Delta_U^{0.5 + \frac{1}{\log^2\log\left(\Delta_U\right)}}$ and
\fullOnly{
\begin{align*}
U^{bad} &= \left\{u \in U : |\Gamma_H(u)| \geq TR(\Delta_U) \text{ and } \Gamma_H(u) \cap V^{sub} = \emptyset \right\} \cup \left\{u \in U : |\Gamma_H(u) \cap V^{sub}| \geq TR(\Delta_U)\right\},
 \end{align*}}
 \shortOnly{
 \begin{align*}
U^{bad} & &&= \\
& &&\left\{u \in U : |\Gamma_H(u)| \geq TR(\Delta_U) \, \&\, \Gamma_H(u) \cap V^{sub} = \emptyset \right\} \\ & &&\cup \left\{u \in U : |\Gamma_H(u) \cap V^{sub}| \geq TR(\Delta_U)\right\},
 \end{align*}
 }
it holds that $\sum_{u \in U^{bad}} imp(u) \leq \frac{1}{\Delta_U^{5}}IMP$.
\end{restatable}

\paragraph{Roadmap}
In \cref{sec:sampling_single}, we use the local rounding framework of \cite{GhaffariK21,faour2022local} to, intuitively speaking, derandomize pairwise independent sampling with probability $1-\eps$. In \cref{sec:sampling_multiple}, we derandomize sampling with probability $e^{-T}$ by repeatedly subsampling with probability $1-\eps$. The resulting algorithm has round complexity $\poly(T)\log^*(N)$. In \cref{sec:pipelining}, we prove \cref{thm:pipelining} using the derandomization of sampling with probability $e^{-T}$ along with a recursive pipelining/bootstrapping approach. \cref{thm:pipelining} is our most general sampling result, and \cref{thm:sampling main} follows as a relatively simple corollary. The formal proof of \cref{thm:sampling main} is given in \cref{sec:sampling_main_result}.

\subsection{Sampling with Probability $1-\eps$}
\label{sec:sampling_single}
 Our starting point is the following result that, intuitively speaking, derandomizes pairwise sampling with probability $1-\eps$. The result follows by a straightforward application of the local rounding framework developed in \cite{GhaffariK21,faour2022local}.
\begin{lemma}
	\label{lem:subsampling_single}
	Let $H$ be a bipartite graph with a given bipartition $V(H) = U \sqcup V$, where each vertex in $V$ has a unique identifier from $\{1,2,\ldots,N\}$ for some $N \in \mathbb{N}$.  Let $imp \in \mathbb{R}^U_{\geq 0}$ be a vector that assigns each node in $U$ an importance and $\eps \in (0,0.1]$.
	
	There exists a deterministic distributed \LOCAL algorithm that computes in $O(\log^*(N) + \poly(1/\eps))$ rounds a subset $V^{sub} \subseteq V$ such that for
	\fullOnly{
	\begin{align*}U^{bad} = \left\{u \in U \colon |\Gamma_H(u)| \geq (1/\eps)^{7} \text{ and } \frac{|\Gamma_H(u) \cap V^{sub}|}{|\Gamma_H(u)|} \notin [1-\eps - \eps^2, 1 - \eps + \eps^2]\right\}\end{align*}
    }
    \shortOnly{
    \begin{align*}
U^{bad} & = &&\{u \in U \colon |\Gamma_H(u)| \geq (1/\eps)^{7} \\ 
& &&\text{ and } \frac{|\Gamma_H(u) \cap V^{sub}|}{|\Gamma_H(u)|} \notin [1-\eps - \eps^2, 1 - \eps + \eps^2]\}
\end{align*} 
    }
	it holds that
	
	\[\sum_{u \in U^{bad}} imp(u) \leq \eps^2\sum_{u \in U} imp(u).\]

\end{lemma}
In fact, the guarantees are a bit weaker compared to what pairwise independent sampling would provide, in particular if the neighborhood size of each node $u$ is much larger than $\poly(1/\eps)$.

\paragraph{Proof Outline}
We first define a cost function $c$ that assigns each subset $V' \subseteq V$ a cost $c(V') = \sum_{u \in U} c_u(V')$. 
If one includes each node $v$ in $V'$ independently with probability $1-\eps$, then the expected cost is small. 
Namely, for any $u \in U$, the expected cost contribution $\E[c_u(V')]$ of $u$ is $(1-\eps) \cdot imp(u)$ and by linearity of expectation the expected total cost is $\E[c(V')] = (1-\eps)\sum imp(u)$. 
Using the rounding framework developed in \cite{GhaffariK21,faour2022local}, one can compute in $O(\log^*(N) + \poly(1/\eps))$ rounds a subset $V^{sub} \subseteq V$ whose cost is only a $(1 + \poly(\eps))$-factor larger compared to the expected cost incurred by random sampling. 
In particular, $c(V^{sub}) \leq (1 - \eps + \eps^6)\sum_{u \in U}imp(u)$.
Using simple but tedious calculations, one can show that $c_u(V^{sub}) \geq 1 - \eps - \eps^6 + \eps^3\1_{\{u \in U^{bad}\}}$ for every $u \in U$ with $|\Gamma_H(u)| \geq (1/\eps)^7$.
Assuming that $|\Gamma_H(u)| \geq (1/\eps)^7$ for every $u \in U$, these properties are sufficient to show that

\[\sum_{u \in U^{bad}}imp(u) \leq 2\eps^3\sum_{u \in U}imp(u) \leq \eps^2 \sum_{u \in U} imp(u).\]

Moreover, note that one can assume without loss of generality that $|\Gamma_H(u)| \geq (1/\eps)^7$ for every $u \in U$, which we will assume from now on.

\paragraph{Cost Functions}

	Fix some subset $V' \subseteq V$. For some $u \in U$, we define the cost that $u$ incurs by $V'$ as
	\fullOnly{
	\begin{align*}c_u(V') = \frac{imp(u)}{\binom{|\Gamma_H(u)|}{2}}\left(\binom{|V' \cap \Gamma_H(u)|}{2} + \frac{1-\eps}{\eps} \binom{|(V \setminus V') \cap \Gamma_H(u)| }{ 2} \right)\end{align*}
    }
    \shortOnly{
    \begin{align*}
c_u(V') &= \frac{imp(u)}{\binom{|\Gamma_H(u)|}{2}} \Bigg(\binom{|V' \cap \Gamma_H(u)|}{2} \\ 
&\quad + \frac{1-\eps}{\eps} \binom{|(V \setminus V') \cap \Gamma_H(u)| }{ 2} \Bigg)
\end{align*}
    }
	
	and the overall cost as 
	
	\[c(V') = \sum_{u \in U}c_u(V').\]

	\paragraph{Random Subsampling}
	Fix some $u \in U$ and assume we choose $V' \subseteq V$ randomly by including each vertex $v \in V$ independently with probability $1-\eps$. We next show that the expected cost that $u$ incurs by $V'$ is at exactly $(1-\eps)imp(u)$. First, note that we can further break down the cost into
	
	\[c_u(V') = \sum_{f \in \binom{{\Gamma_H(u)}}{ 2}} \left( \1_{\{f \subseteq V'\}} + \frac{1-\eps}{\eps}\1_{\{f \cap V' = \emptyset\}} \right)\frac{imp(u)}{\binom{|\Gamma_H(u)|}{2}}.\]
	
	Now, using that 
	
	\[\E \left[\1_{\{f \subseteq V'\}} + \frac{1-\eps}{\eps}\1_{\{f \cap V' = \emptyset\}}\right] = (1-\eps)^2 + \frac{1-\eps}{\eps}\eps^2 = 1-\eps,\]
	
	we directly get that $\E[c_u(V')] = (1-\eps)imp(u)$.

		\paragraph{Derandomization}
	Random sampling gives a set with small expected cost. We next explain how to \emph{deterministically} compute a subset $V^{sub} \subseteq V$ whose cost is at most the expected cost by random sampling.
    Concretely, a standard application of the local rounding framework in \cite{GhaffariK21,faour2022local} gives that, for a given $R$, one can compute in $O(\log^*(N) + R)$ rounds a subset $V^{sub} \subseteq V$ satisfying that 
    
    \[c(V^{sub}) \leq (1-\eps)\sum_{u \in U} imp(u) + \frac{1}{R}\frac{1-\eps}{\eps}\sum_{u \in U} imp(u).\]
    
    Note that term $(1-\eps)\sum_{u \in U} imp(u)$ is equal to the expected cost when sampling with probability $1-\eps$ and the term $\frac{1-\eps}{\eps}\sum_{u \in U} imp(u)$ is equal to $c(\emptyset)$, the largest cost incurred by any subset.
    We next give a high-level describtion of the approach. Let $H$ be the graph with vertex set $V(H) = V$ and where any two nodes in $V$ are connected by an edge if they have a common neighbor in $U$. Namely, $E(H) = \bigcup_{u \in U} \binom{\Gamma_H(u) }{ 2}$. Now, for any $e \in E(H)$, define
    
    \[w(e) = \sum_{u \in U \colon e \subseteq \Gamma_H(u)} \frac{imp(u)}{\binom{|\Gamma_H(u)|}{ 2}}.\]
    
    Note that $\sum_{e \in E(H)} w(e) = \sum_{u \in U} imp(u)$.
    There is algorithm that computes in $O(\log^*(N) + R)$ rounds a coloring with $O(R)$ colors such that the total weight of monochromatic edges is at most a $\frac{1}{R}$-fraction of the overall weight \cite{GhaffariK21}. Let $E(H) = E_{mono} \sqcup E_{bi}$ be the partition into monochromatic and bichromatic edges. We have
    
    \[\sum_{e \in E_{mono}} w(e) \leq \frac{1}{R}\sum_{e \in E(H)}w(e) = \frac{1}{R}\sum_{u \in U} imp(u).\]
     For a given subset $V' \subseteq V$, we can break down the cost into
    
    \begin{align*}
    	c(V') & && \\
        &= &&\sum_{e \in E(H)}  \left( \1_{\{e \subseteq V'\}} + \frac{1-\eps}{\eps}\1_{\{e \cap V' = \emptyset\}} \right) w(e) \\
    	&\leq &&\sum_{e \in E_{bi}}  \left( \1_{\{e \subseteq V'\}} + \frac{1-\eps}{\eps}\1_{\{e \cap V' = \emptyset\}} \right)w(e) \\ & &&+ \frac{1-\eps}{\eps}\sum_{e \in E_{mono}} w(e) \\
    	&\leq &&\sum_{e \in E_{bi}}  \left( \1_{\{e \subseteq V'\}} + \frac{1-\eps}{\eps}\1_{\{e \cap V' = \emptyset\}} \right)w(e) \\ & &&+ \frac{1}{R}\frac{1-\eps}{\eps}\sum_{u \in U} imp(u).
    \end{align*}

	Moreover, by sequentially iterating through the $O(R)$ colors and using the method of conditional expectation, one can compute in $O(R)$ rounds a subset $V^{sub} \subseteq V$ satisfying that 
	
	\begin{align*}
	\sum_{e \in E_{bi}}  \left( \1_{\{e \subseteq V^{sub}\}} + \frac{1-\eps}{\eps}\1_{\{e \cap V^{sub} = \emptyset\}} \right)w(e) \\
     \leq (1-\eps)\sum_{u \in U} imp(u).
     \end{align*}
	
	Note that any \local algorithm in $H$ can be simulated in $G$ with the same asymptotic round complexity. Thus, by setting $R = \lceil (1/\eps^7)\rceil$, we can conclude that there exists a deterministic algorithm that computes in $O(\log^*(N) + \poly(1/\eps))$ rounds a subset $V^{sub} \subseteq V$ with 
	
	\begin{align*}c(V^{sub}) \\ &\leq (1-\eps)\sum_{u \in U} imp(u) +  \frac{1}{\lceil (1/\eps^7)\rceil}\frac{1-\eps}{\eps}\sum_{u \in U} imp(u) \\
  &\leq (1-\eps + \eps^6)\sum_{u \in U} imp(u).\end{align*}

	\paragraph{Lower Bounding the Cost}
	
	We next derive a lower bound on the cost incurred by $V^{sub}$ in terms of $U^{bad}$. In particular, we show that
	
	\[ c(V^{sub}) \geq  (1 - \eps - \eps^6)\sum_{u \in U} imp(u) + \eps^3 \sum_{u \in U^{bad}} imp(u).\]
	
	In particular, we show that
	
	\[c_u(V^{sub}) \geq (1-\eps +\eps^3\1_{\{u \in U^{bad}\}} - \eps^6) imp(u)\] 
	
	for every $u \in U$ ($|\Gamma_H(u)| \geq (1/\eps)^7$). 
	
	 Fix some $u \in U$ and consider the function $f(\delta) = \delta^2 + \frac{1-\eps}{\eps}(1-\delta)^2$. A simple but tedious calculation, given below, gives a lower bound on $c_u(V^{sub})$ in terms of $f$. In particular,
	 
	 \[c_u(V^{sub}) \geq imp(u)\left(f\left(\frac{|\Gamma_H(u) \cap V^{sub}|}{|\Gamma_H(u)|} \right) - \eps^6\right).\]
	 
	  Using simple calculus, one can show that $f(\delta) \geq 1-\eps + \eps^3\1_{\{\delta \notin [1-\eps - \eps^2,1-\eps + \eps^2]\}}$ for every $\delta \in [0,1]$, which directly leads to the desired lower bound of	 
	  
	  \[c_u(V^{sub}) \geq (1-\eps +\eps^3\1_{\{u \in U^{bad}\}} - \eps^6) imp(u).\] 
	  
	  It remains to show that indeed  \[c_u(V^{sub}) \geq imp(u)\left(f\left(\frac{|\Gamma_H(u) \cap V^{sub}|}{|\Gamma_H(u)|} \right) - \eps^6\right)\] for every $u \in U$. First, note that
		
		\begin{align*}		    
		&\left(|\Gamma_H(u) \cap V^{sub}|\right)^2 + \frac{1-\eps}{\eps}\left(|\Gamma_H(u) \setminus V^{sub}|\right)^2 \\ = \; 
        & (|\Gamma_H(u)|)^2f\left(\frac{|\Gamma_H(u) \cap V^{sub}|}{|\Gamma_H(u)|} \right).\end{align*}
		
		Therefore,
		\fullOnly{
		\begin{align*}
			c_u(V^{sub}) =  
   & \sum_{e \in E(H) \cap \binom{V^{sub}}{ 2}} w_u(e)  + \sum_{e \in E(H) \cap \binom{V \setminus V^{sub} }{ 2}} \frac{1-\eps}{\eps}w_u(e) \\ 
			&= \frac{imp(u)}{\binom{|\Gamma_H(u)| }{ 2}} \left( \binom{|\Gamma_H(u) \cap V^{sub}|}{ 2} + \frac{1-\eps}{\eps} \binom{|\Gamma_H(u) \setminus V^{sub}| }{2}\right) \\
			&= \frac{1}{2}\frac{imp(u)}{\binom{|\Gamma_H(u)| }{ 2}} \left( (|\Gamma_H(u) \cap V^{sub}|)^2 - |\Gamma_H(u) \cap V^{sub}| + \frac{1-\eps}{\eps}\left( (|\Gamma_H(u) \setminus V^{sub}|)^2 - |\Gamma_H(u) \setminus V^{sub}|\right) \right) \\
			&\geq \frac{imp(u)}{|\Gamma_H(u)|^2} \left((|\Gamma_H(u) \cap V^{sub}|)^2 + \frac{1-\eps}{\eps} (|\Gamma_H(u) \setminus V^{sub}|)^2  - \left(|\Gamma_H(u) \cap V^{sub}| + \frac{1-\eps}{\eps}|\Gamma_H(u) \setminus V^{sub}|\right)\right)  \\			
			&\geq \frac{imp(u)}{|\Gamma_H(u)|^2}\left(|\Gamma_H(u)|^2f\left(\frac{|\Gamma_H(u) \cap V^{sub}|}{|\Gamma_H(u)|}\right) - \frac{1}{\eps}|\Gamma_H(u)|\right) \\
			&= imp(u)f\left(\frac{|\Gamma_H(u) \cap V^{sub}|}{|\Gamma_H(u)|} \right) - \frac{imp(u)}{\eps |\Gamma_H(u)| } \\
			&\geq imp(u)\left(f\left(\frac{|\Gamma_H(u) \cap V^{sub}|}{|\Gamma_H(u)|} \right) - \eps^6\right),
		\end{align*}
  }
 \shortOnly{ 
\begin{align*}
c_u(V^{sub}) 
&= \sum_{e \in E(H) \cap \binom{V^{sub}}{ 2}} w_u(e)  \\ 
&\quad + \sum_{e \in E(H) \cap \binom{V \setminus V^{sub} }{ 2}} \frac{1-\eps}{\eps}w_u(e) \\
&= \frac{imp(u)}{\binom{|\Gamma_H(u)| }{ 2}} \Bigg( \binom{|\Gamma_H(u) \cap V^{sub}|}{ 2}  \\ 
&\quad + \frac{1-\eps}{\eps} \binom{|\Gamma_H(u) \setminus V^{sub}| }{2}\Bigg) \\
&= \frac{imp(u)}{2\binom{|\Gamma_H(u)| }{ 2}} \Bigg( (|\Gamma_H(u) \cap V^{sub}|)^2  \\
&\quad - |\Gamma_H(u) \cap V^{sub}| \\
&\quad + \frac{1-\eps}{\eps}\Big( (|\Gamma_H(u) \setminus V^{sub}|)^2 \\ 
&\quad - |\Gamma_H(u) \setminus V^{sub}|\Big) \Bigg) \\
&\geq \frac{imp(u)}{|\Gamma_H(u)|^2} \Bigg((|\Gamma_H(u) \cap V^{sub}|)^2 \\
&\quad + \frac{1-\eps}{\eps} (|\Gamma_H(u) \setminus V^{sub}|)^2  \\
&\quad - \Big(|\Gamma_H(u) \cap V^{sub}| \\
&\quad + \frac{1-\eps}{\eps}|\Gamma_H(u) \setminus V^{sub}|\Big)\Bigg)  \\			
&\geq \frac{imp(u)}{|\Gamma_H(u)|^2}\Bigg(|\Gamma_H(u)|^2f\Big(\frac{|\Gamma_H(u) \cap V^{sub}|}{|\Gamma_H(u)|}\Big)  \\
&\quad - \frac{1}{\eps}|\Gamma_H(u)|\Bigg) \\
&= imp(u)f\Big(\frac{|\Gamma_H(u) \cap V^{sub}|}{|\Gamma_H(u)|} \Big) - \frac{imp(u)}{\eps |\Gamma_H(u)| } \\
&\geq imp(u)\Big(f\Big(\frac{|\Gamma_H(u) \cap V^{sub}|}{|\Gamma_H(u)|} \Big) - \eps^6\Big),
\end{align*}
}
  
		where the last inequality follows from our assumption $|\Gamma_H(u)| \geq (1/\eps)^7$.
	
	    \paragraph{Conclusion}
	    
	    The algorithm computes a subset $V^{sub} \subseteq V$ with cost at most
	    
	    \[c(V^{sub}) \leq (1-\eps + \eps^6)\sum_{u \in U} imp(u).\]
	    
	    On the other hand, we have shown that the cost of $V^{sub}$ is at least
	    
	    \[c(V^{sub}) \geq (1 - \eps - \eps^6)\sum_{u \in U} imp(u) + \eps^3 \sum_{u \in U^{bad}} imp(u).\]
	    
	    Combining both inequalities gives 
	    
	    \[\sum_{u \in U^{bad}} imp(u) \leq 2\eps^3 \sum_{u \in U} imp(u) < \eps^2 \sum_{u \in U} imp(u),\]
		
		where the last inequality follows from the input assumption $\eps \leq 0.1$.

\subsection{Repeated Subsampling}
\label{sec:sampling_multiple}
The result below is obtained by subsampling $T$ times with probability $1-\eps$, using the algorithm of \cref{lem:subsampling_single}.
  \begin{lemma}
	\label{lem:subsampling_multiple}
			Let $H$ be a bipartite graph with a given bipartition $V(H) = U \sqcup V$, where each vertex in $V$ has a unique identifier from $\{1,2,\ldots,N\}$ for some $N \in \mathbb{N}$. Let $imp \in \mathbb{R}^U_{\geq 0}$ be a vector that assigns each node in $U$ an importance and $\eps \in (0,0.1]$. Also, let $T \in \mathbb{N}_0$ denote the number of subsampling steps. \\ 
	Assuming that $|\Gamma_H(u)| \geq (1/\eps)^7 \cdot \left( \frac{1}{1-\eps-\eps^2}\right)^T$ for every $u \in U$, there exists a deterministic distributed \LOCAL algorithm that computes in $O\left(T \cdot \left(\log^*(N) + \poly(1/\eps)\right)\right)$ rounds a subset $V^{sub} \subseteq V$ and $bad \in \mathbb{N}^U_{0}$ satisfying
	
	\begin{enumerate}
		\item $\sum_{u \in U} imp(u) \cdot 2^{bad(u)} \leq (1+\eps^2)^{T} \sum_{u \in U} imp(u)$
		\item For every $u \in U$, it holds that $|\Gamma_H(u) \cap V^{sub}| \leq (1 - \eps + \eps^2)^{T - bad(u)}|\Gamma_H(u)|$
		\item  For every $u \in U$ with $bad(u) = 0$, it holds that $|\Gamma_H(u) \cap V^{sub}| \geq  (1 - \eps - \eps^2)^T|\Gamma_H(u)|$
	\end{enumerate}
	
\end{lemma}

  \paragraph{Proof Overview}
The algorithm invokes \cref{lem:subsampling_single} $T$ times, where each invocation further ``subsamples" the set $V$. For $u \in U$, $bad(u) \in \{0,1,\ldots,T\}$ is equal to the number of subsampling steps with $u \in U^{bad}$, i.e., where the ratio $\frac{|\Gamma_H(u) \cap V^{sub}_{after}|}{|\Gamma_H(u) \cap V^{sub}_{before}|}$ is further away from $1-\eps$ compared to what one would expect if one includes each node in $V^{sub}_{before}$ with probability $1-\eps$ in $V^{sub}_{after}$. Each time that $u \in U^{bad}$, the importance of $u$ increases by a $2$-factor in the next subsampling step. Also, note that we assume that $|\Gamma_H(u)| \geq (1/\eps)^7 \left(\frac{1}{1-\eps - \eps^2}\right)^T$. This ensures that even in the last subsampling step, the subsampled neighborhood of $v$ has size at least $(1/\eps)^7$, at least if $u$ was never in $U^{bad}$ across the different subsampling steps. We next give a formal describtion of the algorithm. We phrase the algorithm as a recursive algorithm that first invokes itself with $T_{rec} = T-1$ and then invokes \cref{lem:subsampling_single} to perform one more subsampling step.

\subsubsection{The Algorithm}

\paragraph{Base Case $T = 0$}
For the base case $T = 0$, the algorithm simply outputs $V^{sub} = V$ and $bad(u) = 0$ for every $u \in U$. It is easy to check that all three conditions are satisfied. Moreover, the algorithm does not require any communication round.

\paragraph{Inductive Case $T > 0$}

Next, consider the case $T > 0$. the algorithm first calls itself recursively with input $T^{rec} = T-1$, and the remaining inputs being the same, to obtain a subset $V^{sub}_{rec} \subseteq V$ and $bad_{rec} \in \mathbb{N}^U_0$. Afterwards, the algorithm computes a subset $V^{sub}_{L\ref{lem:subsampling_single}}\subseteq V^{sub}_{rec}$ and $U^{bad}_{L\ref{lem:subsampling_single}}  \subseteq U$ by running the algorithm of \cref{lem:subsampling_single} with the following input:

\begin{itemize}
	\item $H_{L\ref{lem:subsampling_single}} = H[U \sqcup V^{sub}_{rec}]$ 
	\item For every $u \in U$, $imp_{L\ref{lem:subsampling_single}}(u) = imp(u) \cdot 2^{bad_{rec}(u)}$
\end{itemize}

Let $V^{sub} = V^{sub}_{L\ref{lem:subsampling_single}}$ and

\begin{align*}U^{bad} &= \{u \in U \colon |\Gamma_H(u) \cap V^{sub}_{rec}| \geq (1/\eps)^7 \text{ and } \\ &\frac{|\Gamma_H(u) \cap V^{sub}|}{|\Gamma_H(u) \cap V^{sub}_{rec}|} \notin [1-\eps-\eps^2,1-\eps +\eps^2] \}.\end{align*}

The final output is then $V^{sub}$ and $bad(u) = bad_{rec}(u) + \1_{\{u \in U^{bad}\}}$ for every $u \in U$.

\subsubsection{The Analysis}

First, note that \cref{lem:subsampling_single} guarantees that 

\[\sum_{u \in U^{bad}} imp(u) \cdot 2^{bad_{rec}(u)} \leq \eps^2\sum_{u \in U} imp(u) \cdot 2^{bad_{rec}(u)}.\]

Moreover, the algorithm of \cref{lem:subsampling_single}  terminates in $O(\log^*(N) + \poly(1/\eps))$ output guarantees. Thus, a simple induction directly gives the desired round complexity of our algorithm. 

\begin{claim}
	The algorithm terminates in  $O\left(T \cdot \left(\log^*(N) + \poly(1/\eps)\right)\right)$ rounds.
\end{claim}
It therfore remains to prove that the three output guarantees are satisfied. We have already argued that this is the case for $T = 0$. Thus, it only remains to consider the inductive case $T > 0$, and where we can inductively assume that the output guarantess are satisfied for the recursive call with $T_{rec} = T-1$.

\begin{claim}
	It holds that $\sum_{u \in U} imp(u) \cdot 2^{bad(u)} \leq (1+\eps^2)^T \sum_{u \in U} imp(u)$.
\end{claim}
\begin{proof}
	Inductively, we can assume that
	
	\[\sum_{u \in U} imp(u) \cdot 2^{bad_{rec}(u)} \leq (1+\eps^2)^{T - 1} \sum_{u \in U} imp(u).\]
	
	Therefore,
	
	\begin{align*}
		& &&\sum_{u \in U} imp(u) \cdot 2^{bad(u)} \\ &= &&\sum_{u \in U} imp(u) \cdot 2^{bad_{rec}(u) + \1_{\{u \in U^{bad}\}}}  \\
		&= &&\sum_{u \in U} imp(u) \cdot 2^{bad_{rec}(u)}  + \sum_{u \in U^{bad}} imp(u) \cdot 2^{bad_{rec}(u)}  \\
		&\leq &&(1+\eps^2)\sum_{u \in U} imp(u) \cdot 2^{bad_{rec}(u)}\\
		&\leq &&(1+\eps^2)(1+\eps^2)^{T-1}\sum_{u \in U}imp(u) \\
		&= &&(1+\eps^2)^T \sum_{u \in U} imp(u).
	\end{align*}
\end{proof}
\begin{claim}
	For every $u \in U$, it holds that $|\Gamma_H(u) \cap V^{sub}| \leq (1 - \eps + \eps^2)^{T - bad(u)}|\Gamma_H(u)|$.
\end{claim}
\begin{proof}
	Fix some $u \in U$.
	Inductively, we can assume that
	
	\begin{align*}|\Gamma_H(u) \cap V^{sub}_{rec}| \leq (1 - \eps + \eps^2)^{(T-1) - bad_{rec}(u)}|\Gamma_H(u)|  \\
    = (1 - \eps + \eps^2)^{(T-1) - bad(u) + \1_{\{u \in U^{bad}\}}}|\Gamma_H(u)|.\end{align*}
	
	Thus, if $u \in U^{bad}$, we directly get $|\Gamma_H(u) \cap V^{sub}| \leq (1 - \eps + \eps^2)^{T - bad(u)}|\Gamma_H(u)|$. If $u \notin U^{bad}$, then $|\Gamma_H(u) \cap V^{sub}_{rec}| < (1/\eps)^7$ or $\frac{|\Gamma_H(u) \cap V^{sub}|}{|\Gamma_H(u) \cap V^{sub}_{rec}|} \in [1-\eps-\eps^2,1-\eps +\eps^2]$. If $|\Gamma_H(u) \cap V^{sub}_{rec}| < (1/\eps)^7$, our input assumption $|\Gamma_H(u)| \geq (1/\eps)^7 \cdot \left( \frac{1}{1-\eps-\eps^2}\right)^T$ directly gives
	
	\begin{align*}|\Gamma_H(u) \cap V^{sub}| &\leq |\Gamma_H(u) \cap V^{sub}_{rec}| < (1/\eps)^7   \\ &\leq (1 - \eps + \eps^2)^{T}|\Gamma_H(u)| \\ &\le (1 - \eps + \eps^2)^{T - bad(u)}|\Gamma_H(u)|.\end{align*}
	
	If $\frac{|\Gamma_H(u) \cap V^{sub}|}{|\Gamma_H(u) \cap V^{sub}_{rec}|} \in [1-\eps-\eps^2,1-\eps +\eps^2]$, we get
	
	\begin{align*}
		|\Gamma_H(u) \cap V^{sub}| \\ &\leq (1-\eps +\eps^2)|\Gamma_H(u) \cap V^{sub}_{rec}| \\ &\leq (1 - \eps + \eps^2)^{T - bad(u)}|\Gamma_H(u)|.
	\end{align*}
\end{proof}

\begin{claim}
	For every $u \in U$ with $bad(u) = 0$, it holds that $|\Gamma_H(u) \cap V^{sub}| \geq (1-\eps-\eps^2)^T|\Gamma_H(u)|$.
\end{claim}
\begin{proof}
	Fix some $u \in U$ with $bad(u) = 0$. In particular, $bad_{rec}(u) = 0$ and $u \notin U^{bad}$.
	As $bad_{rec}(u) = 0$, we can inductively assume that 
	
	\[|\Gamma_H(u) \cap V^{sub}_{rec}| \geq (1-\eps-\eps^2)^{T-1}|\Gamma_H(u)| \geq (1/\eps)^7.\]
	
	Moreover, $u \notin U^{bad}$ together with $|\Gamma_H(u) \cap V^{sub}_{rec}| \geq (1/\eps)^7$ implies  
	
	which together with our input assumption $|\Gamma_H(u)| \geq (1/\eps)^7 \cdot \left( \frac{1}{1-\eps-\eps^2}\right)^T$ directly gives $|\Gamma_H(u)| \geq (1/\eps)^7$. Furthermore, $u \notin U^{bad}_{L\ref{lem:subsampling_single}}$ together with $|\Gamma_H(u)| \geq (1/\eps)^7$ implies $\frac{|\Gamma_H(u) \cap V^{sub}|}{|\Gamma_H(u) \cap V^{sub}_{rec}|} \geq 1 - \eps - \eps^2$.
	Therefore,
	
	\begin{align*}
		|\Gamma_H(u) \cap V^{sub}| &\geq (1-\eps -\eps^2)|\Gamma_H(u) \cap V^{sub}_{rec}| \\ &\geq (1-\eps - \eps^2)(1 - \eps - \eps^2)^{(T-1)}|\Gamma_H(u)| \\ &\geq (1 - \eps - \eps^2)^{T - bad(u)}|\Gamma_H(u)|.
	\end{align*}

\end{proof}

\subsubsection{Corollary}
The following is a simple corollary of \cref{lem:subsampling_multiple}.
\begin{corollary}
	\label{cor:subsampling_multiple}
		Let $H$ be a bipartite graph with a given bipartition $V(H) = U \sqcup V$, where each vertex in $V$ has a unique identifier from $\{1,2,\ldots,N\}$ for some $N \in \mathbb{N}$.  Let $imp_{low},imp_{up} \in \mathbb{R}_{\geq 0}^U$ and $IMP_{low}$, $IMP_{up} \in \mathbb{R}_{\geq 0}$ with $\sum_{u \in U} imp_{low}(u) \leq IMP_{low}$ and $\sum_{u \in U} imp_{up}(u) \leq IMP_{up}$. Also, let $T \in \mathbb{N}$ with $T \geq 10$.
		
	Assuming that $|\Gamma_H(u)| \geq e^{100 T}$ for every $u \in U$, there exists a deterministic distributed \LOCAL algorithm that computes in $\poly(T)\log^*(N)$ rounds a subset $V^{sub} \subseteq V$ and $bad \in \mathbb{N}^U_{0}$ such that
	\begin{enumerate}
		\item $\sum_{u \in U} imp_{up}(u) \cdot 2^{bad(u)} \leq 10IMP_{up}$
		\item For every $u \in U$ it holds that $|\Gamma_H(u) \cap V^{sub}| \leq e^{- T + 1 + bad(u)/T^2}|\Gamma_H(u)|$.
		\item $\sum_{u \in U \colon |\Gamma_H(u) \cap V^{sub}| <e^{-T - 1}|\Gamma_H(u)|} imp_{low}(u)\leq 0.9 \cdot {IMP}_{low}$. 
	\end{enumerate}
\end{corollary}

\begin{proof}
	We use the algorithm of \cref{lem:subsampling_multiple} with the following input:
	
	\begin{itemize}
		\item For every $u \in U$, $imp_{L \ref{lem:subsampling_multiple}}(u) = \frac{imp_{low}(u)}{IMP_{low}} + 0.2\frac{imp_{up}(u)}{IMP_{up}}$.
		\item $\eps_{L \ref{lem:subsampling_multiple}} = \frac{1}{T^2}$
		\item $T_{L\ref{lem:subsampling_multiple}} = T^3$
	\end{itemize}

	Below, we show that the following inequalities are satisfied for $\eps_{L \ref{lem:subsampling_multiple}} \in [0,1/100]$ and $T\geq 10$:
	
	\begin{enumerate}
		\item $(1/\eps_{L \ref{lem:subsampling_multiple}})^7 \cdot \left( \frac{1}{1- \eps_{L \ref{lem:subsampling_multiple}} -\eps_{L \ref{lem:subsampling_multiple}}^2}\right)^{T_{L\ref{lem:subsampling_multiple}}} \leq e^{100T}$ 
		\item $(1+\eps_{L \ref{lem:subsampling_multiple}}^2)^{T_{L\ref{lem:subsampling_multiple}}} \leq 1.1$
		\item $(1 - \eps_{L \ref{lem:subsampling_multiple}} - \eps_{L \ref{lem:subsampling_multiple}}^2)^{T_{L\ref{lem:subsampling_multiple}}} \geq e^{-T-1}$
		\item $(1 - \eps_{L \ref{lem:subsampling_multiple}} + \eps_{L \ref{lem:subsampling_multiple}}^2)^{T_{L\ref{lem:subsampling_multiple}} - bad(u)} \leq e^{-T  + 1 + bad(u)/T^2}$
	\end{enumerate}

	 Note that our input requirement $\Gamma_H(u)| \geq e^{100T}$ for every $u \in U$ together with the first inequality gives $|\Gamma_H(u)| \geq (1/\eps_{L \ref{lem:subsampling_multiple}})^7 \cdot \left( \frac{1}{1- \eps_{L \ref{lem:subsampling_multiple}} -\eps_{L \ref{lem:subsampling_multiple}}^2}\right)^{T_{L\ref{lem:subsampling_multiple}}}$ for every $u \in U$. Therefore, the algorithm of \cref{lem:subsampling_multiple} outputs a subset $V^{sub} \subseteq V$ and $bad \in \mathbb{N}_0^U$ satisfying
	
	\begin{enumerate}
		\item $\sum_{u \in U} imp_{L \ref{lem:subsampling_multiple}}(u) \cdot 2^{bad(u)} \leq (1+\eps_{L \ref{lem:subsampling_multiple}}^2)^{T_{L\ref{lem:subsampling_multiple}}} \sum_{u \in U} imp_{L \ref{lem:subsampling_multiple}}(u)$
		\item For every $u \in U, |\Gamma_H(u) \cap V^{sub}| \leq (1 - \eps_{L \ref{lem:subsampling_multiple}} + \eps_{L \ref{lem:subsampling_multiple}}^2)^{T_{L\ref{lem:subsampling_multiple}} - bad(u)}|\Gamma_H(u)|$
		\item  For every $u \in U$ with $bad(u) = 0$, it holds that $|\Gamma_H(u) \cap V^{sub}| \geq  (1 - \eps_{L \ref{lem:subsampling_multiple}} - \eps_{L \ref{lem:subsampling_multiple}}^2)^{T_{L\ref{lem:subsampling_multiple}}}|\Gamma_H(u)|.$
	\end{enumerate}

	Note that 
	
	\begin{align*}
		& &&  \sum_{u \in U} imp_{L \ref{lem:subsampling_multiple}}(u) \\ &=  &&\left(\frac{\sum_{u \in U} imp_{low}(u)}{IMP_{low}} + 0.2\frac{\sum_{u \in U} imp_{up}(u)}{IMP_{up}}\right) \leq 1.2.
	\end{align*}

	Thus, combining the guarantees with the inequalities above gives
	
		\begin{enumerate}
		\item $\sum_{u \in U} imp_{L \ref{lem:subsampling_multiple}}(u) \cdot 2^{bad(u)} \leq 1.1 \cdot 1.2 \leq 1.5$,
		\item for every $u \in U, |\Gamma_H(u) \cap V^{sub}| \leq e^{-T  + 1 + bad(u)/T^2}|\Gamma_H(u)|$, and
		\item for every $u \in U$ with $bad(u) = 0$, it holds that \\ $|\Gamma_H(u) \cap V^{sub}| \geq  e^{-T-1}|\Gamma_H(u)|.$
	\end{enumerate}

	The first output guarantee is satisfied because
	
	\begin{align*}
		& &&\sum_{u \in U} imp_{up}(u) \cdot 2^{bad(u)} \\ &\leq &&(IMP_{up}/0.2)\sum_{u \in U} imp_{L \ref{lem:subsampling_multiple}}(u)  \cdot 2^{bad(u)} \\ &\leq &&(IMP_{up}/0.2) \cdot 1.5 \leq 10IMP_{up}.
	\end{align*}

	The second output guarantee is directly satisfied, it thus remains to verify that the third output guarantee is satisfied.
	By contraposing the third property, one gets that $|\Gamma_H(u) \cap V^{sub}| < e^{-T-1}|\Gamma_H(u)|$ implies $bad(u) \geq 1$. Therefore,
	
	\begin{align*}
		& \sum_{u \in U \colon |\Gamma_H(u) \cap V^{sub}| < e^{-T-1}|\Gamma_H(u)|} imp_{low}(u)  \\ &\leq \sum_{u \in U \colon bad(u) \geq 1} imp_{low}(u) \\ &\leq IMP_{low}\sum_{u \in U \colon bad(u) \geq 1} imp_{L\ref{lem:subsampling_multiple}}(u) \\
		&\leq IMP_{low} \cdot 0.5\sum_{u \in U} imp_{L\ref{lem:subsampling_multiple}}(u) \cdot 2^{bad(u)} \\
		&\leq IMP_{low} \cdot 0.5\cdot 1.5. \\
		&\leq 0.9 \cdot IMP_{low}.
	\end{align*}

	\paragraph{Inequalities}
		Recall that $\eps_{L \ref{lem:subsampling_multiple}} = \frac{1}{T^2}$, $T_{L\ref{lem:subsampling_multiple}} = T^3$ and $T \geq 10$. It remains to verify the following inequalities:
		\begin{enumerate}
		\item $(1/\eps_{L \ref{lem:subsampling_multiple}})^7 \cdot \left( \frac{1}{1- \eps_{L \ref{lem:subsampling_multiple}} -\eps_{L \ref{lem:subsampling_multiple}}^2}\right)^{T_{L\ref{lem:subsampling_multiple}}} \leq e^{100T}$ 
		\item $(1+\eps_{L \ref{lem:subsampling_multiple}}^2)^{T_{L\ref{lem:subsampling_multiple}}} \leq 1.1$
		\item $(1 - \eps_{L \ref{lem:subsampling_multiple}} - \eps_{L \ref{lem:subsampling_multiple}}^2)^{T_{L\ref{lem:subsampling_multiple}}} \geq e^{-T-1}$
		\item $(1 - \eps_{L \ref{lem:subsampling_multiple}} + \eps_{L \ref{lem:subsampling_multiple}}^2)^{T_{L\ref{lem:subsampling_multiple}} - bad(u)} \leq e^{-T  + 1 + bad(u)/T^2}$
	\end{enumerate}
	\paragraph{(1)}
	     It holds that $\frac{1}{1 - \delta - \delta^2} \leq 1 +2\delta$ for any $\delta \in [0,1/100]$. Therefore,
		 \begin{align*}
		(1/\eps_{L \ref{lem:subsampling_multiple}})^7 \cdot \left( \frac{1}{1- \eps_{L \ref{lem:subsampling_multiple}} -\eps_{L \ref{lem:subsampling_multiple}}^2}\right)^{T_{L\ref{lem:subsampling_multiple}}} \\ \leq T^{14} \cdot \left(1 + 2 \eps_{L \ref{lem:subsampling_multiple}}\right)^{T_{L\ref{lem:subsampling_multiple}}} \leq T^{14} \cdot e^{2 \eps_{L \ref{lem:subsampling_multiple}} \cdot T_{L\ref{lem:subsampling_multiple}}} \leq e^{100T}.
	\end{align*}
	\paragraph{(2)}
	\[(1+\eps_{L \ref{lem:subsampling_multiple}}^2)^{T_{L\ref{lem:subsampling_multiple}}} = \left(1 + \frac{1}{T^4}\right)^{T^3} \leq e^{1/T} \leq e\]
	\paragraph{(3)}
	
	 It holds that $1 - \delta - \delta^2 \geq e^{-\delta - 5\delta^2}$ for any $\delta \in [0,1/100]$. Therefore,
	 
	 \begin{align*}
	 	(1 - \eps_{L \ref{lem:subsampling_multiple}} - \eps_{L \ref{lem:subsampling_multiple}}^2)^{T_{L\ref{lem:subsampling_multiple}}}  = (1 - (1/T)^2 - (1/T)^4)^{T^3} \\ \geq e^{-(1/T^2 + 5/T^4)T^3} \geq  e^{-T - 1}.
	 \end{align*}
	
	\paragraph{(4)}
	
		\begin{align*}
		& (1 - \eps_{L \ref{lem:subsampling_multiple}} + \eps_{L \ref{lem:subsampling_multiple}}^2)^{T_{L\ref{lem:subsampling_multiple}} - bad(u)} \\ & \leq e^{-(\eps_{L \ref{lem:subsampling_multiple}} - \eps_{L \ref{lem:subsampling_multiple}}^2)(T_{L\ref{lem:subsampling_multiple}} - bad(u))} \\
		&= e^{- \left((1/T)^2 - (1/T)^4 \right)(T^3 - bad(u))} \\
		&= e^{-T + bad(u)/T^2 + 1/T - bad(u)/T^4} \\
		&\leq e^{-T + 1 + bad(u)/T^2}.
	\end{align*}

\end{proof}
\subsection{Pipelining}
\label{sec:pipelining}

Next, we derive our most general result in this section. In the statement below, $T_{outer}$ refers to the number of outer subsampling steps,  $T_{inner}$ to the number of inner subsampling steps, and $T_{rep}$ refers to the number of repetitions. For derandomizing sampling with probability $p$, one sets $T_{inner} = \lceil\log^{10}\log(1/p)\rceil$, $T_{outer} = \lceil \log(1/p)/T_{inner}\rceil$ and $T_{rep} = \lceil100\log(1/p)\rceil$. Our pipelining idea ensures that the number of subsampling steps and the number of repetitions don't multiply in the round complexity. Without this pipelining idea, we would get a round complexity of $T_{outer} \cdot T_{rep} = \Omega(\log^2(1/p)))$.

\begin{restatable}{theorem}{sampling}
\label{thm:pipelining}
	Let $H$ be a bipartite graph with a given bipartition $V(H) = U \sqcup V$, where each vertex in $V$ has a unique identifier from $\{1,2,\ldots,N\}$ for some $N \in \mathbb{N}$.  Let $imp_{low},imp_{up} \in \mathbb{R}_{\geq 0}^U$ and $IMP_{low}$, $IMP_{up} \in \mathbb{R}_{\geq 0}$ with $\sum_{u \in U} imp_{low}(u) \leq IMP_{low}$ and $\sum_{u \in U} imp_{up}(u) \leq IMP_{up}$. Let $ T_{inner} \in \mathbb{N}$ with $T_{inner}\geq 10 $ and $ T_{rep}, T_{outer} \in \mathbb{N}_0 $.
	
	There exists a deterministic distributed \LOCAL algorithm that computes in $ (T_{outer} + T_{rep}) \cdot \poly(T_{inner})\log^*(N) $ rounds a subset $ V^{sub} \subseteq V $ , $ bad \in \mathbb{N}^U_{ 0} $ and $U^{bad} \subseteq U$ such that

\fullOnly{
	\begin{enumerate}
		\item $ \sum_{u \in U} imp_{up}(u) \cdot 2^{bad(u)} \leq 11^{T_{rep} + T_{outer}} IMP_{up} $
        \item For every $u \in U$, 
    $|\Gamma_H(u) \cap V^{sub}| \leq (T_{rep})^{T_{outer}} \max\Big(e^{- (T_{inner} - 1) \cdot T_{outer}  + \frac{bad(u)}{(T_{inner})^2}}|\Gamma_H(u)|, e^{100T_{inner}}\Big).$
    \item $U^{bad} = \{u \in U \colon \frac{|\Gamma_H(u) \cap V^{sub}|}{|\Gamma_H(u)|} < e^{-(T_{inner} + 1) \cdot T_{outer}}  \text{ and } |\Gamma_H(u)| \geq e^{(T_{inner} + 1) \cdot T_{outer}} \cdot e^{100T_{inner}}\}$
       \item $ \sum_{u \in U^{bad}} imp_{low}(u) \leq (T_{rep} + 1)^{T_{outer}}0.9^{T_{rep}}IMP_{low} $
    \end{enumerate}
 }
 \shortOnly{
    \begin{enumerate}
    \item $\sum_{u \in U} imp_{up}(u) \cdot 2^{bad(u)} \leq 11^{T_{rep} + T_{outer}} IMP_{up} $
    \item For every $u \in U$,
          \begin{align*}
              &|\Gamma_H(u) \cap V^{sub}| \\ 
              &\leq (T_{rep})^{T_{outer}} \\
              &\quad \cdot \max\Big(e^{- (T_{inner} - 1) \cdot T_{outer}  + \frac{bad(u)}{(T_{inner})^2}}|\Gamma_H(u)|, \\
              &\qquad\qquad e^{100T_{inner}}\Big).
          \end{align*}
    \item \begin{align*}
          U^{bad} &= \{u \in U \colon \frac{|\Gamma_H(u) \cap V^{sub}|}{|\Gamma_H(u)|} < e^{-(T_{inner} + 1) \cdot T_{outer}} \\ 
          &\quad \text{and } |\Gamma_H(u)| \geq e^{(T_{inner} + 1) \cdot T_{outer}} \cdot e^{100T_{inner}}\}    
          \end{align*} 
    \item $\sum_{u \in U^{bad}} imp_{low}(u) \leq (T_{rep} + 1)^{T_{outer}}0.9^{T_{rep}}IMP_{low} $
\end{enumerate}
}
\end{restatable}

\subsubsection{The Algorithm}
\paragraph{Overview}
The algorithm is recursive. If $T_{outer} = 0$ or $T_{rep} = 0$, the algorithm directly outputs a solution without further communication. If $T_{outer} \geq 1$ and $T_{rep} \geq 1$, then the algorithm computes a subsampled set $V' \subseteq V$ using the algorithm of \cref{cor:subsampling_multiple}. Afterwards, the algorithm performs two recursive calls. One recursive call further subsamples $V'$ to obtain a subset $V'^{sub} \subseteq V'$, whereas the other recursive call subsamples anew to compute a subset $V^{''sub} \subseteq V$. The final output is then the union of the two subsamples, namely $V^{sub} = V'^{sub} \cup V''^{sub}$.  In the recursive call that further subsamples, the number of outer subsampling steps is decreased by one, i.e., $T'_{outer} = T_{outer} - 1$. In the recursive call that subsamples anew, the number of repetitions is reduced by one, i.e., $T''_{rep} = T_{rep} - 1$. Therefore, the algorithm only performs recursive calls where $T^{rec}_{outer} + T^{rec}_{rep} = T_{outer} + T_{rep} - 1$ (\cref{fig:pipelining} illustrates one execution of our algorithm where initially $(T_{outer},T_{rep}) = (2,2)$). Thus, we can prove that the output guarantees are satisfied by induction on $T_{outer} + T_{rep}$. For example, the first subsampling step to compute $V' \subseteq V$ takes $\poly(T_{inner})\log^*(N)$ rounds. Moreover, both recursive calls can be executed in parallel, and inductively we can assume that they terminate in $ (T_{outer} + T_{rep} - 1) \cdot \poly(T_{inner})\log^*(N) $ rounds. Thus, the algorithm terminates in $O(T_{outer} + T_{rep}) \cdot \poly(T_{inner})\log^*(N)$ rounds. Also, note that we can assume without loss of generality that $|\Gamma_H(u)| \geq e^{100T_{inner}}$ for every $u \in U$. We will assume from now on that this holds in order to simplify the description of the algorithm. 

\begin{figure}[htbp]
\centering
\includegraphics[width=0.8\linewidth]{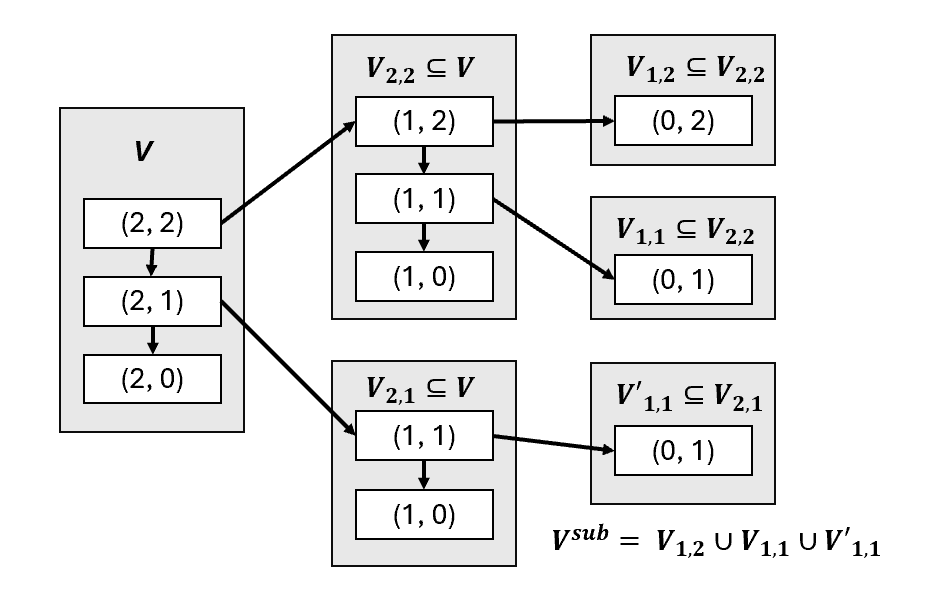}
\caption{The figure illustrates one execution of our algorithm where initially $(T_{outer},T_{rep}) = (2,2)$. Each arrow corresponds to one recursive call. All recursive calls in the same grey rectangle sample from the same ground set. This ground set is written on top of the grey rectangle.} 
\label{fig:pipelining}
\end{figure}

\paragraph{Base Case}
 We first discuss the two base cases. We start with the base case $T_{outer} = 0$ and $T_{rep} \geq 1$. In that case, the algorithm simply outputs $V^{sub} = V$, $bad(u) = 0$ for every $u \in U$ and $U^{bad} = \emptyset$. In particular, the algorithm finishes in $0$ communication rounds and it is easy to verify that the output satisfies the guarantees of \cref{thm:pipelining}.
The second base case is $T_{rep} = 0$. In that case, the algorithm outputs $V^{sub} = \emptyset$, $bad(u) = 0$ for every $u \in U$ and $U^{bad} = \emptyset$. As in the first base case, the algorithm finishes in $0$ communication rounds and it is easy to verify that the output satisfies the guarantees of \cref{thm:pipelining}.

\paragraph{First Subsample}
It remains to consider the case that $T_{outer} \geq 1 $ and $T_{rep} \geq 1$. In that case, the algorithm will first invoke the algorithm of \cref{cor:subsampling_multiple} with $T_{C\ref{cor:subsampling_multiple}} = T_{inner}$. Thus, the algorithm computes in $\poly(T_{inner})\log^*(N)$ rounds a subset $V' \subseteq V$ and $bad_{C\ref{cor:subsampling_multiple}} \in \mathbb{N}^U_{0}$ satisfying:

	\begin{enumerate}
	\item $\sum_{u \in U} imp_{up}(u) \cdot 2^{bad(u)} \leq 10IMP_{up}$
	\item For every $u \in U$, it holds that $|\Gamma_H(u) \cap V'| \leq e^{- T_{inner} + 1 + bad_{C\ref{cor:subsampling_multiple}}(u)/(T_{inner})^2}|\Gamma_H(u)|$.
	\item $\sum_{u \in U \colon |\Gamma_H(u) \cap V'| <e^{-T_{inner} - 1}|\Gamma_H(u)|} imp_{low}(u)\leq 0.9 \cdot {IMP}_{low}$
\end{enumerate}

 After this subsampling step, the algorithm performs two recursive calls. Importantly, the two recursive calls are independent of each other and can therefore be executed in parallel. 

\paragraph{Recursive call that further subsamples}
One recursive call computes a subset $V'^{sub} \subseteq V'$, $bad' \in \mathbb{N}_0^U$ and $U'^{bad} \subseteq U$ by performing a recursive call with $T'_{outer} = T_{outer} -1$. The concrete input of the recursive call is defined as follows:

\begin{enumerate}
	\item $G' = G[U \sqcup V']$
	\item $imp'_{up}(u) = imp_{up}(u) \cdot 2^{bad_{C\ref{cor:subsampling_multiple}}(u)}$ for every $u \in U$ and $IMP'_{up} = 10IMP_{up}$
	\item $imp'_{low}(u) = imp_{low}(u)$ for every $u \in U$ and $IMP'_{low} = IMP_{low}$
	\item $T'_{outer} = T_{outer} -1$ and $T'_{rep} = T_{rep}$
\end{enumerate} 

Note that it directly follows from the first output guarantee of \cref{cor:subsampling_multiple} that the precondition $\sum_{u \in U} imp'_{up}(u) \leq IMP'_{up}$ is satisfied.

\paragraph{Recursive call that subsamples anew}

The other recursive call computes a subset $V''^{sub} \subseteq V^{sub}_{C\ref{cor:subsampling_multiple}}$, $bad'' \in \mathbb{N}_0^U$ and $U''^{bad} \subseteq U$ by performing a recursive call with $T''_{rep} = T_{rep} -1$.
 The concrete input of the recursive call is defined as follows:

\begin{enumerate}
	\item $G''= G$
	\item $imp''_{up}(u) = imp_{up}(u)$ for every $u \in U$ and $IMP''_{up} = IMP_{up}$
	\item $imp''_{low}(u) = imp_{low}(u) \cdot \1_{\left\{\frac{|\Gamma_H(u) \cap V'|}{|\Gamma_H(u)|} < e^{-(T_{inner} + 1)} \right\}}$ for every $u \in U$ and $IMP''_{low} = 0.9 \cdot IMP_{low}$
	\item $T''_{outer} = T_{outer}$ and $T''_{rep} = T_{rep} - 1$
\end{enumerate} 

Note that it directly follows from the third output guarantee of \cref{cor:subsampling_multiple} that the precondition $\sum_{u \in U} imp''_{low}(u) \leq IMP''_{low}$ is satisfied.
\paragraph{Final output}
From the first subsampling step, we get $bad_{C\ref{cor:subsampling_multiple}} \in \mathbb{N}^U_{0}$.
From the two recursive calls, we obtain two subsets $V'^{sub} ,V''^{sub} \subseteq V$ and $bad',bad'' \in \mathbb{N}_0^U$. The final output $V^{sub} \subseteq V$, $bad \in \mathbb{N}_0^U$ and $U^{bad} \subseteq U$ is then defined as follows:

\begin{enumerate}
	\item $V^{sub} = V'^{sub} \cup V''^{sub}$
	\item $bad(u) = \max(bad_{C\ref{cor:subsampling_multiple}}(u) + bad'(u),  bad''(u))$ for every $u \in U$
	\item $U^{bad} = \{u \in U \colon \frac{|\Gamma_H(u) \cap V^{sub}|}{|\Gamma_H(u)|} < e^{-(T_{inner} + 1) \cdot T_{outer}} \text{ and } |\Gamma_H(u)| \geq e^{(T_{inner} + 1) \cdot T_{outer}} \cdot e^{100T_{inner}}\} $
\end{enumerate}

\subsubsection{Analysis}

\begin{claim}
	It holds that \[\sum_{u \in U} imp_{up}(u) \cdot 2^{bad(u)} \leq 11^{T_{rep} + T_{outer}}IMP_{up}.\]
\end{claim}
\begin{proof}
	We only have to consider the case $\min(T_{outer},T_{rep}) \geq 1$. We have 
	\fullOnly{
        \begin{align*}
		\sum_{u \in U} imp_{up}(u) \cdot 2^{bad(u)}  &= \sum_{u \in U} imp_{up}(u) \cdot 2^{\max(bad_{C\ref{cor:subsampling_multiple}}(u) + bad'(u),  bad''(u))} \\
		&\leq \sum_{u \in U} \left(imp_{up}(u) \cdot 2^{bad_{C\ref{cor:subsampling_multiple}}(u)}\right) \cdot  2^{bad'(u)} + \left(\sum_{u \in U} imp_{up}(u) \cdot 2^{bad''(u)} \right) \\
		&= \sum_{u \in U} imp_{up}'(u) \cdot  2^{bad'(u)} + \left(\sum_{u \in U} imp''_{up}(u) \cdot 2^{bad''(u)} \right).
	\end{align*}
    }
    \shortOnly{
            \begin{align*}
        \sum_{u \in U} & imp_{up}(u) \cdot 2^{bad(u)} \\
        &= \sum_{u \in U} imp_{up}(u) \cdot 2^{\max(bad_{C\ref{cor:subsampling_multiple}}(u) + bad'(u),  bad''(u))} \\
        &\leq \sum_{u \in U} \Big(imp_{up}(u) \cdot 2^{bad_{C\ref{cor:subsampling_multiple}}(u)}\Big) \cdot  2^{bad'(u)} \\
        &\quad + \Big(\sum_{u \in U} imp_{up}(u) \cdot 2^{bad''(u)} \Big) \\
        &= \sum_{u \in U} imp_{up}'(u) \cdot  2^{bad'(u)} \\
        &\quad + \Big(\sum_{u \in U} imp''_{up}(u) \cdot 2^{bad''(u)} \Big).
        \end{align*}
    }
	
	Inductively, we can assume that
	
	\begin{align*}& &&\sum_{u \in U} imp_{up}'(u) \cdot  2^{bad'(u)} \\ &\leq &&11^{T'_{rep} + T'_{outer}}IMP'_{up} = 11^{T_{rep} + (T_{outer} - 1)}10IMP_{up}\end{align*}
	
	and
	
	\begin{align*} & &&\sum_{u \in U} imp''_{up}(u) \cdot 2^{bad''(u)} \\ &\leq && 11^{T''_{rep} + T''_{outer}}IMP''_{up} = 11^{(T_{rep} - 1) + T_{outer}}IMP_{up}. \end{align*}
	
	Therefore, combining the inequalities above gives
	
	\begin{align*}& &&\sum_{u \in U} imp_{up}(u) \cdot 2^{bad(u)} \\ & \leq && 11^{T_{rep} + (T_{outer} - 1)}10IMP_{up} + 11^{(T_{rep} - 1) + T_{outer}}IMP_{up} \\ &\leq &&11^{T_{rep} + T_{outer}}IMP_{up}.\end{align*}

\end{proof}

\begin{claim}
	For any $u \in U$, $|\Gamma_H(u) \cap V^{sub}|$ is at most 
 \fullOnly{
 \begin{align*}
     (T_{rep})^{T_{outer}}\max(e^{- (T_{inner} - 1) \cdot T_{outer}  + bad(u)/(T_{inner})^2}|\Gamma_H(u)|, e^{100T_{inner}}).
 \end{align*}
 }
 \shortOnly{
 \begin{align*}
     (T_{rep})^{T_{outer}}\max(&e^{- (T_{inner} - 1) \cdot T_{outer}  + bad(u)/(T_{inner})^2}|\Gamma_H(u)|,\\  & e^{100T_{inner}}).
 \end{align*}
 }
\end{claim}
We only have to consider the case $\min(T_{outer},T_{rep}) \geq 1$. Fix some $u \in U$, and recall that $V^{sub} = V'^{sub} \cup V''^{sub}$.  We start by deriving an upper bound on $|\Gamma_H(u) \cap V'^{sub}|$. We have

\[|\Gamma_H(u) \cap V'| \leq e^{- T_{inner} + 1 + bad_{C\ref{cor:subsampling_multiple}}(u)/(T_{inner})^2}|\Gamma_H(u)|.\]

Moreover, by induction, we can upper bound $|N_{G'}(u) \cap V'^{sub}|$ by
\fullOnly{
\begin{align*}
	&(T'_{rep})^{T'_{outer}}\max(e^{- (T_{inner} - 1) \cdot T'_{outer} + bad'(u)/(T_{inner})^2}|N_{G'}(u)|, e^{100T_{inner}}) \\
	&= (T_{rep})^{T_{outer} - 1}\max(e^{- (T_{inner} - 1) \cdot (T_{outer} - 1) + bad'(u)/(T_{inner})^2}|N_{G}(u) \cap V'|, e^{100T_{inner}}).
\end{align*}
}
\shortOnly{
\begin{align*}
&(T'_{rep})^{T'_{outer}} \\
&\quad \cdot \max(e^{- (T_{inner} - 1) \cdot T'_{outer} + bad'(u)/(T_{inner})^2}|N_{G'}(u)|,  \\
&\qquad\qquad e^{100T_{inner}}) \\
&= (T_{rep})^{T_{outer} - 1} \\
&\quad \cdot \max(e^{- (T_{inner} - 1) \cdot (T_{outer} - 1) + bad'(u)/(T_{inner})^2}|N_{G}(u) \cap V'|, \\
&\qquad\qquad e^{100T_{inner}}).
\end{align*}
}
Note that

\[- T_{inner} + 1  - (T_{inner} - 1) \cdot (T_{outer} - 1)  = - (T_{inner} - 1) \cdot T_{outer}.\]

Therefore, combining the two inequalities, we can upper bound $|\Gamma_H(u) \cap V'^{sub}|$ by
\fullOnly{
\begin{align*}
	& (T_{rep})^{T_{outer} - 1}\max(e^{- (T_{inner} - 1) \cdot T_{outer} + bad'(u)/(T_{inner})^2 + bad_{C\ref{cor:subsampling_multiple}}(u)/(T_{inner})^2}|\Gamma_H(u)|, e^{100T_{inner}}) \\
	&= (T_{rep})^{T_{outer} - 1}\max(e^{- (T_{inner} - 1) \cdot T_{outer} + bad(u)/(T_{inner})^2}|\Gamma_H(u)|, e^{100T_{inner}}). 
\end{align*}
}
\shortOnly{
\begin{align*}
&(T_{rep})^{T_{outer} - 1} \\
&\quad \cdot \max(e^{- (T_{inner} - 1) \cdot T_{outer} + \frac{bad'(u)}{(T_{inner})^2} + \frac{bad_{C\ref{cor:subsampling_multiple}}(u)}{(T_{inner})^2}}|\Gamma_H(u)|,  \\
&\qquad\qquad e^{100T_{inner}}) \\
&= (T_{rep})^{T_{outer} - 1} \\
&\quad \cdot \max(e^{- (T_{inner} - 1) \cdot T_{outer} + bad(u)/(T_{inner})^2} |\Gamma_H(u)|,  \\
&\qquad\qquad e^{100T_{inner}}). 
\end{align*}
}

On the other hand, we can directly use induction to deduce that $|\Gamma_H(u) \cap V''^{sub}| = |N_{G''}(u) \cap V''^{sub}|$ is at most

\fullOnly{
\begin{align*}
	&(T''_{rep})^{T''_{outer}}\max(e^{- (T_{inner} - 1) \cdot T''_{outer} + bad''(u)/(T_{inner})^2}|\Gamma_H(u)|, e^{100T_{inner}}) \\
	&\leq (T_{rep} - 1)^{T_{outer}}\max(e^{- (T_{inner}-1) \cdot T_{outer} + bad(u)/(T_{inner})^2}|\Gamma_H(u)|, e^{100T_{inner}}).
\end{align*}
}
\shortOnly{
\begin{align*}
&(T''_{rep})^{T''_{outer}} \\
&\quad \cdot \max(e^{- (T_{inner} - 1) \cdot T''_{outer} + bad''(u)/(T_{inner})^2}|\Gamma_H(u)|, \\
&\qquad\qquad e^{100T_{inner}}) \\
&\leq (T_{rep} - 1)^{T_{outer}} \\
&\quad \cdot \max(e^{- (T_{inner}-1) \cdot T_{outer} + bad(u)/(T_{inner})^2}|\Gamma_H(u)|, \\
&\qquad\qquad e^{100T_{inner}}).
\end{align*}

}

We can therefore upper bound $|\Gamma_H(u) \cap V^{sub}|$ by

\fullOnly{
\begin{align*}
	&|\Gamma_H(u) \cap V'^{sub}| + |\Gamma_H(u) \cap V''^{sub}| \\ &\leq \left((T_{rep})^{T_{outer}- 1} + (T_{rep} - 1)^{T_{outer}} \right) \cdot \max(e^{- (T_{inner}-1) \cdot T_{outer} + bad(u)/(T_{inner})^2}|\Gamma_H(u)|, e^{100T_{inner}}) \\
	&\leq (T_{rep})^{T_{outer}}\max(e^{- (T_{inner}-1) \cdot T_{outer}  + bad(u)/(T_{inner})^2}|\Gamma_H(u)|, e^{100T_{inner}}).
\end{align*}
}
\shortOnly{
\begin{align*}
&|\Gamma_H(u) \cap V'^{sub}| + |\Gamma_H(u) \cap V''^{sub}| \\ 
&\leq \Big((T_{rep})^{T_{outer}- 1} + (T_{rep} - 1)^{T_{outer}} \Big) \\
&\quad \cdot \max(e^{- (T_{inner}-1) \cdot T_{outer} + bad(u)/(T_{inner})^2}|\Gamma_H(u)|, \\
&\qquad\qquad e^{100T_{inner}}) \\
&\leq (T_{rep})^{T_{outer}} \\
&\quad \cdot \max(e^{- (T_{inner}-1) \cdot T_{outer}  + bad(u)/(T_{inner})^2}|\Gamma_H(u)|, \\
&\qquad\qquad e^{100T_{inner}}).
\end{align*}
}

\begin{claim}
	\label{cl:pipelining_bad}
	If $\min(T_{outer},T_{rep}) \geq 1$, then we would have $U^{bad} \subseteq U'^{bad} \cup \left\{u \in U''^{bad} \colon \frac{|\Gamma_H(u) \cap V'|}{|\Gamma_H(u)|} < e^{-(T_{inner} + 1)}\right\}$.
\end{claim}
\begin{proof}
	Consider some arbitrary $u \in U^{bad}$. In particular,
	
	\begin{enumerate}
		\item $\frac{|\Gamma_H(u) \cap V^{sub}|}{|\Gamma_H(u)|} < e^{-(T_{inner} + 1) \cdot T_{outer}}$
		\item $|\Gamma_H(u)| \geq e^{(T_{inner} + 1) \cdot T_{outer}} \cdot e^{100T_{inner}}$
	\end{enumerate}
	
	Note that it suffices to show the following:
	
	\begin{enumerate}
		\item $u \in U''^{bad}$ 
		\item If $|\Gamma_H(u) \cap V'| \geq  e^{-(T_{inner} + 1)}|\Gamma_H(u)|$, then $u \in U'^{bad}$.
	\end{enumerate}
	
	The first property follows essentially by definition, whereas the second property requires a small calculation.

	\paragraph{(1)}
	To show that $u \in U''^{bad}$, we have to show that $\frac{|N_{G''}(u) \cap V''^{sub}|}{|N_{G''}(u)|} < e^{-(T_{inner} + 1) \cdot T''_{outer}}$ and $|N_{G''}(u)| \geq e^{(T_{inner} + 1) \cdot T''_{outer}} \cdot e^{100T_{inner}}$ or equivalently that
	
	\begin{enumerate}
		\item $\frac{|N_{G}(u) \cap V''^{sub}|}{|N_{G}(u)|} < e^{-(T_{inner} + 1) \cdot T_{outer}}$
		\item $ |N_{G}(u)| \geq e^{(T_{inner} + 1) \cdot T_{outer}} \cdot e^{100T_{inner}}$.
	\end{enumerate}
	Both properties directly follow from our assumption $u \in U^{bad}$.

	\paragraph{(2)}
	Assume that $|\Gamma_H(u) \cap V'| \geq e^{-(T_{inner} + 1)}|\Gamma_H(u)|$. To show that $u \in U'^{bad}$, we have to show that $\frac{|N_{G'}(u) \cap V'^{sub}|}{|N_{G'}(u)|} < e^{-(T_{inner} + 1) \cdot T'_{outer}}$ and $|N_{G'}(u)| \geq e^{(T_{inner} + 1) \cdot T'_{outer}} \cdot e^{100T_{inner}}$ or equivalently that
	
	\begin{enumerate}
		\item $\frac{|N_{G}(u) \cap V'^{sub}|}{|N_{G}(u) \cap V'|} < e^{-(T_{inner} + 1) \cdot (T_{outer} - 1)}$
		\item $ |N_{G}(u) \cap V'| \geq e^{(T_{inner} + 1) \cdot (T_{outer} - 1)} \cdot e^{100T_{inner}}$
	\end{enumerate}
	
	The first property follows because
	
	\begin{align*}
		\frac{|N_{G}(u) \cap V'^{sub}|}{|N_{G}(u) \cap V'|} &\leq \frac{|\Gamma_H(u) \cap V^{sub}|}{|\Gamma_H(u)|} \cdot \frac{|\Gamma_H(u)|}{|\Gamma_H(u) \cap V'|}  \\
        &< e^{-(T_{inner} + 1) \cdot T_{outer}} \cdot e^{(T_{inner} + 1)} \\ &= e^{-(T_{inner} + 1) \cdot (T_{outer} - 1)}.
	\end{align*}
	
	The second property follows because
	
	\begin{align*}
		|N_{G}(u) \cap V'| &\geq e^{-(T_{inner} + 1)}|\Gamma_H(u)| \\ &\geq e^{-(T_{inner} + 1)} e^{(T_{inner} + 1) \cdot T_{outer}} \cdot e^{100T_{inner}} \\ &= e^{-(T_{inner} + 1) \cdot (T_{outer} - 1)} \cdot e^{100T_{inner}}.
	\end{align*}
	
	Therefore, $u \in U'^{bad}$.
\end{proof}

\begin{claim}
	$\sum_{u \in U^{bad}} imp_{low}(u) \leq (T_{rep} + 1)^{T_{outer}} \cdot 0.9^{T_{rep}}IMP_{low}$.
\end{claim}

\begin{proof}
	We only have to consider the case $\min(T_{outer},T_{rep}) \geq 1$. Inductively, we can assume that
	
	\begin{align*}\sum_{u \in U'^{bad}} imp_{low}(u) &= \sum_{u \in U'^{bad}} imp'_{low}(u) \\ &\leq (T'_{rep} + 1)^{T'_{outer}}0.9^{T'_{rep}}IMP'_{low} \\ &=  (T_{rep} + 1)^{T_{outer} - 1}0.9^{T_{rep}}IMP_{low}\end{align*}
	
	and
	
	\begin{align*}
		\sum_{u \in U''^{bad}} imp''_{low}(u) &\leq (T''_{rep} + 1)^{T''_{outer}}0.9^{T''_{rep}}IMP''_{low} \\ &= (T_{rep})^{T_{outer}}0.9^{T_{rep}}IMP_{low}.
	\end{align*}
	
	\cref{cl:pipelining_bad} gives
	
	\begin{align*}& \sum_{u \in U^{bad}} imp_{low}(u) \leq \\ &\sum_{u \in U'^{bad}} imp_{low}(u) + \sum_{\substack{u \in U''^{bad} \colon\\ \frac{|\Gamma_H(u) \cap V'|}{|\Gamma_H(u)|} < e^{- (T_{inner} + 1) }}} imp_{low}(u).\end{align*}
	
	Moreover, using that $imp''_{low}(u) = imp_{low}(u)$ for every $u \in U$ with $|\Gamma_H(u)| < e^{- (T_{inner} + 1) }$, we get
	
	\begin{align*}
	\sum_{u \in U^{bad}} imp_{low}(u) &\leq && \sum_{u \in U'^{bad}} imp_{low}(u) \\ & && + \sum_{\substack{u \in U''^{bad} \colon \\ \frac{|\Gamma_H(u) \cap V'|}{|\Gamma_H(u)|} < e^{- (T_{inner} + 1) }}} imp_{low}(u) \\
	& \leq && \sum_{u \in U'^{bad}} imp_{low}(u) \\ & &&+ \sum_{u \in U''^{bad}} imp''_{low}(u) \\
    & \leq &&(T_{rep} + 1)^{T_{outer}} \cdot 0.9^{T_{rep}}IMP_{low}.
\end{align*}
\end{proof}

\subsection{Proof of Our Main Deterministic Sampling Reuslt}
\label{sec:sampling_main_result}
We are now ready to prove \cref{thm:sampling main}, restated below. 
\samplingmain*

It follows as a simple corollary of \cref{thm:pipelining}, restated below.
\sampling*
\begin{proof}[Proof of \cref{thm:sampling main}]
We invoke \cref{thm:pipelining} with input $T_{inner} = \lceil\log^{10}\log(\Delta_U)\rceil$, $T_{outer} = \lceil\frac{1}{2}\log(\Delta_U) /T_{inner}\rceil$ and $T_{rep} = \lceil10\log(\Delta_U)\rceil$. As a result, we obtain $V^{sub} \subseteq V$, $bad \in \mathbb{N}_0^U$ and $U_{T\ref{thm:pipelining}}^{bad} \subseteq U$ satisfying
	\begin{enumerate}
	\item $ \sum_{u \in U} imp(u) \cdot 2^{bad(u)} \leq 11^{T_{rep} + T_{outer}} IMP $
        \item For every $u \in U$, 
    \begin{align*}
    &|\Gamma_H(u) \cap V^{sub}| \leq (T_{rep})^{T_{outer}} \cdot \\ 
    &\max\Big( e^{- (T_{inner} - 1) \cdot T_{outer}  + \frac{bad(u)}{(T_{inner})^2}}|\Gamma_H(u)|, e^{100T_{inner}}\Big).
    \end{align*}
    \item $U_{T\ref{thm:pipelining}}^{bad} = \{u \in U \colon \frac{|\Gamma_H(u) \cap V^{sub}|}{|\Gamma_H(u)|} < e^{-(T_{inner} + 1) \cdot T_{outer}}  \text{ and } |\Gamma_H(u)| \geq e^{(T_{inner} + 1) \cdot T_{outer}} \cdot e^{100T_{inner}}\}$
       \item $ \sum_{u \in U^{bad}_{T\ref{thm:pipelining}}} imp(u) \leq (T_{rep} + 1)^{T_{outer}}0.9^{T_{rep}}IMP$
    \end{enumerate}
    We next show that $V^{sub}$ satisfies the output guarantees of \cref{thm:sampling main}. 
    We first show that $\sum_{u \in U \colon bad(u) \geq 10^9\log(\Delta_U)} imp(u) \leq \frac{1}{2\Delta_U^5}IMP$. This follows as

    \begin{align*}
        & &&\sum_{\substack{u \in U \colon \\ bad(u) \geq 10^{9}\log(\Delta_U)}} imp(u) \\ &\leq &&\frac{1}{2^{10^9\log(\Delta_U)}}
        \sum_{\substack{u \in U \colon \\ bad(u) \geq 10^9\log(\Delta_U)}} imp(u) \cdot 2^{bad(u)} 
        \\ &\leq &&\frac{11^{T_{rep} + T_{outer}}}{2^{10^9\log(\Delta_U)}}IMP \\
        &\leq &&\frac{1}{2\Delta_U^5}IMP.
    \end{align*}
    Next, consider any $u \in U$ with $bad(u) < 10^{9}\log(\Delta_U)$. We next show that $|\Gamma_H(u) \cap V^{sub}| \leq TR(\Delta_U)$. In particular,
    \fullOnly{
        \begin{align*}
    &|\Gamma_H(u) \cap V^{sub}| \\
    &\leq (T_{rep})^{T_{outer}} \max\Big(e^{- (T_{inner} - 1) \cdot T_{outer}  + \frac{bad(u)}{(T_{inner})^2}}|\Gamma_H(u)|, e^{100T_{inner}}\Big) \\
    &\leq (\lceil10\log(\Delta_U)\rceil)^{\lceil\frac{1}{2}\log(\Delta_U) /\lceil\log^{10}\log(\Delta_U)\rceil\rceil} \cdot e^{- (T_{inner} - 1) \cdot \lceil\frac{1}{2}\log(\Delta_U) /T_{inner}\rceil  + \frac{10^{9}\log(\Delta_U)}{(\lceil\log^{10}\log(\Delta_U)\rceil)^2}}\Delta_U \\
    &\leq \Delta_U^{0.5 + \frac{1}{\log^2\log(\Delta_U)}}.
    \end{align*}
    }
    \shortOnly{
    \begin{align*}
&|\Gamma_H(u) \cap V^{sub}| \\
&\leq (T_{rep})^{T_{outer}}  \\
&\quad \cdot \max\Big(e^{- (T_{inner} - 1) \cdot T_{outer}  + \frac{bad(u)}{(T_{inner})^2}}|\Gamma_H(u)|, \\
&\qquad\qquad e^{100T_{inner}}\Big) \\
&\leq (\lceil10\log(\Delta_U)\rceil)^{\lceil\frac{1}{2}\log(\Delta_U) /\lceil\log^{10}\log(\Delta_U)\rceil\rceil} \\
&\quad \cdot e^{- (T_{inner} - 1) \cdot \lceil\frac{1}{2}\log(\Delta_U) /T_{inner}\rceil  + \frac{10^{9}\log(\Delta_U)}{(\lceil\log^{10}\log(\Delta_U)\rceil)^2}}\Delta_U \\
&\leq \Delta_U^{0.5 + \frac{1}{\log^2\log(\Delta_U)}}.
\end{align*}
    }

    A similar calculation shows that if for a given $u \in U$, $|\Gamma_H(u)| \geq TR(\Delta_U)$ and $\Gamma_H(u) \cap V^{sub} = \emptyset$, then $u \in U^{bad}_{T\ref{thm:pipelining}}$. In particular,

    \[U^{bad} \subseteq U^{bad}_{T\ref{thm:pipelining}} \cup \{u \in U \colon bad(u) \geq 10^{9}\log(\Delta_U)\}.\]

    Therefore,

    \begin{align*}
        &\sum_{u \in U^{bad}} imp(u) \\
        &\leq \sum_{u \in U^{bad}_{T\ref{thm:pipelining}}} imp(u) +   \sum_{u \in U \colon bad(u) \geq 10^{9}\log(\Delta_U)} imp(u) \\
        &\leq (T_{rep} + 1)^{T_{outer}}0.9^{T_{rep}}IMP + \frac{1}{2\Delta^5_U}IMP  \\
        &\leq \frac{1}{\Delta^5_U}IMP.
    \end{align*}
    
\end{proof}

\section{Ruling Set}
\label{sec:rulingSet}

This section presents our ruling set result, as formally stated below (this is a more detailed variant of the statement in \Cref{thm:RSmain}).

\begin{theorem}
\label{thm:ruling_set_main}
There exists a deterministic distributed algorithm that computes in $\tilde{O}(\log n)$ rounds an \emph{independent} set $U_{RS} \subseteq V(G)$ such that for any node $v \in V(G)$, $dist(v,U_{RS}) = O(\log \log \Delta)$, where $\Delta$ is an upper bound on the maximum degree, assumed to be known to all nodes.
\end{theorem}

\noindent Comment: We do not need $\Delta$ to be a sharp upper bound on the maximum degree. The asymptotic quality and round complexity would remain the same if the upper bound $\Delta$ on the actual maximum degree $\Delta'$ satisfies $\Delta\in [\Delta',(\Delta')^{\poly (\log \Delta')}]$. Moreover, this assumption could be removed at no cost via the standard technique of trying guesses $\Delta=2^{2^{i}}$ for $i=1, 2, 3, \dots$, and removing each time all nodes within $O(\log\log \Delta)$ distance of the computed set.  

To obtain \Cref{thm:ruling_set_main}, below, we first discuss in \Cref{cor:ruling_set_sampling} a helper tool, which is a fast deterministic almost dominating set computation with polynomially reduced induced degree, and follows as a simple corollary of our main sampling result presented in \Cref{thm:sampling main}. Then we prove \Cref{thm:ruling_set_main} using this tool.

Before formally stating the corollary, let us intuitively describe the provided functionality: given a set $W$ of vertices with maximum induced degree $\Delta_W$, the corollary gives a fast algorithm for computing a subset $W'\subset W$ with a \textit{polynomially smaller} maximum induced degree of $\Delta^{0.9}_{w}$, such that $W'$ dominates \textit{almost all} of $W$. Concretely, at most $1/\Delta_W$ fraction of nodes of $W$ do not have neighbors in $W'$. The complexity that is roughly $\tilde{O}(\log \Delta_{W})$ is the key strength of this result, which opens the road for our final ruling set result.
\begin{corollary}[Sampling for Ruling Set]
\label{cor:ruling_set_sampling}
There exists a constant $c$ such that the following holds: Let $W \subseteq V(G)$ and $\Delta_W \geq c$ be an upper bound on the maximum degree of the induced subgraph $G[W]$, where $\Delta_W$ is known to all nodes. Define the set of high degree nodes $W_{high} := \{w \in W \colon |\Gamma_{G}(w)\cap W| \geq \Delta^{0.9}_W\}$, i.e., nodes in $W$ that have at least $\Delta^{0.9}_W$ neighbors in it. Let $N_{high}$ be an upper bound on the size of $W_{high}$, known to all nodes.

There exists a deterministic distributed \local algorithm with round complexity $\tilde{O}(\log (\Delta_W)\log^*(n))$ that computes a subset $W' \subseteq W$ such that any node $w \in W'$ has at most $\Delta^{0.9}_W$ neighbors in $W'$. Moreover, at most $\frac{N_{high}}{\Delta_W}$ nodes in $W_{high}$ have no neighbor in $W'$.
\end{corollary}
\begin{proof}
Compute $W^{sub} \subseteq W$ using \Cref{thm:sampling main} on the bipartite graph with $U_{T\ref{thm:sampling main}} = \{w \in W \colon |\Gamma_G(w) \cap W| > \Delta_W^{0.9}\}$ as the side to be hit and $V_{T\ref{thm:sampling main}} = W$ as the hitting side, and $imp(u)=1$ for all $u\in U_{T\ref{thm:sampling main}}$. The algorithm runs in $\tilde{O}(\log (\Delta_W)\log^*(n))$ rounds, and gives a subset $W^{sub} \subseteq W$ with the following two guarantees: 
\begin{itemize} 
\item[(1)] Define $W_{missed}=\{w\in W^{high} : \Gamma_{G}(w)\cap W^{sub}=\emptyset \}$. Then, we have $|W_{missed}| \leq N_{high}/(\Delta_W)^{5}$. 
\item[(2)] Define $W_{dense} = \{w\in W : \, |\Gamma_{G}(w)\cap W^{sub}| > \Delta_{W}^{0.9}\}$. Then, we have $|W_{dense}| \leq N_{high}/(\Delta_W)^{5}$. 
\end{itemize}
We set $W'\gets W^{sub} \setminus W_{dense}$. That is, $W'$ is equal to $W^{sub}$ except for removing nodes $w \in W^{sub}$ that have more than $\Delta_{W}^{0.9}$ neighbors in $W^{sub}$. These removals happen simultaneously in one round. 

The removal of each node of $W_{dense}$ causes at most $\Delta_W$ nodes in $W_{high}$ to lose all neighbors in $W'$. Hence, the total number of nodes in $W_{high}$ that do not have a neighbor in $W'$ is at most \begin{align*}
    &|W_{missed}| + |W_{dense}| \cdot \Delta_W \\ \leq
&N_{high}/(\Delta_W)^{5} + N_{high}/(\Delta_W)^{5} \cdot \Delta_{W} \leq N_{high}/\Delta_W.
\end{align*} 
\end{proof}

We are now ready to present the ruling set algorithm of \Cref{thm:ruling_set_main}, which computes a $(2, O(\log \log \Delta))$ ruling set in $\tilde{O}(\log n)$ rounds. On a high level, given the fast sampling subroutine of \Cref{cor:ruling_set_sampling}, the approach is simple: we iteratively sparsify the set of nodes in $O(\log\log \Delta)$ stages, while increasing the distance to the original set slowly, by $1$ per stage. In each stage, via $O(\log_{\Delta_W} n)$ invocations of \Cref{cor:ruling_set_sampling}, we will be able to compute a set $W''\subset W$ that dominates \textit{all} of $W$ and has polynomially smaller maximum induced degree $\Delta^{0.9}_{W}$. Repeating this scheme for $O(\log\log \Delta)$ stages reduces the maximum induced degree to a constant, and then we can finish with a standard maximal independent set computation.

\begin{proof}[Proof of \Cref{thm:ruling_set_main}]
The algorithm consists of $s=O(\log\log \Delta)$ stages. Initially, we let $W=V(G)$. Throughout the stages, we gradually sparsify $W$ and remove more and more nodes from it, thus shrinking the maximum degree in the subgraph $G[W]$ induced by $W$, while ensuring that per stage the set $W$ goes at most $1$ distance further away (i.e., the set $W$ at the end of the stage will be a $1$-ruling set of the set $W$ at the start of that stage). Hence, in the end, the set $W$ will be an independent set, and each node $v\in V(G)$ will have distance at most $s=O(\log\log \Delta)$ from it.

Consider stage $i$ and let $\Delta_W$ denote the maximum degree in the subgraph $G[W]$ induced by vertices in $W$. The objective of stage $i$ is to compute a $1$-ruling set $W''\subset W$ such that, in the subgraph induced by $W''$, the maximum degree is at most $(\Delta_W)^{0.9}$. At the end of the stage, we set $W=W''$ and go to the next stage. This shrinks the maximum degree polynomially per stage. Once we have reached the base case that $\Delta_W\leq c$ for some fixed constant $c$, then we let $W''$ be an MIS of the subgraph $G[W]$---i.e., a $1$-ruling set with maximum induced degree at most $0$--- computed easily in $O(\log^* n)$ rounds using classic MIS algorithms for constant-degree graphs~\cite{linial92}. For the rest of this discussion, we focus on stages where $\Delta_W \geq c$. 

The stage involves $\tilde{O}(\log n)$ rounds. Fix $d=\Delta_W$ at the start and set $W''=\emptyset$. The stage is composed of $O(\log_{d} n)$ similar iterations, each made of $\tilde{O}(\log d \cdot \log^*(n))$ rounds. In each iteration $j$, we apply the sampling subroutine provided by \Cref{cor:ruling_set_sampling} on $G[W]$, with $N_{high}=(N/d^{j})$, where $N$ is the initial polynomial upper bound on $n$ known to all nodes. This sampling subroutine computes a subset $W'\subseteq W$ with maximum induced degree at most $d^{0.9}$ such that, all nodes of $W_{high}$ except at most $N_{high}/d$ of them have a neighbor in $W$. Then we set $W''\gets W'' \cup W'$ and $W\gets W\setminus N^+_G(W')$. That is, we add $W'$ to $W''$ and we remove from $W$ all nodes of $W'$ and their neighbors. After $O(\log_{d} n)$ iterations, there is no node in $W$ that has degree at least $d^{0.9}$ in $W$. We then stop and add all the nodes remaining in $W$ to the set $W''$. The maximum induced degree of $W'$ remains at most $d^{0.9}$ because in each iteration we remove from $W$ all nodes added to $W'$ and their neighbors. Moreover, $W''$ is a $1$-ruling set of $W$ because each node $w\in W$ was either removed in an iteration, in which case itself or a neighbor $w'$ was put in $W''$ in that iteration, or $w$ remained in $W$ after the last iteration and thus itself was added to $W''$.
\end{proof}


\section{Network Decomposition}
This section presents our network decomposition result, summarized in the statement below (this is equivalent to the statement in \Cref{thm:NDmain}, restated here for convenience):

\begin{theorem}
\label{thm:nd}
        There exists a deterministic distributed \LOCAL algorithm that computes in $\tilde{O}(\log^2 n)$ rounds a network decomposition $col \colon U_{ND} \mapsto [C(B)]$ of $G[U_{ND}]$ with diameter $O(\log n)$ using $O(\log n)$ colors.
\end{theorem}

\noindent Roadmap: The above result is built out of three ingredients, which we present next: a sampling result presented in \Cref{subsec:nd_sampling} which is a corollary of our main sampling theorem and will be used as a key tool in our decomposition algorithm, a base case algorithm presented in \Cref{subsec:nd_base_case} that handles the easier case where we have a headstart function with at most polylogarithmic badness and builds a network decomposition based on it, and the main recursive algorithm presented in \Cref{subsec:nd_recursive} that computes the network decomposition recursively, via the help of the aforementioned sampling tool and the base case algorithm.   

\subsection{Sampling Corollary for Network Decomposition}
\label{subsec:nd_sampling}

\begin{corollary}[Sampling Corollary for Network Decomposition]
	\label{cor:nd_sampling}

        There exists an absolute constant $c$ such that the following holds: Let $U \subseteq V(G)$, $d = \lceil\log(N) \cdot \left(\log\log(N)\right)^c \rceil$, $N_U \in \mathbb{R}$ with $N_U \geq |U|$, $B\in [\log^5 N, N]$, and $h \colon V(G) \mapsto \mathbb{N}$ with $\max_{u \in U} bad_{h,d}(u) \leq B$.

        There exists a deterministic distributed \LOCAL algorithm that computes in $\tilde{O}(\log(N)\log(B))$ rounds a mapping $h' \colon V(G) \mapsto \mathbb{N}$ with $\max_{v \in V(G)} h'(v) \leq 2\max_{v \in V(G)} h(v)+1$ such that for $U_{good} = \left\{u \in U \colon bad_{h',d}(u) \leq B^{0.5 + 1/\log^2\log(B)}\right\}$, it holds that $|U \setminus U_{good}| \leq \frac{N_U}{B^3}$.
\end{corollary}
\begin{proof}
We apply our main sampling theorem \cref{thm:sampling main}, on the following bipartite graph $H$: we put $d$ copies $U_1$, $U_2$, \dots, $U_d$ of the set $U$ , all on one side of $H$, and $V$ on the other side. We then add the following connections: for each node $u\in U$ and each $d'\in\{1, \dots, d\}$, connect its $d'$ copy $u_{d'}\in U_{d'}$ to all of $\left\{v \in S_{d'}(u) \colon h(v) = \max_{v' \in S_{d'}(u)} h(v') \right\}$, where $S_{d'}(u)$ is the set of vertices at distance exactly $d'$ from $u$ in graph $G$. Set $imp(u')=1/d$ for each copy node $u'\in U_1\cup U_2 \cup U_{d}$ and $IMP=N_U$. Notice that simulating each \local round of this graph $H$ can be done in $O(d) = \tilde{O}(\log n)$ rounds of graph $G$. The theorem computes in $\tilde{O}(\log B \log^* n)$ rounds of $H$, and thus $\tilde{O}(\log B \log n)$ rounds of $G$, a subset $V^{subset} \subseteq V$. We then set $h'(v) = 2h(v) + Indicator(v \in V^{sub})$.

We call copy $d'$ of node $u\in U$, for $d'\in\{1, \dots, d\}$ \textit{bad} if it has more than $B^{0.5+1/\log^2\log B}$ neighbors in $V^{subset}$, in the bipartite graph $H$. We call a node $u\in U$ bad if it has at least one \textit{bad copy}. Node $u \in U$ is called good otherwise, and notice that the set of good nodes is exactly equal to $U_{good}$ in the statement of the corollary. 

From \cref{thm:sampling main}, we conclude that $\sum_{\textit{bad \,} u\in U_1\cup \dots U_d} 1/d \leq N_U/B^5$. Thus, the total number of bad copies is at most $d\cdot N_U/B^5 \leq N_U/B^3$; the inequality holds as $d = \tilde{O}(\log n)$ and $B\geq \log^5 N$. Thus the number of bad nodes in $U$ is also at most $N_U/B^3$, i.e., $|U \setminus U_{good}| \leq \frac{N_U}{B^3}$. 
\end{proof}

\subsection{Base Case Algorithm for Partial Network Decomposition}
\label{subsec:nd_base_case}

\begin{lemma}
\label{lem:nd_base_case}
        There exists an absolute constant $c$ such that the following holds: Let $U \subseteq V(G)$,  $d = \lceil\log(N) \cdot \left(\log\log(N)\right)^c \rceil$, $N_U \in \mathbb{R}$ with $N_U \geq |U|$ and $h \colon V(G) \mapsto \mathbb{N}$ with $\max_{v \in V(G)} h(v) \leq \log(N)$ and  $\max_{u \in U} bad_{h,d}(u) \leq \log^{10}(N)$.
        
        There exists a deterministic distributed \LOCAL algorithm that computes in $\tilde{O}\left(\log(n)\right)$ rounds a set $U_{ND} \subseteq U$ with $|U \setminus U_{ND}| < \frac{N_U}{\log^{20}(N)}$ and a network decomposition $col \colon U_{ND} \mapsto [\lceil 100 \log \log (N) \rceil]$ of $G[U_{ND}]$ with diameter $O(\log n)$ using $\lceil 100 \log \log (N)\rceil$ colors.
\end{lemma}
\cref{lem:nd_base_case} is a simple corollary of \cref{thm:clusterone}. 

\begin{restatable*}{theorem}{clusterone}
    \label{thm:clusterone}
    Let $h \colon V(G) \mapsto \mathbb{N}_0$ with $\max_{v \in V(G)} h(v) = \tilde{O}(\log n)$ and $U \subseteq V(G)$ such that for every node $u \in U$, it holds that 

    \begin{align*}
        &|\{v \in V(G) \colon \\ & d(u,v) - h(v)  \leq \min_{v' \in V(G)} d(u,v') - h(v') + \log^{10}\log(N) \}|\\ & \leq \log^{100}(N).
    \end{align*}
    
    There exists a deterministic distributed \local algorithm that computes in $\tilde{O}(\log n)$ rounds a subset $U^{sub} \subseteq U$, $|U^{sub}| \geq |U| / 2$, such that every connected component in $G[U^{sub}]$ has diameter $O(\log n)$.
\end{restatable*}

\cref{thm:clusterone} is proven in \cref{sec:appendix_nd} and follows in a straightforward manner using techniques developed in \cite{ghaffari2023netdecomp}. We now use \cref{thm:clusterone} to prove \cref{lem:nd_base_case}.
\begin{proof}[Proof of \cref{lem:nd_base_case}]
Let $\tilde{h} \colon V(G) \mapsto \mathbb{N}_0$ with $\tilde{h}(v) = \log^{20}\log(N)\cdot h(v)$ for every $v \in V(G)$. The scaled head starts $\tilde{h}$ satisfy the input condition of \cref{thm:clusterone}. In particular, for every $u \in U$,

    \begin{align*}&|\{v \in V(G) \colon \\ &d(u,v) - \tilde{h}(v) \leq \min_{v' \in V(G)} d(u,v') - \tilde{h}(v') + \log^{10}\log(N) \}| \\ &\leq \tilde{O}(\log N)\cdot bad_{h,d}(u) \leq  \log^{100}(N).
    \end{align*}
We now invoke the algorithm of \cref{thm:clusterone}$\lceil 100\log \log (N)\rceil$ times, each time with the head starts $\tilde{h}$. In the $i$-th invocation, we give as input the subset $U_i$ consisting of all nodes in $U$ that have not been assigned a color in the iterations before. As a result, we obtain a subset $U^{sub}_i \subseteq U_i$ satisfying $|U^{sub}_i| \geq 0.5|U_i|$ such that every connected component in $G[U^{sub}_i]$ has diameter $O(\log n)$. We then assign each node in $U^{sub}_i$ the color $i$. Finally, we set $U_{ND} = \sqcup_{i=1}^{\lceil 100\log \log (N)\rceil} U^{sub}_i$. We have

\[|U \setminus U_{ND}| \leq 0.5^{\lceil 100\log \log (N)\rceil}|U| \leq \frac{N_U}{\log^{10}(N)}.\]
\end{proof}

\subsection{Recursive Algorithm for Partial Network Decomposition}
\label{subsec:nd_recursive}

In this subsection, we describe the main part of our network decomposition algorithm, as a recursive algorithm for computing a partial network decomposition. A concise pseudocode of this recursive algorithm is presented in \Cref{alg:ND}, and \Cref{lem:nd_recursive} states the formal guarantee provided by the algorithm. To help the readability, we have added an intuitive but informal description of the algorithm after the lemma statement.

\begin{algorithm}[ht]
        Let $N$ denote the polynomial upper bound on the number of vertices known to all nodes. We assume that $N \geq c$ where $c$ is a sufficiently large constant and $d := \lceil\log(N) \cdot \left(\log\log(N)\right)^c \rceil$. \\
	Input: $U \subseteq V(G)$, $N_U \in \mathbb{R}_{\geq 0}$ with $N_U \geq |U|$, $B\in [\log^5 N, N]$ and 
	$h \colon V(G) \mapsto \mathbb{N}$ satisfying \\
        (1) $\max_{v \in  V(G)} h(v) \leq (4.5-\frac{2\log B}{\log N})\left(1 + \frac{100}{\log \log (B)}\right)\frac{\log(N)}{\log(B)}$, and \\
	\hspace{100pt} (2) $\max_{u \in U} bad_{h,d}(u) \leq B$ \\
        Output: Set $U_{ND} \subseteq U$ with $|U \setminus U_{ND}| \leq \frac{N_U}{B}$ and a network decomposition $col \colon U_{ND} \mapsto [C(B)]$ of $G[U_{ND}]$ with diameter $O(\log n)$ using $C(B) = 50\left\lceil  \left( 1 - \frac{100}{\log\log(B)}\right)\log(B) \right\rceil = O(\log(B))$ colors.
	\caption{Recursive (Partial) Network Decomposition Algorithm}
   \begin{algorithmic}[1]
		\Procedure{\textit{Partial-ND}}{$U, N_U, B, h$}
		\If{$B \leq \log^{10}(N)$}
		\State $(U_{ND},col) \gets \mathcal{A}_{L\ref{lem:nd_base_case}}(U, N_U, h)$ \Comment{$|U \setminus U_{ND}| \leq \frac{N_U}{\log^{20}(N)}$}
		\State \Return $(U_{ND},col)$
		\Else
		\State $h_{rec} \gets \mathcal{A}_{C\ref{cor:nd_sampling}}(U, N_U, B, h)$ \Comment{$\max_{v \in V(G)} h_{rec}(v) \leq 2\max_{v \in V(G)} h(v)$+1}
            \State $\eps_B \gets \frac{1}{\log^2\log(B)}$
		\State $U^{(1)} = \left\{ u \in U \colon bad_{h^{'},d}(u) \leq B^{0.5 + \eps_B}\right\}$ \Comment{$|U \setminus U^{(1)}| \leq \frac{N_U}{B^3}$}
		\State $(U^{(1)}_{ND},col^{(1)})  \gets \textit{Partial-ND}\left(U^{(1)},N_U, B^{0.5 +\eps_B}, h_{rec} \right)$
		\State $U^{(2)} \gets U^{(1)} \setminus U^{(1)}_{ND}$ \Comment{$|U^{(2)}| < \frac{N_U}{B^{0.5 + \eps_B}}$}
            \State $N^{(2)}_U \gets N_U/B^{0.5 + \eps_B}$
		\State $(U^{(2)}_{ND},col^{(2)})  \gets \textit{Partial-ND}\left(U^{(2)},N^{(2)}_U, B^{0.5 + \eps_B}, h_{rec} \right)$ \Comment{$|U^{(2)} \setminus U^{(2)}_{ND}| < \frac{N^{(2)}_U}{B^{0.5 + \eps_B}}$}
		\State $U_{ND} \gets U^{(1)}_{ND} \sqcup U^{(2)}_{ND}$
		\State $col(u) = \begin{cases}
			col^{(1)}(u)  \text{ if $u \in U^{(1)}_{ND}$} \\
			col^{(2)}(u) + C(B)/2 \text{ if $u \in U^{(2)}_{ND}$ }
		\end{cases}$
		\State \Return $(U_{ND}, col)$
		\EndIf
		\EndProcedure
	\end{algorithmic}
 \label{alg:ND}
\end{algorithm}

\begin{lemma}
\label{lem:nd_recursive}
   There exists an absolute constant $c$ such that the following holds: Let $U \subseteq V(G)$,  $d = \lceil\log(N) \cdot \left(\log\log(N)\right)^c \rceil$, $N_U \in \mathbb{R}$ with $N_U \geq |U|$ and $h \colon V(G) \mapsto \mathbb{N}$ with $\max_{v \in V(G)} h(v) \leq 3\left(1 + \frac{100}{\log \log (B)}\right)\frac{\log(N)}{\log(B)}$ and $\max_{u \in U} bad_{h,d}(u) \leq B$ for $B\in [\log^5 N, N]$.
        
   \Cref{alg:ND} computes in $\tilde{O}\left(\log(n) \log(B)\right)$ rounds of the \local model a set $U_{ND} \subseteq U$ with $|U \setminus U_{ND}| < \frac{N_U}{B^2}$ and a network decomposition $col \colon U_{ND} \mapsto [C(B)]$ of $G[U_{ND}]$ with diameter $O(\log n)$ using $C(B)= 50\left\lceil  \left( 1 - \frac{100}{\log\log(B)}\right)\log(B) \right\rceil = O(\log(B))$ colors.
\end{lemma}

\paragraph{An intuitive description of the Partial Network Decomposition (\Cref{alg:ND}):} We present an intuitive but informal discussion here, which hopefully helps in reading \Cref{alg:ND}, particularly discussing the intention of the parts of the pseudocode. The base case where the badness upper bound $B$ of the headstart function $h$ is at most polylogarithmic is handled directly (in line 3) via the base case network decomposition algorithm we described in \Cref{lem:nd_base_case}. 

Let us consider larger values of $B>\log^{10} N$, when the algorithm uses recursion. In this case, the Partial-ND algorithm computes its partial network decomposition via two recursive calls to itself, done one after the other (in lines 9 and 12). The output network decomposition is the union of these two decompositions (line 13), where the colors of the first decomposition are shifted up sufficiently (in line 14) to remain disjoint from the colors used by the latter.

These two partial decompositions are computed via recursions with a roughly quadratically smaller badness value of $B'=B^{0.5+\eps_B}$ for $\eps_{B}=1/\log\log^2 B$. In the first reading, the reader may find it more intuitive to think of $\eps_{B}$ as $0$ as that gives a simpler to analyze recursive structure. But we need to provide such a head start function with badness at most $B'$ as an input to these two recursive calls. We do that via invoking \Cref{cor:nd_sampling} (in line 6), at the start of the recursive algorithm, which refines the headstart function $h()$ with which we started into another head start function $h_{rec}$ whose badness is at most $B'$. We note that this comes at a cost: $h_{rec}$ will have roughly a $2$ factor larger maximum value compared to $h$, but this growth will remain manageable as we see in the formal recursion in the proof. 

Finally, let us discuss a key aspect in the design of this recursion: each of the three procedures used--- the refinement of the headstart function, and also the two recursive calls to partial ND--- drop a fraction of nodes (e.g., do not reduce their badness or put them in the network decomposition) and this is also the reason that the overall scheme remains a \textit{partial} network decomposition. The recursion is arranged such that, overall, we can guarantee that the number $|U\setminus U_{ND}|$ of such dropped nodes (i.e., those that are not in the computed network decomposition) is at most a $1/B^2$ fraction of the nodes of $U$--formally, of the upper bound $N_{U}$ of $|U|$ provided to all nodes. In particular, the refinement of the head start function from $h()$ to $h_{rec}$ (in line 6) drops at most $N_{U}/B^3$ nodes, the first recursive call to Partial ND (in line 9) computes a decomposition of all but at most $N_{U}/(B')^2$ of the remaining nodes, and the second recursive call (in line 12) computes another decomposition of these remaining nodes to reduce the final number of remaining nodes further down to $N_{U}/(B')^2 \cdot 1/(B')^2$. Hence, overall, the number of nodes dropped from the computed partial decomposition is at most $N_{U}/(B)^3 + N_{U}/(B')^2 \cdot 1/(B')^2 \leq N_{U} / B^2$, as desired.  

Having this intuitive discussion of the algorithm, and the pseudocode in \Cref{alg:ND}, we are now ready to prove the guarantees of the recursive partial network decomposition algorithm claimed in \Cref{lem:nd_recursive}.

\begin{proof}[Proof of \Cref{lem:nd_recursive}] 
We first discuss correctness and then analyze the round complexity.

\paragraph{Correctness.} The proof is again by induction. Let us start with the base case, where $B\leq \log^{10} N$. Notice that in this case, by input assumptions, we also have $B\geq \log^5 N$. We obtain the output network decomposition directly by invoking \Cref{lem:nd_base_case} with the inputs $U$ and $h$. We argue that the input assumptions of \Cref{lem:nd_base_case}  are satisfied. In particular, we have $\max_{v \in V(G)} h(v) \leq (4.5 - \frac{2\log B}{\log N})\left(1 + \frac{100}{\log \log (B)}\right)\frac{\log(N)}{\log(B)} \leq \log(N)$ given that $B \geq \log^5 N$ and $N$ is lower bounded by a large constant. We also have  $\max_{u \in U} bad_{h,d}(u) \leq B \leq \log^{10}(N)$. Hence, from  \Cref{lem:nd_base_case}, we get a set $U_{ND} \subseteq U$ with $|U \setminus U_{ND}| \leq \frac{N_U}{\log^{20}(N)} \leq \frac{N_U}{B^2}$ and a network decomposition $col \colon U_{ND} \mapsto [\lceil 100 \log \log (N) \rceil]$ of $G[U_{ND}]$ with diameter $O(\log n)$ using a number of colors upper bounded by $\lceil 100 \log \log (N)\rceil \leq 50\left\lceil  \left( 1 - \frac{100}{\log\log(B)}\right)\log(B)\right\rceil$, where the inequality holds since $B\geq \log^5 N$ and $N$ is lower bounded by a sufficiently large constant.

We now discuss the inductive case, assuming that $B> \log^{10} N$. We can inductively assume that the two recursive calls to Partial Network Decomposition, which work on smaller values of $B$, provide the desired guarantees if their input conditions are met. In the case with $B> \log^{10} N$, we first call \Cref{cor:nd_sampling} to compute an updated head-start function $h_{rec} \colon V(G) \mapsto \mathbb{N}$ such that $\max_{v \in V(G)} h_{rec}(v) \leq 2\max_{v \in V(G)} h(v)+1$, and which satisfies the following property: Let $U^{(1)} = \left\{u \in U \colon bad_{h_{rec},d}(u) \leq B^{0.5 + \eps_{B}}\right\}$. Note that this set was called $U_{good}$ in the statement of \Cref{cor:nd_sampling}. This corollary guarantees that $|U \setminus U^{(1)}| \leq \frac{N_U}{B^3}$. Now, we proceed to the two recursive calls to Partial-ND, using this updated head-start function $h_{rec}$, and with the reduced badness upper bound $B'=B^{0.5+\eps_{B}}$. For both calls, we need to argue that the input requirements are met. In particular, given that $B\leq N$, we have
\begin{align*}
    &max_{v}\in V(G) h_{rec}(v) \leq \\ 
    &2\max_{v\in V(G)} h(v)+1 \leq \\ 
    &2(4.5-\frac{2\log B}{\log N}) \left(1 + \frac{100}{\log \log (B)}\right)\frac{\log(N)}{\log(B)}+1 \leq \\ 
    &2(4.5-\frac{2\log B}{\log N})(1+\frac{\log B}{5\log N}) \left(1 + \frac{100}{\log \log (B)}\right)\frac{\log(N)}{\log(B)} \leq \\
    & (4.5 - \frac{2\log B^{0.5+1/\log^2 \log B}}{\log N}) \cdot \\ & \left(1 + \frac{100}{\log \log \left(B^{0.5 + 1/\log^2\log(B)}\right)}\right)\frac{\log N}{\log \left(B^{0.5 + 1/\log^2\log(B)}\right)} 
\end{align*}
Hence, from the first recursive Partial-ND call, by the inductive hypothesis, we get a partial network decomposition $(U^{(1)}_{ND},col^{(1)})$ with diameter $O(\log n)$ using $C(B') = 50\left\lceil  \left( 1 - \frac{100}{\log\log(B')}\right)\log(B') \right\rceil$ colors, such that $U^{(1)}\setminus U^{(1)_{ND}} \leq \frac{N_{U}}{(B^{0.5+\eps_{B}})^2}.$ We then apply the second recursive Partial-ND call, on this remainder set $U^{(2)} \gets U^{(1)} \setminus U^{(1)}_{ND}$  and with the input upper bound $N^{(2)}_U \gets N_U/(B^{0.5 + \eps_B})^2$. We thus get, by the inductive hypothesis, we get a partial network decomposition$(U^{(2)}_{ND},col^{(2)})$  with diameter $O(\log n)$ using $C(B') = 50\left\lceil  \left( 1 - \frac{100}{\log\log(B')}\right)\log(B') \right\rceil$ additional colors, and such that $U^{(2)}\setminus U^{(2)_{ND}} \leq \frac{N^{(2)}_{U}}{(B^{0.5+\eps_{B}})^2}.$ To provide the output network decomposition, we shift the colors of this second decomposition by $C(B)/2$ --- i.e., we add $C(B)/2$ to each color---which ensures that the two partial network decompositions use disjoint colors. The reason is that 

\begin{align*}
    C(B^{0.5+\eps_{B}}) &= 50\left\lceil\left( 1 - \frac{100}{\log\log(B^{0.5+\eps_B})}\right)\log(B^{0.5+\eps_B})\right\rceil \\ &< 25\left\lceil\left( 1 - \frac{100}{\log\log(B)}\right)\log(B^{0.5})\right\rceil = \frac{C(B)}{2},
\end{align*}

as $\eps_{B}=1/\log^2\log B$. Similarly, we see that the total number of colors used is $C(B)/2+C(B^{0.5+\eps_{B}}) \leq C(B)$. This provides a network decomposition for all of $U_{ND} = U^{(1)}_{ND} \sqcup U^{(2)}_{ND}$. The only nodes that are not in this network decomposition are those in $U\setminus U^{(1)}$, and those in $U^{(2)} \setminus U^{(2)}_{ND}$, and we can upper bound their number by $\frac{N_{U}}{B^{3}} + \frac{N^{(2)}_U}{(B^{0.5 + \eps_B})^2} = \frac{N_{U}}{B^{3}} + \frac{N_U}{(B^{0.5 + \eps_B})^4} < \frac{N_{U}}{B^2}$, as desired.

\paragraph{Round complexity:} In the base case $B\leq \log^{10} N$, the round complexity is $\tilde{O}(\log n)$ by \Cref{lem:nd_base_case}. In the recursive case where $B> \log^{10} N$, the recursive complexity $T(B)$ consists of $\tilde{O}(\log N \log B)$ rounds for refining the head start function via \Cref{cor:nd_sampling}, and then two sequential calls to $T(B^{0.5+\eps_B})$. Hence, overall, we have 

\begin{align*} &T(B) \\ \leq &\begin{cases}
		\tilde{O}(\log n) \text{ if $B \leq \log^{10} N$} \\
		2T\left(B^{0.5 + \frac{1}{\log^2\log(B)}} \right) + \tilde{O}(\log N \log B) \text{ if $B > \log^{10} N$}.
	\end{cases}
 \end{align*}
Hence, one can see that $T(B)\leq \tilde{O}(\log N \log B).$ 

\end{proof}

Having this recursive network decomposition algorithm, we are now ready to provide the proof of our main network decomposition claim, stated in \Cref{thm:nd}:

\begin{proof}[Proof of \Cref{thm:nd}]
    Invoke \Cref{lem:nd_recursive} with $U=V(G)$, $B=N$, $N_U=N$, and $h(v)=1$ for all $v\in V(G)$. All the input guarantees are satisfied. Thus, $|U\setminus U_{ND}| < \frac{N_U}{B^2} = 1/N \ll 1$ and therefore $U_{ND}$ is a network decomposition of the entire set $U=V(G)$.
\end{proof}

\section{MIS}
This section presents our maximal independent set result, summarized in the statement below. This is equivalent to the part of the statement in \Cref{thm:MISmain} about MIS. The other parts of \Cref{thm:MISmain} about maximal matching and colorings follow by standard reductions\cite{linial92}.
\begin{theorem}
    \label{thm:mis_main}
    There is a deterministic distributed algorithm that, in any $n$-node graph, computes a maximal independent set in $\tilde{O}(\log^{5/3} n)$ rounds of the \local model.
\end{theorem}

\noindent Roadmap: The above result is built out of three ingredients, which we present next: a sampling result presented in \Cref{subsec:mis_sampling} which is a corollary of our main sampling theorem and will be used as a key tool in our MIS algorithm, a base case algorithm presented in \Cref{subsec:mis_base_case} that handles the easier case of the independent set computation given a headstart function with somewhat small badness, and the main recursive algorithm presented in \Cref{subsec:mis_algorithm_main} that computes an almost-maximal independent set recursively, via the help of the aforementioned sampling tool and the base case algorithm.

\subsection{Sampling Corollary for MIS}
\label{subsec:mis_sampling}
We make use of the following tool, which is a corollary of our main deterministic sampling result: 

\begin{corollary}[Sampling Corollary for MIS]
	\label{cor:mis_sampling}

        There exists an absolute constant $c$ such that the followng holds: Let $U \subseteq V(G)$, $M_U \in \mathbb{R}$ with $M_U \geq |E(G[U])|$, $B \in \mathbb{R}$ with $B \geq c$, $k \in \mathbb{N}$ with $B \geq k$, and consider a function $h \colon V(G) \mapsto \mathbb{N}$ such that $\max_{u \in U} h(u) \leq 10k$ and and $\max_{u \in U} bad_{h,100k}(u) \leq B$.

        There exists a deterministic distributed \LOCAL algorithm that computes in $\tilde{O}(k\log(B)\log^*(N))$ rounds a mapping $h' \colon V(G) \mapsto \mathbb{N}$ with $\max_{v \in V(G)} h'(v) \leq 2\max_{v \in V(G)} h(v)+1$ such that for $U_{good} = \left\{u \in U \colon bad_{h',100k}(u) \leq B^{0.5 + 1/\log^2\log(B)}\right\}$, it holds that $|E(G[U]) \setminus E(G[U^{good}])| \leq \frac{M_U}{B^3}$.
\end{corollary}
\begin{proof}
We apply our main sampling theorem \cref{thm:sampling main}, on the following bipatite graph $H$: we put $100k$ copies $U_1$, $U_2$, \dots, $U_{100k}$ of the set $U$ , all on one side of $H$, and $V$ on the other side. We then add the following connections: for each node $u\in U$ and each $d'\in\{1, \dots, 100k\}$, connect its $d'$ copy $u_{d'}\in U_{d'}$ to all of $\left\{v \in S_{d'}(u) \colon h(v) = \max_{v' \in S_{d'}(u)} h(v') \right\}$, where $S_{d'}(u)$ is the set of vertices at distance exactly $d'$ from $u$ in graph $G$. Set $imp(u')=|\Gamma_{G}(u)\cap U|/(200k)$ for each copy node $u'\in U_1\cup U_2 \cup U_{d}$ and $IMP=M_U$. Notice that simulating each \local round of this graph $H$ can be done in $O(k)$ rounds of graph $G$. The theorem computes in $\tilde{O}(\log B \log^* n)$ rounds of $H$, and thus $\tilde{O}(k\log B \log^* n)$ rounds of $G$, a sample subset $V^{subset} \subseteq V$. We then set $h'(v) = 2h(v) + Indicator(v \in V^{sub})$.

We call copy $d'$ of node $u\in U$, for $d'\in\{1, \dots, 100k\}$, \textit{bad} if it has more than $B^{0.5+1/\log^2\log B}$ neighbors in $V^{subset}$, in the bipartite graph $H$. We call a node $u\in U$ bad if it has at least one \textit{bad copy}. Node $u \in U$ is called good otherwise, and notice that the set of good nodes is exactly equal to $U_{good}$ in the statement of the corollary. 

From \cref{thm:sampling main}, we conclude that $\sum_{\textit{bad \,} u\in U_1\cup \dots U_{100k}} |\Gamma_{G}(u)\cap U|/(200k) \leq M_U/B^5$. Thus, the total number of edges in $G[U]$ incident on bad copies is at most $200k\cdot M_U/B^5 \leq M_U/B^3$; the inequality holds as $200k\leq B^2$. Thus, the total number of edges of $G[U]$ incident on bad nodes in $U$ is also at most $M_U/B^3$, which means $|E(G[U]) \setminus E(G[U^{good}])| \leq \frac{N_U}{B^3}$. 
\end{proof}

\subsection{Base Case Algorithm for Almost-Maximal Independent Set}
\label{subsec:mis_base_case}

\begin{lemma}[Base Case Algorithm for Almost Maximal Independent Set]
\label{lem:mis_base_case}
        Let $U \subseteq V(G)$, $k \in \mathbb{N}$ and $h \colon V(G) \mapsto \mathbb{N}$ with $\max_{v \in V(G)} h(v) \leq 10k$ and suppose that $ \max_{u \in U} bad_{h,100k}(u) \leq N^{1/k}$.
        There exists a deterministic distributed \LOCAL algorithm that computes in $\tilde{O}\left(\frac{\log^3(n)}{k^3} + \log(n)\right)$ rounds an independent set $U_{IS}$ such that the following holds: Let $U_{rem} = \left\{u \in U \setminus U_{IS} \colon \Gamma_G(u) \cap U_{IS} = \emptyset \right\}$ consists of all nodes in $U$ that are not in $U_{IS}$ and that have no neighbor in $U_{IS}$. Then,

        \[\left|E(G[U_{rem}])\right| < N^{-100/k}|E(G[U])|.\]
\end{lemma}
\cref{lem:mis_base_case} follows as a simple corollary of \cref{thm:independent_set_given_partition}.
\begin{restatable*}{theorem}{is}
\label{thm:independent_set_given_partition}
Let $D, DEG \in \mathbb{N}$, $U \subseteq V(G)$ and $\fC$ a partition of diameter at most $D$. Assume that each node in $U$ has cluster degree at most $DEG$ with respect to $\fC$. Then, there exists a distributed \local algorithm that computes in $O\left(D + \log^2(DEG\log n)\right)$ rounds an independent set $U_{IS} \subseteq U$ with the following guarantee: Let $U_{rem}$ denote the set that consists of all nodes in $U$ that are not in $U_{IS}$ and also do not have a neighbor in $U_{IS}$. Then, the total number of edges in the graph $G[U_{rem}]$ is at most half the number of edges in the graph $G[U]$.
\end{restatable*} 
\cref{thm:independent_set_given_partition} is proven in \cref{sec:appendix_is} and follows in a straightforward manner using techniques developed in \cite{ghaffarigrunau2023fasterMIS}. We now use \cref{thm:independent_set_given_partition} to prove \cref{lem:mis_base_case}.
\begin{proof}[Proof of \cref{lem:mis_base_case}]
We first apply \cref{lem:partition_given_head} with the provided h() function  to obtain a partition $\fC$ of $U$ with diameter $O(k)$ such that for any node $u \in U$, the degree of $u$ with respect to $\fC$ is at most $O(k \cdot N^{1/k})$. Then, we apply \cref{thm:independent_set_given_partition} for $200(\log(N)/k)$ iterations, each time adding the newly computed IS to the previous IS and then removing it along with its neighbors from the graph. Each application takes $O(k+ \log^2(kN^{1/k}))\cdot \poly(\log\log n) = 
O(k+ (\log(N)/k)^2))\cdot \poly(\log\log n)$ rounds, and hence the total round complexity is $\tilde{O}((k+ (\log(n)/k)^2)) \log n/k) = \tilde{O}(\log^3 n/k^3 + \log n)$. In each iteration, the fraction of remaining edges decreases by a $2$ factor. Hence, in the end, for the remaining set of vertices $U_{rem}$, we have $\left|E(G[U_{rem}])\right| \leq N^{-200/k}|E(G[U])| < N^{-100/k}|E(G[U])|.$
\end{proof}

\subsection{Recursive Algorithm for Almost-Maximal Independent Set}
\label{subsec:mis_algorithm_main}
In this subsection, we describe the main part of our MIS algorithm, as a recursive algorithm that computes an almost-maximal independent set. A concise pseudocode of this recursive algorithm is presented in \Cref{alg:almost_MIS}, and \Cref{lem:mis_algorithm_main} states the formal guarantee provided by the algorithm. To help the readability, we have added an intuitive but informal description of the algorithm after the lemma statement.

\begin{algorithm}[th]
        Let $N$ denote the polynmial upper bound on the number of vertices known to all nodes. We assume that $N \geq c$ where $c$ is a sufficiently large constant and we define $k := \lceil \log^{2/3}(N)\rceil$. \\
	Input: $U \subseteq V(G)$, $M_U \in \mathbb{R}_{\geq 0}$ with $M_U \geq |E(U)|$, $B \in [N^{0.5/k}, N]$ and 
	$h \colon V(G) \mapsto \mathbb{N}$ satisfying \\
        (1) $\max_{v \in  V(G)} h(v) \leq (4.5-\frac{2\log B}{\log N}) \left(1 + \frac{100}{\log \log (B)}\right)\frac{\log(N)}{\log(B)}$, and \\
	\hspace{100pt} (2) $\max_{u \in U} bad_{h,100k}(u) \leq B$ \\
        Output: Independent Set $U_{IS} \subseteq U$ with $|E(U_{rem})| < \frac{M_U}{B^2}$ for $U_{rem} := U \setminus (U_{IS} \cup \Gamma_{G}(U_{IS}))$
	\caption{Recursive (Almost) MIS Algorithm}
   	\label{alg:almost_MIS}
	\begin{algorithmic}[1]
		\Procedure{\textit{Almost-MIS}}{$U, M_U, B, h$}
		\If{$B \leq N^{1/k}$}
		\State $U_{IS} \gets \mathcal{A}_{L\ref{lem:mis_base_case}}(U, k, h)$ \Comment{$|E(U_{rem})| \leq \frac{M_U}{N^{100/k}}$ for $U_{rem} := U \setminus  (U_{IS} \cup \Gamma_{G}(U_{IS}))$}
		\State \Return $U_{IS}$
		\Else
		\State $h_{rec} \gets \mathcal{A}_{C\ref{cor:mis_sampling}}(U, M_U, B, k, h)$ \Comment{$\max_{v \in V(G)} h_{rec}(v) \leq 2\max_{v \in V(G)} h(v)$+1}
            \State $\eps_B \gets \frac{1}{\log^2\log(B)}$
		\State $U^{(1)} = \left\{ u \in U \colon bad_{h_{rec},100k}(u) \leq B^{0.5 + \eps_B}\right\}$ \Comment{$|E(U) \setminus E(U^{(1)})| \leq \frac{M_U}{B^3}$}
		\State $U^{(1)}_{IS} \gets \textit{Almost-MIS}\left(U^{(1)},M_U, B^{0.5 +\eps_B}, h_{rec} \right)$
		\State $U^{(2)} \gets U^{(1)} \setminus (U^{(1)}_{IS} \cup \Gamma_{G}(U^{(1)}_{IS}))$ \Comment{$|E(U^{(2)})]| < \frac{M_U}{(B^{0.5 + \eps_B})^2}$}
            \State $M^{(2)}_U \gets M_U/(B^{0.5 + \eps_B})^2$
		\State $U^{(2)}_{IS} \gets \textit{Almost-MIS}\left(U^{(2)},M^{(2)}_U, B^{0.5 + \eps_B}, h_{rec} \right)$ \Comment{$|E(U^{(2)} \setminus B_1(U^{(2)}_{IS}))| < \frac{M^{(2)}_U}{(B^{0.5 + \eps_B})^2}$}
		\State $U_{IS} \gets U^{(1)}_{IS} \sqcup U^{(2)}_{IS}$
		\State \Return $U_{IS}$ 
		\EndIf
		\EndProcedure
	\end{algorithmic}
\end{algorithm}

\begin{lemma}
\label{lem:mis_algorithm_main}
\cref{alg:almost_MIS} computes an independent set $U_{IS} \subseteq U$. Moreover, let $U_{rem}=U\setminus (U_{IS} \cup \Gamma_{G}(U_{IS}))$ consist of all nodes in $U$ that are not in $U_{IS}$ and have no neighbor in $U_{IS}$. Then, $|E(U_{rem})| < \frac{M_U}{B^2}$. \cref{alg:almost_MIS} can also be implemented in $\tilde{O}\left(\left(\frac{\log^2(n)}{k^2} + k \right)\log(B)\right) = \tilde{O}(\log^{2/3}(N)\log(B))$ rounds in the distributed \LOCAL model.
\end{lemma}
\paragraph{An intuitive description of the Almost-Maximal Independent Set Algorithm (\Cref{alg:almost_MIS}):} The base case where the badness upper bound $B$ of the headstart function $h$ is at most $N^{1/k}$ is handled directly (in line 3) via the base case algorithm we described in \Cref{lem:mis_base_case}. 

Let us consider larger values of $B$. In this case, the algorithm computes its almost-maximal independent set via two recursive calls to itself, done one after the other (in lines 9 and 12), which return two independent sets $U^{(1)}_{IS}$ and $U^{(2)}_{IS}$. The output set is the union of these two sets (line 13), i.e., $U_{IS} \gets U^{(1)}_{IS} \sqcup U^{(2)}_{IS}$. Since we have removed the neighborhood of the first set from the graph after the first call (line 10), the second set is indeed independent of the first, and thus their union is an independent set in the original graph.

These two independent sets are computed via recursions with the assumption that there is head start function with roughly a quadratically smaller badness value of $B'=B^{0.5+\eps_B}$ for $\eps_{B}=1/\log\log^2 B$, instead of $B$. But we need to provide such a head function with badness at most $B'$ as an input to these two recursive calls. We do that via invoking \Cref{cor:mis_sampling} (in line 6), at the start of the recursive algorithm, which refines the headstart function $h()$ with which we started into another head start function $h_{rec}$ whose badness is at most $B'$. We note that this comes at a cost: $h_{rec}$ will have roughly a $2$ factor larger maximum value compared to $h$, but this growth will remain manageable as we see in the formal recursion in the proof. 

Finally, let us discuss a key aspect in the design of this recursion: each of the three procedures used--- the refinement of the headstart function, and also the two recursive calls to almost-maximal IS--- drop a fraction of nodes (e.g., do not reduce their badness or have them in the neighborhood of the computed independent set) and this is also the reason that the overall scheme remains \textit{almost-maximal} independent set rather than a maximal independent set. The recursion is arranged such that, overall, we can guarantee that the number of edges $E(U_{rem})$ in the subgraph induced by the remaining nodes $U_{rem}=U\setminus (U_{IS} \cup \Gamma_{G}(U_{IS}))$ is at most a $1/B^2$ fraction of the number of edges in $E[U]$--formally, of its upper bound $M_{U}$ provided to all nodes. In particular, the refinement of the head start function from $h()$ to $h_{rec}$ (in line 6) drops nodes incident on at most $M_{U}/B^3$ edges, the first recursive call to almost-MIS (in line 9) computes an almost-maximal independent set such that if we remove it and its neighborhood from the graph, $M_{U}/(B')^2$ edges remain in the induced subgraph of remaining nodes, and the second recursive call (in line 12) computes another almost-maximal independent set of these remaining nodes so that we further reduce the number of edges in the subgraph induced by the finally remaining nodes to $M_{U}/(B')^2 \cdot 1/(B')^2$. Hence, overall, the number of edges of the subgraph induced by the remaining nodes is at most $M_{U}/(B)^3 + M_{U}/(B')^2 \cdot 1/(B')^2 \leq M_{U} / B^2$, as desired.  

Having this intuitive discussion of the algorithm, and the pseudocode in \Cref{alg:almost_MIS}, we are now ready to prove the guarantees of the recursive partial network decomposition algorithm claimed in \Cref{lem:mis_algorithm_main}.

\begin{proof}[Proof of \Cref{lem:mis_algorithm_main}]
We first discuss correctness and then analyze the round complexity.

\paragraph{Correctness:} Let us start with the base case, where $B \leq N^{1/k}$. 
One can verify that $U,M_U,k,h$ satisfy the input requirements of \Cref{lem:mis_base_case}. In particular, notice that by the input assumptions, we have $B \geq N^{0.5/k}$ which implies $max_{v\in V(G)} h(v) \leq (4.5-\frac{2\log B}{\log N})\left(1 + \frac{100}{\log \log (B)}\right)\frac{\log(N)}{\log(B)} \leq 10k$, and moreover, $\max_{u \in U} bad_{h,100k}(u) \leq B < N^{1/k}$. Therefore, \Cref{lem:mis_base_case} gives us an independent set $U_{IS}$ with the guarantee that for the set $U_{rem}$ of remaining nodes, we have 
$\left|E(G[U_{rem}])\right| \leq \frac{|E(G[U])|}{N^{100/k}} < \frac{M_U}{B^2}.$

It remains to consider the inductive case $B > N^{1/k}$, and we can inductively assume that the two recursive calls provide the desired guarantee, given that the input conditions are met. 
Note that $U,M_U,B,k,h$ directly satisfy the input requirements of $\mathcal{A}_{C\ref{cor:mis_sampling}}$. In particular, $\max_{u\in U} h(u) \leq (4.5-\frac{2\log B}{\log N})(1+\frac{100}{\log\log B}) \frac{\log(N)}{\log B} \leq 4.5(1+\frac{100}{\log\log B}) k\leq 10k$, where the penultimate inequality holds as $B>N^{1/k}$ and the ultimate inequality holds as $B$ and $N$ are lower bounded by sufficiently large constants. Therefore, from \Cref{cor:mis_sampling}, we get an updated function $h_{rec} \colon U \mapsto \mathbb{N}$ with $\max_{u \in U} h_{rec}(u) \leq 2\max_{u \in U} h(u)$ with the following property: define the set $U^{(1)} = \left\{u \in U \colon bad_{h_{rec},100k}(u) \leq B^{0.5 + 1/\log^2\log(B)}\right\}$. We note that this was called $U^{good}$ in the statement of \Cref{cor:mis_sampling}. We have that that $|E(G[U]) \setminus E(G[U^{(1)}])| \leq \frac{M_U}{B^3}$. Intuitively, $U^{(1)}$ is the set of vertices that remain in the process for this iteration of independent set computation, and nodes in $U\setminus U^{(1)}$ are temporarily discarded; they make us a much smaller graph with at most $\frac{M_U}{B^3}$ edges, which will be handled in later calls to recursive AMIS.

Next, we check that the input requirements of the first recursive call are satisfied. In particular, given that $B\leq N$, we have
\begin{align*}
    &max_{v}\in V(G) h_{rec}(v) \leq \\ 
    &2\max_{v\in V(G)} h(v)+1 \leq \\ 
    &2(4.5-\frac{2\log B}{\log N}) \left(1 + \frac{100}{\log \log (B)}\right)\frac{\log(N)}{\log(B)}+1 \leq \\ 
    &2(4.5-\frac{2\log B}{\log N})(1+\frac{\log B}{5\log N}) \left(1 + \frac{100}{\log \log (B)}\right)\frac{\log(N)}{\log(B)} \leq \\
    & (4.5 - \frac{2\log B^{0.5+1/\log^2 \log B}}{\log N}) \cdot \\
    &\left(1 + \frac{100}{\log \log \left(B^{0.5 + 1/\log^2\log(B)}\right)}\right)\frac{\log N}{\log \left(B^{0.5 + 1/\log^2\log(B)}\right)}
\end{align*}
Hence, the requirement on $h_{rec}$ in the recursive calls is satisfied. Therefore, from the first recursive AMIS call, by the induction hypothesis, we get a set $U^{(1)}_{IS}$ such that $|E(U^{(2)}| \leq \frac{M_U}{(B^{0.5+\eps_{B}})^2}$ for $U^{(2)} = U^{(1)} \setminus (U^{(1)}_{IS} \cup \Gamma_{G}(U^{(1)}_{IS}))$ that satisfies. Hence, the input requirements of the second recursive call to AMIS are also satisfied with the input $M^{(2)}_{U} = M_U/(B^{0.5+\eps_B})^2$. Thus from that call we get a set $U^{(2)}_{IS}$ such that $|E(U^{(2)} \setminus (U^{(2)}_{IS} \cup \Gamma_{G}(U^{(2)}_{IS})))| \leq \frac{M^{(2)}_U}{(B^{0.5 + \eps_B})^2}$. Hence, overall, for $U_{IS} = U^{(1)}_{IS} \cup U^{(2)}_{IS}$ and $U_{rem} = U\setminus (U_{IS} \cup \Gamma_{G}(U_{IS}))$ satisfies $\left|E(G[U_{rem}])\right| < \frac{M_U}{B^3} + \frac{M_U}{(B^{0.5+\eps_B})^4}< \frac{M_U}{B^2}$, where the last inequality holds as $B$ is lower bounded by a sufficiently large constant and $\eps_{B}=1/\log^2\log B$.

\medskip

\paragraph{Round complexity:} We now argue that that the \local-model round complexity of \cref{alg:almost_MIS} is in $\tilde{O}\left(\left(\frac{\log^2(n)}{k^2} + k \right)\log(B)\right) = \tilde{O}(\log^{2/3}(N)\log(B))$. Let $T(B)$ denote an upper bound on the round complexity of implementing $\mathcal{A}$ in the \local model. For the base case $B \leq N^{1/k}$, we directly get that $T(B) = \tilde{O}\left(\left(\frac{\log^2(n)}{k^2} + k \right) \log(B)\right)$ using the guarantees of \cref{lem:mis_base_case}.
	If $B > N^{1/k}$, the invocation of $\mathcal{A}_{L\ref{cor:mis_sampling}}$ takes $\tilde{O}(k \cdot \log(B))$ rounds.  Therefore, $T(B)$ satisfies the recurrence relationship
	
	\begin{align*}
	    &T(B) \leq \\
     &\begin{cases}
		\tilde{O}\left(\left(\frac{\log^2(n)}{k^2} + k \right) \log(B)\right) \text{ if $B \leq N^{1/k}$} \\
		2T\left(B^{0.5 + \frac{1}{\log^2\log(B)}} \right) + \tilde{O}(k \cdot \log(B)) \text{ if $B > N^{1/k}$}.
     \end{cases}
	\end{align*}
	
	Using this recurrence relationship, one can see that $T(B) = \tilde{O}\left(\left(\frac{\log^2(n)}{k^2} + k \right) \log(B)\right) $.
\end{proof}


\medskip

\begin{proof}[Proof of \cref{thm:mis_main}]
Invoke \cref{alg:almost_MIS} with inputs $U = V(G)$, $M_U = N^2$, $B = N$, and $h(v) = 1$ for every $v \in V(G)$. All the input guarantees are satisfied. Thus, $|E(U_{rem})| < \frac{M_U}{B^2} = 1$ and therefore, there is no edge in the subgraph induced by $E_{rem}$. This means we can output the union of $U_{IS}$ and $U_{rem}$ as a maximal independent set of $G$.
\end{proof}

\appendix
\section{Base Case Clusterings, with the Help of Techniques From \cite{ghaffari2023netdecomp}}
\label{sec:appendix_nd}
In this section, we prove \cref{thm:clusterone}, which we used inside the base case of our recursive network decomposition algorithm. The proof uses techniques developed in \cite{ghaffari2023netdecomp}. They developed these techniques to obtain an $\tilde{O}(\log^3(n))$ round deterministic network decomposition algorithm. Some of their techniques and results have to be slightly adapted to fit our needs, but these adaptations are relatively straightforward and we do not claim any novelty in this section. 

\clusterone
We use the following definition throughout this section.

\begin{definition}[$s$-Hop Frontier]
    Let $s \in \mathbb{R}_{\geq 0}$ and $h \colon V(G) \mapsto \mathbb{N}$. For every $u \in V(G)$, we define the $s$-hop frontier of $u$ (with respect to $h$) as the set
    \begin{align*}
        &F_{h,s}(u) := \\
        &\left\{v \in V(G) \colon d(u,v) - h(v) \leq \min_{v' \in V(G)} d(u,v') - h(v') + s\right\}.
    \end{align*}

\end{definition}

In particular, \cref{thm:clusterone} requires that for every node $u \in U$, the $\log^{100}\log(N)$-hop frontier of $u$ with respect to $h$ is at most $\log^{10}(N)$. \cref{thm:clusterone} is proven in two steps. In the first step, new head starts are computed such that for most nodes $u \in U$, the $\log^2\log(N)$-frontier of $u$ with respect to $h$ is at most $\poly(\log \log n)$. \cref{thm:clustertwo} captures this step. Afterwards, we use the new head starts to compute a subset $U^{sub} \subseteq U$ with the desired properties. \cref{thm:clusterthree} captures this step. We next state \cref{thm:clustertwo,thm:clusterthree} and then formally prove \cref{thm:clusterone}.

\begin{restatable}{theorem}{clustertwo}
    \label{thm:clustertwo}
    There exists an absolute constant $c > 0$ such that the following holds:
    Let $h \colon V(G) \mapsto \mathbb{N}_0$ with $\max_{v \in V(G)} h(v) = \tilde{O}(\log n)$ and $U \subseteq V(G)$ such that for every node $u \in U$, the $\log^{10}\log(N)$-hop frontier of $u$ with respect to $h$ has size at most $\log^{10}(N)$.

    There exists a deterministic distributed \local algorithm that computes in $\tilde{O}(\log n)$ rounds a subset $U^{sub} \subseteq U$ with $|U^{sub}| \geq 0.9|U|$ and $h_{out} \colon V(G) \mapsto \mathbb{N}_0$ with $\max_{v \in V(G)} h_{out}(v) = \tilde{O}(\log n)$  such that for every $u \in U^{sub}$, the $\log^2\log(N)$-hop frontier of $u$ with respect to $h_{out}$ has size at most $(\log\log(N))^c$.
\end{restatable}
\begin{restatable}{theorem}{clusterthree}
    \label{thm:clusterthree}
    The following holds for every absolute constant $c \geq 1$:
    Let $h \colon V(G) \mapsto \mathbb{N}_0$ with $\max_{v \in V(G)} h(v) = \tilde{O}(\log n)$ and $U \subseteq V(G)$ such that for every node $u \in U$, the $\log^2\log(N)$-hop frontier of $u$ with respect to $h$ has size at most $(\log \log (N))^c$.

    There exists a deterministic distributed \local algorithm that computes in $\tilde{O}(\log n)$ rounds a subset $U^{sub} \subseteq U$, $|U^{sub}| \geq 0.6|U|$, such that every connected component in $G[U^{sub}]$ has diameter $O(\log n)$.
\end{restatable}
\begin{proof}[Proof of Theorem~\ref{thm:clusterone}]
The proof is structured into two main steps, aligning with the applications of Theorem~\ref{thm:clustertwo} and Theorem~\ref{thm:clusterthree}:

\textbf{Step 1: Reducing Frontier Size using Theorem~\ref{thm:clustertwo}.}
We begin with the set \(U\) and the function \(h\) as given in Theorem~\ref{thm:clusterone}. By applying Theorem~\ref{thm:clustertwo}, we compute a subset \(U^{sub}_{T\ref{thm:clustertwo}} \subseteq U\), $|U^{sub}_{T\ref{thm:clustertwo}}| \geq 0.9|U|$ and a new function \(h_{out}\) such that for every node \(u \in U^{sub}_{T\ref{thm:clustertwo}}\), the \(\log^2\log(N)\)-hop frontier of \(u\) with respect to \(h_{out}\) has size at most \((\log\log(N))^c\). This step takes $\tilde{O}(\log n)$ rounds. 

\textbf{Step 2: Clustering with Small Diameter using Theorem~\ref{thm:clusterthree}.}
With \(U^{sub}_{T\ref{thm:clustertwo}}\) and \(h_{out}\) obtained from the previous step, we then apply Theorem~\ref{thm:clusterthree}. This application yields a subset \(U^{sub}_{T\ref{thm:clusterthree}} \subseteq U^{sub}_{T\ref{thm:clustertwo}}\) such that every connected component in \(G[U^{sub}_{T\ref{thm:clusterthree}}]\) has diameter \(O(\log n)\). Theorem~\ref{thm:clusterthree} guarantees that \(|U^{sub}_{T\ref{thm:clusterthree}}| \geq 0.6|U^{sub}_{T\ref{thm:clustertwo}}|\). This step takes $\tilde{O}(\log n)$ rounds.

\textbf{Conclusion:}
Both steps together take $\tilde{O}(\log n)$ rounds. Moreover,  every connected component in \(G[U^{sub}_{T\ref{thm:clusterthree}}]\) has diameter \(O(\log n)\) and 

\[|U^{sub}_{T\ref{thm:clusterthree}}| \geq 0.9 U^{sub}_{T\ref{thm:clustertwo}} \geq 0.9 \cdot 0.6|U| \geq 0.5|U|.\]
\end{proof}

\subsection{Preliminaries}
\cref{def:clustering_further_definitions} will be used in \cref{sec:clusterthree}.
        \begin{definition}[$s$-separated clustering, $s$-hop degree]
	\label{def:clustering_further_definitions}
        For a given $s \in \mathbb{N}$ and clustering $\fC$, we say that the clustering $\fC$ is $s$-separated if for any two clusters $C,C' \in \fC$ with $C \neq C'$, we have $d(C,C') \geq s$. For a node $u$, we say that the $s$-hop degree of $u$ with respect to $\fC$ is the total number of clusters in $\fC$ within distance $s$ of $u$. The $s$-hop degree of the clustering $\fC$ is the largest $s$-hop degree of any clustered node.
	\end{definition}

The following result follows as a simple corollary of \cref{lem:subsampling_multiple}.

  \begin{corollary}
	\label{cor:subsampling_multiple_2}
                There exists an absolute constant $c > 0$ such that the following holds:
			Let $H$ be a bipartite graph with a given bipartition $V(H) = U_H \sqcup V_H$, where each vertex in $V$ has a unique identifier from $\{1,2,\ldots,N\}$ for some $N \in \mathbb{N}$.
	There exists a deterministic distributed \LOCAL algorithm that computes in $\poly(\log \log N)$ rounds a subset $V^{sub} \subseteq V_H$ such that for 

    \begin{align*}
        U^{bad} := &\bigg\{u \in U \colon |\Gamma_H(u)| \geq (\log\log(N))^c \\ &\text{ and } \frac{|\Gamma_H(u) \cap V^{sub}|}{|\Gamma_H(u)|} \notin [1/3,2/3]\bigg\},
    \end{align*}

    it holds that $|U^{bad}| \leq \frac{|U_H|}{\log^{10}\log(N)}$.
\end{corollary}

\subsection{Proof of \cref{thm:clustertwo}}
\label{sec:two}
We will now use \cref{cor:subsampling_multiple_2} to prove \cref{thm:clustertwo}, which is restated below.
\clustertwo*

The algorithm computes a sequence of triples 
\begin{align*}
    &(V^{actice}_0,U_0,h_0),(V^{actice}_1,U_1,h_1),\ldots, \\&(V^{active}_{\lceil\log^2\log(N) \rceil},U_{\lceil\log^2\log(N) \rceil},h_{\lceil\log^2\log(N) \rceil}).
\end{align*}
Initially, $V^{active}_0 = V(G)$, $U_0 = U$ and $h_0 = h$. In the end, the algorithm outputs $U^{sub} = U_{\lceil\log^2\log(n) \rceil}$ and $h_{out} = h_{\lceil\log^2\log(n) \rceil}$.

Fix some $i \in [\lceil \log^2\log(N)\rceil]$. We next explain how the algorithm computes $(V^{actice}_i,U_i,h_i)$ given $(V^{actice}_{i-1},U_{i-1},h_{i-1})$. Let $H$ be the bipartite graph where one side of the bipartition contains two copies for every node $u \in U_{i-1}$ and the other side of the bipartition is $V^{active}_{i-1}$. The first copy of every node $u \in U_{i-1}$ is connected to every node $v \in F_{h,\log^{10}\log(N)}(u) \cap V^{active}_{i-1}$. The second copy of every node $u \in U_{i-1}$ is connected to every node in $F_{h_{i-1},\log^2\log(N)}(u) \cap V^{active}_{i-1}$.

Using the algorithm of \cref{cor:subsampling_multiple_2}, we can compute in $\poly(\log \log n)$ rounds in $H$ a subset $V^{active}_i \subseteq V^{active}_{i-1}$ such that for
\fullOnly{
\begin{align*}U^{bad}_i :=& \left\{u \in U_{i-1} \colon |F_{h,\log^{10}\log(N)}(u) \cap V^{active}_{i-1}| \geq \log^{c}\log(N) \text{ and } \frac{|F_{h,\log^{10}\log(N)}(u) \cap V^{active}_i|}{|F_{h,\log^{10}\log(N)}(u) \cap V^{active}_{i-1}|} > 2/3.\right \} \\
&\bigcup \left\{u \in U_{i-1} \colon |F_{h_{i-1},\log^2\log(N)}(u) \cap V^{active}_{i-1}| \geq \log^{c}\log(N) \text{ and } F_{h_{i-1},\log^2\log(N)}(u) \cap V^{active}_i = \emptyset\right \}.\end{align*}
}
\shortOnly{
\begin{align*}
&U^{bad}_i := \\ 
& \{u \in U_{i-1} \colon |F_{h,\log^{10}\log(N)}(u) \cap V^{active}_{i-1}| \geq \log^{c}\log(N) \\
&\quad \text{and } \frac{|F_{h,\log^{10}\log(N)}(u) \cap V^{active}_i|}{|F_{h,\log^{10}\log(N)}(u) \cap V^{active}_{i-1}|} > 2/3\} \; \cup \\
& \{u \in U_{i-1} \colon |F_{h_{i-1},\log^2\log(N)}(u) \cap V^{active}_{i-1}| \geq \log^{c}\log(N) \\
&\quad \text{and } F_{h_{i-1},\log^2\log(N)}(u) \cap V^{active}_i = \emptyset\}.
\end{align*}

}

it holds that $|U^{bad}_i| \leq \frac{2|U_{i-1}|}{\log^{10}\log(N)}.$ We set $U_i = U_{i-1} \setminus U^{bad}_i$ and for every $v \in V(G)$,

\[h_i(v) = h_{i-1}(v) + \1_{\{v \in V^{active}_i\}}\lceil10\log^2\log(N)\rceil.\]

\begin{claim}
    \label{cl:cluster1}
    For every $i \in \{0,1,\ldots,\lceil \log^2\log(N)\rceil\}$ and every $u \in U_i$, it holds that 

    \begin{align*}
        &|F_{h,\log^{10}\log(N)}(u) \cap V^{active}_i| \leq \\ &\max\left((\log\log (N))^c, (2/3)^i|F_{h,\log^{10}\log(N)}(u)|\right).
    \end{align*}
    In particular, for every $u \in U^{sub}$, it holds that

    \[\left|F_{h,\log^{10}\log(N)}(u) \cap V^{active}_{\lceil \log^2\log(N)\rceil}\right| \leq \left(\log\log(N)\right)^c.\] 
\end{claim}
\begin{proof}
The statement holds for $i = 0$. Now, consider some fixed $i \in [\lceil \log^2\log(N)\rceil]$ and some $u \in U_i$. In particular, $u \in U_{i-1}$ and $u \notin U^{bad}_i$. As $u \in U_{i-1}$, we can inductively assume that 

\begin{align*}
    &|F_{h,\log^{10}\log(N)}(u) \cap V^{active}_{i-1}| \leq \\ &\max\left((\log\log (N))^c, (2/3)^{i-1}|F_{h,\log^{10}\log(N)}(u)|\right).
\end{align*}
As $u \notin U^{bad}_i$, we get

\begin{align*} &|F_{h,\log^{10}\log(N)}(u) \cap V^{active}_{i}|  \leq \\ &\max(\log^c\log(N), (2/3) |F_{h,\log^{10}\log(N)}(u) \cap V^{active}_{i-1}| ).\end{align*}

Combining both inequalities, we indeed get

\begin{align*}&|F_{h,\log^{10}\log(N)}(u) \cap V^{active}_i| \leq \\ & \max\left((\log\log (N))^c, (2/3)^i|F_{h,\log^{10}\log(N)}(u)|\right).\end{align*}
 Now, consider some $u \in U^{sub} := U_{\lceil \log^2\log(N)\rceil}$. Using the input assumption $|F_{h,\log^{100}\log(N)}(u)| \leq \log^{100}(N)$, we get
 
\begin{align*}
    &|F_{h,\log^{10}\log(N)}(u) \cap V^{active}_{\lceil \log^2\log(N)\rceil}| \\ 
    &\leq \max\left((\log\log (N))^c, (2/3)^{\lceil \log^2\log(N)\rceil}|F_{h,\log^{10}\log(N)}(u)|\right) \\
    &\leq \max\left((\log\log (N))^c, (2/3)^{\lceil \log^2\log(N)\rceil}\log^{100}(N)\right) \\
    &= (\log\log (N))^c.
 \end{align*}
\end{proof}

\begin{claim}
    \label{cl:cluster2}
    For every $i \in \{0,1,\ldots,\lceil \log^2\log(N)\rceil\}$ and every $u \in U_i$, it holds that
    \[|F_{h_i,\log^2\log(N)}(u) \setminus V^{active}_i| \leq i \cdot (\log\log (N))^c.\]
\end{claim}
\begin{proof}
The statement holds for $i = 0$ as $V^{active}_0 = V(G)$. Now, consider some $i \in [\lceil \log^2\log(N)\rceil]$ and $u \in U_i$. In particular, $u \in U_{i-1}$ and $u \notin U^{bad}_i$. As $u \in U_{i-1}$, we can inductively assume that

   \[|F_{h_{i-1},\log^2\log(N)}(u) \setminus V^{active}_{i-1}| \leq (i-1) \cdot (\log\log (N))^c.\]

As $u \notin U^{bad}_i$, at least one of the following holds:

\begin{enumerate}
    \item $|F_{h_{i-1},\log^2\log(N)}(u) \cap V^{active}_{i-1}| \leq \log^{c}\log(N)$ 
    \item $F_{h_{i-1},\log^2\log(N)}(u) \cap V^{active}_i \neq \emptyset$
\end{enumerate}
    If $|F_{h_{i-1},\log^2\log(N)}(u) \cap V^{active}_{i-1}| \leq \log^{c}\log(N)$, then

    \begin{align*}
        & &&|F_{h_i,\log^2\log(N)}(u) \setminus V^{active}_i| \\
        &\leq &&|F_{h_{i-1},\log^2\log(N)}(u)| &&\\
        &= &&|F_{h_{i-1},\log^2\log(N)}(u) \setminus V^{active}_{i-1}| \\ & &&+ |F_{h_{i-1},\log^2\log(N)}(u) \cap V^{active}_{i-1}| \\
        &\leq &&(i-1) \cdot (\log\log (N))^c + \log^{c}\log(N) 
        \\ &\leq &&i \cdot (\log\log (N))^c. 
    \end{align*}
    On the other hand, if $F_{h_{i-1},\log^2\log(N)}(u) \cap V^{active}_i \neq \emptyset$, then $F_{h_i,\log^{2}\log(N)}(u) \setminus V^{active}_i = \emptyset$. 
\end{proof}
The claim below directly follows from the fact that $|U^{bad}_i| \leq \frac{2}{\log^{10}\log(N)}|U|$ for every $i \in [\lceil \log^2\log(N)\rceil]$.
\begin{claim}
    \label{cl:cluster3}
     For every $i \in \{0,1,\ldots,\lceil \log^2\log(N)\rceil\}$, it holds that

     \[|U_i| \geq \left(1 - \frac{2i}{\log^{10}\log(n)}\right) |U|.\]
     In particular,

     \[|U^{sub}| \geq 0.9|U|.\]
\end{claim}
The claim below directly follows from the fact that $h_{out}(v) \in [h(v),h(v) + \lceil \log^2\log(N)\rceil \lceil \log^{10}\log(N)\rceil]$ for every $v \in V(G)$.
\begin{claim}
    \label{cl:cluster4}
    For every $u \in U$, we have $F_{h_{out},\log^2\log(N)}(u) \subseteq F_{h,\log^{10}\log(N)}(u)$.
\end{claim}
\begin{proof}
\end{proof}
\begin{claim}
    \label{cl:cluster5}
For every $u \in U^{sub}$, the $\log^2\log(N)$-hop frontier of $u$ with respect to $h_{out}$ has size at most $\poly(\log \log n)$.
\end{claim}
\begin{proof}
    Consider some $u \in U^{sub}$.
    Using \cref{cl:cluster4}, we have

    \begin{align*}
    |F_{h_{out},\log^2\log(N)}(u)| \leq &|F_{h_{out},\log^2\log(N)}(u) \setminus V^{active}_{\lceil \log^2\log(N)\rceil}| \\ + &\left|F_{h,\log^{10}\log(N)}(u) \cap V^{active}_{\lceil \log^2\log(N)\rceil}\right|.\end{align*}
    \cref{cl:cluster1} gives
   \[\left|F_{h,\log^{10}\log(N)}(u) \cap V^{active}_{\lceil \log^2\log(N)\rceil}\right| \leq \left(\log\log(N)\right)^c.\] 
   \cref{cl:cluster2} gives
   \[|F_{h_{out},\log^2\log(N)}(u) \setminus V^{active}_{\lceil \log^2\log(N)\rceil}| \leq \poly(\log \log n).\]
   
   Therefore, the $\log^2\log(N)$-hop frontier of $u$ with respect to $h_{out}$ has size at most $\poly(\log \log n)$.
\end{proof}
\begin{claim}
    \label{cl:cluster6}
    The algorithm can be implemented in $\tilde{O}(\log n)$ rounds.
\end{claim}
\begin{proof}
    In each of the $\poly(\log \log n)$ rounds, we have to simulate $\poly(\log \log n)$ rounds in some bipartite graph $H$. Any edge in $H$ is between two nodes in $G$ of distance at most $\tilde{O}(\log n)$, and therefore each round in $H$ can be simulated in $\tilde{O}(\log n)$ rounds in $G$. 
\end{proof}

\subsection{Proof of \cref{thm:clusterthree}}
\label{sec:clusterthree}
In this section, we prove \cref{thm:clusterthree}, which we restate below for convenience.
\clusterthree*

\cref{thm:clusterthree} follows as a simple corollary of \cref{thm:lowdegreetohalf}.

\begin{restatable}{theorem}{lowdegreetohalf}
\label{thm:lowdegreetohalf}
	For a given $\separation \in \mathbb{N}$ with $s \geq 100$ and $DEG$, assume we are given a clustering $\fC$ with 
		
		\begin{enumerate}
			\item weak-diameter $\tilde{O}(\separation \log (n))$ and
			\item $\separation$-hop degree of at most $DEG$.
        \end{enumerate}
		Assuming that $DEG \leq \frac{1.01^s}{100}$, there exists a deterministic distributed \local algorithm running in $\tilde{O}(\separation \cdot DEG \cdot \log(n))$ rounds which computes a clustering $\fC^{out}$ with
		
		\begin{enumerate}
			\item strong-diameter $O(\log n)$,
			\item separation of $2$, and
			\item the number of clustered nodes is at least $0.6|\fC|$.
		\end{enumerate} 
              The clustering $\fC^{out}$ only clusters nodes that are also clustered in $\fC$.
\end{restatable}
\begin{proof}[Proof of \cref{thm:clusterthree}]
Let $\fC_h$ be the clustering one obtains by clustering every node $u \in U$ to the node $v \in V(G)$ that minimizes $d(u,v) - h(v)$, favoring nodes with larger identifiers.  The strong-diameter of $\fC_H$ can be upper bounded by $\max_{v \in V(G)} h(v)$, and is therefore $\tilde{O}(\log n)$. For every node $u \in U$, the $\log^2\log(N)$-hop frontier of $u$ with respect to $h$ has size at most $(\log \log(N))^c$. A simple calculation shows that the size of the $s$-hop frontier of $u$ with respect to $h$ is at least the $s/3$-hop cluster degree of $u$ with respect to $\fC_h$. Therefore, any node $u \in U$ has a $\log^2\log(N)/3$-hop cluster degree of at most $(\log \log(N))^c$. Now, let $\fC$ be the clustering we obtain from the partition of $\fC_h$ by only keeping the nodes in $U$ in each cluster. Then, the clustering $\fC$ has a $\log^2\log(N)/3$-hop cluster degree of at most $(\log \log(N))^c$ and weak-diameter $\tilde{O}(\log n)$. Thus, we can use \cref{thm:lowdegreetohalf} with $s = \lfloor\log^2\log(N)/3 \rfloor$ and $DEG = (\log \log(N))^c$ to compute in $\tilde{O}(\log n)$ rounds a clustering $\fC^{out}$ with 
\begin{enumerate}
			\item strong-diameter $O(\log n)$,
			\item separation of $2$,
			\item the number of clustered nodes is at least $\frac{|\fC|}{2}$,
		\end{enumerate} 
and $\fC^{out}$ only clusters nodes that are also clustered in $\fC$. Therefore, we can simply define $U^{sub}$ as the set of nodes that are clustered in $\fC^{out}$.
\end{proof}
\paragraph{Roadmap}
We prove \cref{thm:lowdegreetohalf} in \cref{sec:lowdegreetohalf}. Before, we prove two lemmas that we use for \cref{thm:lowdegreetohalf}. \cref{sec:cluster_subsampling} is dedicated to prove \cref{lemma:subsampling_main}. \cref{sec:LowDegtoSmallBoundary}  is dedicated to prove \cref{lemma:clusteringsmallboundary}.
\subsubsection{Cluster Subsampling}
\label{sec:cluster_subsampling}
 	\begin{restatable}{lemma}{subsamplingmain}
		\label{lemma:subsampling_main}
		For a given $\separation, DEG \in \mathbb{N}$, assume we are given a clustering $\fC$ with 
		
		\begin{enumerate}
			\item weak-diameter $\tilde{O}(\separation \log (n))$ and
			\item $\separation$-hop degree of at most $DEG$.
        \end{enumerate}
		There exists a deterministic distributed \local algorithm running in $\tilde{O}(\separation \log^2(DEG)\log(n))$ rounds which computes a clustering $\fC^{out}$ with
		
		\begin{enumerate}
			\item weak-diameter $\tilde{O}(\separation \log (n))$,
			\item separation of $\separation$, and
			\item the number of clustered nodes is at least $\frac{|\fC|}{8 DEG}$.
		\end{enumerate} 
            The clustering $\fC^{out}$ only clusters nodes that are also clustered in $\fC$.
	\end{restatable}
\begin{proof}
    The clustering $\fC^{out}$ is computed in two steps.
    In the first step, we compute a clustering $\fC'$ which one obtains from $\fC$ by only keeping some of the clusters in $\fC$ (any such cluster is kept in its entirety). Intuitively, we derandomize the random process which would include each cluster $C$ from $\fC$ in the clustering $\fC'$ with probability $p = \frac{1}{2DEG}$, pairwise independently.
        Given the clustering $\fC'$, we keep each node $u \in \fC'$ clustered in $\fC^{out}$ if and only if no other cluster within distance $s$ of $u$ has been selected. Hence, the resulting clustering is indeed $s$-separated.
		Note that given $\fC'$, the output clustering $\fC^{out}$ can be computed in $\tilde{O}(\separation \log n)$ rounds. The property that the resulting clustering has weak-diameter $\tilde{O}(s \log (n))$ directly follows from the fact that the input clustering has weak-diameter $\tilde{O}(s \log (n))$.

		Thus, it remains to show that we can compute $\fC'$ in such a way that $\fC^{out}$ clusters at least $\frac{|\fC|}{8 DEG}$ many nodes. As mentioned before, we use the local rouding framework of \cite{faour2022local}, as reviewed in \Cref{subsec:prelim-localRounding}.

    The labeling space is for each cluster whether it is contained in $\fC'$ or not, i.e., each cluster takes simply one of two possible labels $\{0,1\}$ where $1$ indicates that the corresponding cluster is in $\fC'$. For a given label assignment $\vec{x} \in \{0,1\}^{\fC}$, we define 
    \[\utility(\vec{x}) = \sum_{C \in \fC} |C|x_C\]
    and    
    
    \[\cost(\vec{x}) = \sum_{\text{$u$ is clustered in $\fC$}}\sum_{C' \in \fC \colon d(u,C') \leq \separation \text{ and } u \notin C'} x_{C}x_{C'}.\]

    If the label assignment is relaxed to be a fractional assignment $\vec{x} \in [0,1]^{\fC}$, where intuitively now $x_C$ is the probability of $C$ being contained in $\fC'$, the same definitions apply for the utility and cost of this fractional assignment.

    To capture the cost function as a summation of costs over edges, we define an auxiliary multi-graph $H$ as follows: Each cluster in $\fC$ corresponds to a vertex in $H$ and for each node $u$ clustered in $\fC$ and every cluster $C'$ with $d(u,C') \leq s$ and $u \notin C'$, we add an auxiliary edge between the cluster of $u$ and $C'$ with a cost function which is equal to $1$ when both $C$ and $C'$ are contained in $\fC'$, and zero otherwise.

    \begin{claim}
    Let $\vec{x} \in [0,1]^{\fC}$ with $x_C = \frac{1}{2DEG}$ for every $C \in \fC$. Then, we have $\utility(\vec{x}) - \cost(\vec{x}) \geq \utility(\vec{x})/2$.
    \end{claim}
    \begin{proof}
    We have
    \[\utility(\vec{x}) = \sum_{C \in \fC} |C|x_C = \sum_{C \in \fC} \frac{|C|}{2DEG} = \frac{|\fC|}{2DEG},\]
    and 
    \begin{align*}
        \cost(\vec{x}) &= \sum_{\text{$u$ is clustered in $\fC$}}\sum_{C' \in \fC \colon d(u,C') \leq \separation \text{ and } u \notin C'} x_{C}x_{C'} \\
        &\leq \sum_{\text{$u$ is clustered in $\fC$}} \frac{1}{4DEG} \frac{DEG}{DEG} \leq \frac{|\fC|}{4DEG}.
    \end{align*}

    Therefore, indeed $\utility(\vec{x}) - \cost(\vec{x}) \geq \utility(\vec{x})/2$.
    \end{proof}
    We now invoke the rounding of \cref{lemma:rounding} with parameters $\mu=0.5$ and $\eps=0.5$ on the fractional label assignment of $\vec{x} \in [0,1]^{\fC}$ where $x_C = \frac{1}{2DEG}$ for every $C \in \fC$. The rounding takes $O(\log^*(n) + \log^2(DEG))$ rounds in $H$, and therefore can be simulated in $\tilde{O}(s \log^2(DEG)\log n)$ rounds in the base graph $G$. Hence,  as output we get an integral label assignment $\vec{y} \in \{0,1\}^{\fC}$ which satisfies
    \[\utility(\vec{y}) - \cost(\vec{y}) \geq  0.5 (\utility(\vec{x}) - \cost(\vec{x})) \geq  \frac{|\fC|}{8DEG}.\]

    Let $C' = \{C \in \fC \colon y_{C} = 1\}$. Note that for every $u \in \fC$,
    \[\1_{\{\text{$u$ is clustered in $\fC^{out}$}\}} \geq y_{C} - \sum_{C' \in \fC \colon d(u,C') \leq \separation \text{ and } u \notin C'} y_{C} y_{C'}.\]
    Therefore,
    \begin{align*}
        |\fC^{out}|&\geq \sum_{\text{$u$ is clustered in $\fC$}} \1_{\{\text{$u$ is clustered in $\fC^{out}$}\}} \\
        &\geq \sum_{\text{$u$ is clustered in $\fC$}} \left( y_{C} - \sum_{C' \in \fC \colon d(u,C') \leq \separation \text{ and } u \notin C'} y_{C} y_{C'}\right)\\
        &= \utility(\vec{y}) - \cost(\vec{y}) \\
        &\geq \frac{|\fC|}{8DEG}.
    \end{align*}
    and therefore $\fC^{out}$ clusters enough vertices to prove \Cref{lemma:subsampling_main}.

    \end{proof}

\subsubsection{Clustering With Small Boundary}
\label{cluster:small_boundary}
Next, we use a simple ball-growing argument to ensure that the boundary of our output clustering is small compared to its size.
 \label{sec:LowDegtoSmallBoundary}
 \begin{restatable}{lemma}{clusteringsmallboundary}
    \label{lemma:clusteringsmallboundary}
	For a given $\separation \in \mathbb{N}$ with $s \geq 100$ and $DEG$, assume we are given a clustering $\fC$ with 
		
		\begin{enumerate}
			\item weak-diameter $\tilde{O}(\separation \log (n))$ and
			\item $\separation$-hop degree of at most $DEG$.
        \end{enumerate}
		Assuming that $DEG \leq \frac{101^s}{100}$, there exists a deterministic distributed \local algorithm running in $\tilde{O}(\separation \log^2(DEG)\log(n))$ rounds which computes a clustering $\fC^{out}$ with
		
		\begin{enumerate}
			\item weak-diameter $\tilde{O}(\separation \log (n))$,
			\item separation of $2$,
			\item the number of clustered nodes is at least $\frac{|\fC|}{16DEG}$ and
                \item the total number of nodes in $\fC$ that are not clustered in $\fC^{out}$ and have a neighbor in $\fC^{out}$ is at most $|\fC^{out}|/10$.
		\end{enumerate} 
        The clustering $\fC^{out}$ only clusters nodes that are also clustered in $\fC$.
\end{restatable}
\begin{proof}
We first use the algorithm of \cref{lemma:subsampling_main} to compute in $\tilde{O}(\separation \log^2(DEG)\log(n))$ rounds a clustering $\fC'$ with
		
		\begin{enumerate}
			\item weak-diameter $\tilde{O}(\separation \log (n))$,
			\item separation of $\separation$ and
			\item the number of clustered nodes is at least $\frac{|\fC|}{8 DEG}$.
		\end{enumerate} 
The clustering $\fC'$ only clusters nodes that are also clustered in $\fC$.
For each cluster $C' \in \fC'$ and $k \in \mathbb{N}$, we define 

\[C'_{\leq k} = \{v \in V(G) \colon d_G(v,C') \leq k \text{ and $v$ is clustered in $\fC$}\}.\]

For each cluster $C' \in \fC'$, if there exists $k \in \{0,1,\ldots,\lfloor s/3\rfloor\}$ with $|C'_{\leq k+1}| \leq 1.1|C'_{\leq k}|$, then we add $C'_{\leq k}$ to the output clustering $\fC^{out}$. This is also the only cluster that we add to $\fC^{out}$. The only output property that is not immediate is that the number of clustered nodes in $\fC^{out}$ is at least $\frac{|\fC|}{16DEG}$. Let $\fC^{bad}$ contain each cluster $C' \in \fC'$ such that there does not exists $k \in \{0,1,\ldots,\lfloor s/3\rfloor\}$ with $|C'_{\leq k+1}| \leq 1.1|C'_{\leq k}|$. For each $C^{bad} \in \fC^{bad}$, we have $|C^{bad}| \leq \left(1/1.1\right)^{\lfloor s/3\rfloor}|C^{bad}_{\leq \lfloor s/3\rfloor}|$. Therefore, the total number of clustered nodes in $\fC_{bad}$ is at most 

\[\left(1/1.1\right)^{\lfloor s/3\rfloor}|\fC| \leq \frac{|\fC|}{16DEG}.\]
Together with the property that the number of nodes clustered in $\fC'$ is at least $\frac{|\fC|}{8 DEG}$, this directly gives that the total number of nodes clustered in  $\fC^{out}$ is at least $\frac{|\fC|}{16DEG}$. 
\end{proof}

\subsubsection{Proof of \cref{thm:lowdegreetohalf}}
\label{sec:lowdegreetohalf}
We are now ready to prove \cref{thm:lowdegreetohalf}, which we restate below.
\lowdegreetohalf*
\begin{proof}
The algorithm computes a sequence of pairs of clusterings $(\fC_0^{out},\fC_0)$, $(\fC_1^{out},\fC_1)$, $\ldots$, \\ $(\fC_{1000DEG}^{out},\fC_{1000DEG})$, where $\fC_0^{out} = \emptyset$ and $\fC_0 = \fC$. For every $i \in \{0,1,\ldots,1000DEG - 1\}$, we compute $(\fC_{i+1}^{out},\fC_{i+1})$ from $(\fC^{out}_{i}, \fC_i)$ as follows: We first invoke the algorithm of \cref{lemma:clusteringsmallboundary} with input clustering $\fC_{i}$ to compute in $\tilde{O}(\separation \log^2(DEG)\log(n))$ rounds a clustering $\fC'$. The clustering $\fC^{out}_{i+1}$ then contains all clusters in $\fC^{out}_i$ and all clusters in $\fC'$. Finally, we obtain $\fC_{i+1}$ from $\fC_i$ by removing all nodes that are contained in $\fC'$ or have a neighbor that is in $\fC'$. The guarantees of \cref{lemma:clusteringsmallboundary} give that 

\[|\fC^{out}_{i+1}| \geq |\fC^{out}_{i}| + \frac{|\fC_i|}{16DEG}\] 
and 
\[|\fC_i| - |\fC_{i+1}| \leq 0.1(|\fC^{out}_{i+1}| - |\fC^{out}_{i}|).\]

A simple induction therefore gives that $|\fC_{1000DEG}^{out}| \geq (2/3)|\fC|$. Moreover, $\fC_{1000DEG}^{out}$ has a weak-diameter of $\tilde{O}(s \log n)$ and a separation of $2$. Finally, we obtain our output clustering $\fC^{out}$ as follows: for each cluster $C \in \fC_{1000DEG}^{out}$, there exists $U \subseteq C$ with $|U| \geq 0.99|C|$ such that each connected component of $G[U]$ has diameter $O(\log n)$. We add one cluster corresponding to each connected component in $G[U]$ to $\fC_{out}$. Thus, $\fC_{out}$ has strong-diameter $O(\log n)$ and clusters at least $0.99\frac{2}{3}|\fC| \geq 0.6|\fC|$ many nodes.
\end{proof}

\section{Base Case IS, with the Help of Techniques From \cite{ghaffarigrunau2023fasterMIS}}
\label{sec:appendix_is}

\label{sec:mis_preliminaries}
In this section, we prove \cref{thm:independent_set_given_partition}, which we used in the base case of our recursive MIS algorithm. \cref{thm:independent_set_given_partition} is a slightly more general version of a result proven in \cite{ghaffarigrunau2023fasterMIS}. The proof can be adapted and we do not claim any novelty in this section.

\is

We use the following definition throughout this section.

\begin{definition}[$q$-integral vector]
Let $S$ be a finite set. A vector $x \in \mathbb{R}^S$ is called $q$-integral for some $q \in \mathbb{R}$ if for every $s \in S$, it holds that $x_s \in \{0\} \cup [q, 1]$.
\end{definition}

\subsection{Partial Rounding Lemma}
\label{sec:partial_rounding}
One ingredient of \cref{thm:independent_set_given_partition} is a probabilistic method argument. This argument is captured in \cref{lemma:accelerated_probabilistic_method}.

\begin{lemma}[Partial Rounding Lemma]
\label{lemma:accelerated_probabilistic_method}
Let $\fS$ be a finite collection of sets with $S \subseteq S_{ground}$ for every $S \in \fS$, where $S_{ground}$ is some finite ground set. Let $x \in [0,1]^{S_{ground}}$. For a given $q \in (0,1]$, there exists a (computable) $q$-integral $y \in [0,1]^{S_{ground}}$ such that for every $S \in \fS$, it holds that

\[\left|\sum_{j \in S} x_j - \sum_{j \in S} y_j\right| \leq 0.01\sum_{j \in S} x_j + 10^6q\log(2|\fS|).\]
\end{lemma}
It follows by a simple application of the Chernoff bound, which we restate here for convenience:
\begin{theorem}[Chernoff Bound]
    \label{thm:accelerated_rounding_chernoff}
    Let $X := \sum_{i \in [n]} X_i$, where $X_i$, $i \in [n]$, are independently distributed and $0 \leq X_i \leq 1$. Then, for a given $\delta > 0$, we have
    \[Pr[|X - \E[X]| \geq \delta\E[X]] \leq 2e^{-\frac{\min(\delta,\delta^2)\E[X]}{3}}.\]
\end{theorem}

We are now ready to prove \cref{lemma:accelerated_probabilistic_method}

\begin{proof}[Proof of \cref{lemma:accelerated_probabilistic_method}]
First, note that we can assume without loss of generality that $x_j < q$ for every $j \in S_{ground}$, which we will assume from now on. Let $(Y_j)_{j \in S_{ground}}$ be a collection of fully independent $0/1$ variables with $\Pr[Y_j = 1] = x_j/q$ for every $j \in S_{ground}$. Fix some $S \in \fS$. Let $Z_S = \sum_{j \in S} Y_j$ and note that

\[\mathbb{E}[Z_S] = \sum_{j \in S} \mathbb{E}[Y_j] = \frac{1}{q}\sum_{j \in S} x_j.\]

Also, let

\[\delta_S = \frac{0.01\sum_{j \in S} x_j + 10^6q\log(2|\fS|)}{\sum_{j \in S} x_j} \geq 0.01.\]

By applying the Chernoff Bound from Theorem \ref{thm:accelerated_rounding_chernoff} to $Z_S$, we get
\begin{align*}
&\Pr\left[\left|qZ_S - \sum_{j \in S} x_j\right| \geq 0.01\sum_{j \in S} x_j + 10^6q\log(2|\fS|)\right] \\
&= \Pr\left[\left|Z_S - \frac{1}{q}\sum_{j \in S} x_j\right| \geq \delta_S \mathbb{E}[Z_S]\right] \\
&\leq 2e^{-\frac{\min(\delta_S, \delta_S^2) \mathbb{E}[Z_S]}{3}} 
\leq 2e^{\frac{-0.01\delta_S \mathbb{E}[Z_S]}{3}}  < \frac{1}{|\fS|}.
\end{align*}
Thus, with positive probability, for every $S \in \fS$, it holds that

\[\left|qZ_S - \sum_{j \in S} x_j\right| \leq 0.01\sum_{j \in S} x_j + 10^6q\log(2|\fS|).\] 

We define the vector $y$ by setting $y_j = q$ if $Y_j = 1$, and $y_j = 0$ otherwise, for each $j \in \fS$. In particular, $y$ is $q$-integral. Moreover, the claim's condition holds with non-zero probability, which proves the existence of such a vector $y$.
\end{proof}

\subsection{Proof of \cref{thm:independent_set_given_partition}}
We now prove \cref{thm:independent_set_given_partition}. To keep the notation simple, we assume that $U = V(G)$ in the proof below. 

\paragraph{Luby's Randomized MIS Algorithm.} The starting point is to recall Luby's classic randomized algorithm from \cite{luby86}. 
The first round of Luby's algorithm works as follows:
We mark each node $u \in V(G)$ with probability $1/(20 \deg(u))$. Then, for each edge $\{u,v\} \in E(G)$, let us orient the edge as $u\rightarrow v$ if and only if $\deg(u)<\deg(v)$ or $\deg(u)=\deg(v)$ and $ID(u)<ID(v)$. For each marked node $u$, we add $u$ to the independent set if and only if $v$ does not have a marked out-neighbor. It is well-known that removing all nodes in the independent set along with its neighbors results in a constant fraction of edges being removed, in expectation. We next show how to get this guarantee in a deterministic manner.

\paragraph{Derandomizing Luby via Rounding.} We denote by $\vec{x} \in \{0, 1\}^{V(G)}$ the indicator vector of whether different nodes are marked, that is, we have $x_v=1$ if $v$ is marked and $x_v=0$ otherwise. Let $R_v(\vec{x})$ be the indicator variable of the event that $v$ gets removed, for the marking vector $\vec{x}$. Let $Z(\vec{x})$ be the corresponding number of removed edges. Luby's algorithm determines the markings $\vec{x}$ randomly. Our task is to derandomize this and select the marked nodes in a deterministic way such that when we remove nodes added to the independent set (those marked nodes that do not have a marked out-neighbor) and their neighbors, along with all the edges incident on these nodes, at least a constant fraction of edges gets removed.

\paragraph{Good and bad nodes and prevalence of edges incident on good nodes.} We call any node $v \in V(G)$ \emph{good} if and only if it has at least $\deg(v)/3$ incoming edges. A node $v$ that is not good is called bad. It holds that
\begin{align}
    \label{eq:luby0a}
\sum_{\textit{good vertex\;} v} \deg(v) \geq |E(G)|/2.
\end{align}
 
 Let us use $IN(u)$ and $OUT(u)$ to denote in-neighbors and out-neighbors of a vertex $u$. Consider a good node $v$ and consider all its incoming neighbors $u$, i.e., neighbors $u$ such that $(\deg(u), ID(u))<(\deg(v), ID(v))$. Since $v$ is good, it has at least $\deg(v)/3$ such neighbors. 
Hence, we have 

\[\sum_{\textit{incoming neighbor}\; u} \frac{1}{\deg(u)} \geq 1/3.\] 
Choose a subset $IN^*(v)\subseteq IN(v)$ of incoming neighbors such that 
\begin{align}
\label{eq:luby1a}
\sum_{u \in IN^{*}(v)} \frac{1}{\deg(u)} \in [1/3, 4/3].
\end{align}
 
Notice that such a subset $IN^*(v)$ exists since the summation over all incoming neighbors is at least $1/3$ and each neighbor contributes at most $1$ to the summation. 
On the other hand, notice that for any node $u$, we have 
\begin{align}
\label{eq:luby2a}
\sum_{w \in OUT(u)} \frac{1}{\deg(w)} \leq 1.    
\end{align}
This is because $|OUT(u)| \leq \deg(u)$ and for each ${w}\in OUT(u)$, we have $(\deg({w}), ID({w}))>(\deg(u), ID(u))$.

A sufficient event $\mathcal{E}(v,u)$ that causes $v$ to be removed is if some $u\in IN^*(v)$ is marked and no other node in $IN^{*}(v)\cup OUT(u)$ is marked. 
By union bound, this event's indicator is lower bounded by $$x_u - \sum_{u' \in IN^*(v), u\neq u'} x_{u}\cdot x_{u'} - \sum_{w \in OUT(u)} x_{u}\cdot x_{w}.$$ 
Furthermore, the events $\fE(v,u_1), \fE(v,u_2), \dots, \fE(v,u_{|IN^*(v)|})$ are mutually disjoint for different $u_1, u_2, \dots,$ $u_{|IN^*(v)|} \in IN^*(v)$. 
Hence, we can sum over these events for different $u\in IN^*(v)$ and conclude that

\fullOnly{
\begin{align*}
R_v(\vec{x})  \geq \sum_{u\in IN^*(v)} \bigg(x_u - \sum_{u' \in IN^*(v), u\neq u'} x_{u}\cdot x_{u'} - \sum_{w \in OUT(u)} x_{u}\cdot x_{w}\bigg)  \\
  =
\sum_{u\in IN^*(v)} x_u - \sum_{u, u' \in IN^*(v)} x_{u}\cdot x_{u'}
- \sum_{u\in IN^*(v)}\sum_{w \in OUT(u)} x_{u}\cdot x_{w}.
\end{align*}
}
\shortOnly{

\begin{align*}
R_v(\vec{x})  &\geq \sum_{u\in IN^*(v)} \bigg(x_u  \\
&\quad - \sum_{\substack{u' \in IN^*(v), \\ u\neq u'}} x_{u}\cdot x_{u'} \\
&\quad - \sum_{w \in OUT(u)} x_{u}\cdot x_{w}\bigg)  \\
&= \sum_{u\in IN^*(v)} x_u \\
&\quad - \sum_{u, u' \in IN^*(v)} x_{u}\cdot x_{u'} \\
&\quad - \sum_{u\in IN^*(v)}\sum_{w \in OUT(u)} x_{u}\cdot x_{w}.
\end{align*}
}

Therefore, our overall pessimistic estimator for the number of removed edges gives that

\fullOnly{
\begin{align*}
Z(\vec{x}) \geq &\sum_{\textit{good vertex \,} v}  (\deg(v)/2) \cdot R_v(\vec{x}) &  \\
\geq   &\sum_{\textit{good vertex \,} v} (\deg(v)/2) \cdot \bigg(
\sum_{u\in IN^*(v)} x_u - \sum_{u, u' \in IN^*(v)} x_{u}\cdot x_{u'}  -  \sum_{u\in IN^*(v)}\sum_{w \in OUT(u)} x_{u}\cdot x_{w}\bigg).
\end{align*}
}
\shortOnly{
\begin{align*}
Z(\vec{x}) &\geq \sum_{\textit{good vertex \,} v}  (\deg(v)/2) \cdot R_v(\vec{x})  \\
&\geq \sum_{\textit{good vertex \,} v} (\deg(v)/2) \\
&\quad \cdot \bigg( \sum_{u\in IN^*(v)} x_u \\
&\qquad - \sum_{u, u' \in IN^*(v)} x_{u}\cdot x_{u'}  \\
&\qquad - \sum_{u\in IN^*(v)}\sum_{w \in OUT(u)} x_{u}\cdot x_{w}\bigg).
\end{align*}
}

Now, consider that we relax the label assignment such that it also allows for a fractional assignment $\vec{x} \in [0,1]^{V(G)}$. Then, $Z(\vec{x})$ is a pessimistic estimator on the expected number of edges removed if we mark each vertex $u$ fully independently (or pairwise independent) with probability $x_u$.
In Luby's algorithm, one marks each node $u \in V(G)$ with probability $x_u := \frac{1}{20\deg(u)}$, and a simple calculation shows that for this marking probability we get $Z(\vec{x}) = \Omega(|E(G)|)$. 

\paragraph{Step 1: Intra-Cluster Rounding}

We round $\vec{x}$ to a $q$-integral $\vec{y}$, for some $q = \Omega\left(\frac{1}{DEG \log n}\right)$, such that $Z(\vec{y}) = \Omega(Z(\vec{x}))$, in $O(D)$ rounds.
To that end, consider some fixed cluster $C \in \fC$. Using \cref{lemma:accelerated_probabilistic_method} with parameters

\begin{itemize}
    \item $S_{ground} = C$
    \item $\fS_{L\ref{lemma:accelerated_probabilistic_method}}$ contains the set $IN^*(v) \cap C$ for every good vertex $v$, and the set $OUT(u) \cap C$ for every vertex $u \in V(G)$
    \item $q_{L\ref{lemma:accelerated_probabilistic_method}} = \frac{1}{DEG \cdot 10^8\log(4N)}$
    \item $(x_{L\ref{lemma:accelerated_probabilistic_method}})_v = x_v$ for every $v \in C$,
\end{itemize}
we can conclude that there exists a $q$-integral $\tilde{y} \in [0,1]^C$, for some $q = \Omega\left(\frac{1}{DEG\log(n)} \right)$, such that for every $V' \in \fS_{L\ref{lemma:accelerated_probabilistic_method}} $, it holds that

\[\left|\sum_{v \in V'} x_v - \sum_{v \in V'} \tilde{y}_v\right| \leq 0.01\sum_{v \in V'} x_v + 10^6q_{L\ref{lemma:accelerated_probabilistic_method}}\log(4N).\]

Finally, we define $y_u = \tilde{y}_u/5$ for every $u \in V(G)$.

\begin{claim}
\label{cl:good}
For every good vertex $v$, we have

\[1/1000 \leq \sum_{u \in IN^*(v)} y_u \leq 1/3.\]

For every $u \in V(G)$, we have $\sum_{w \in OUT(u)} y_w \leq 1/4$.

\end{claim}
\begin{proof}

Let $S \subseteq V(G)$ with either $S = IN^*(v)$ for some good vertex $v$ or $S = OUT(u)$ for some vertex $u \in V(G)$. For any cluster $C \in \fC$, it holds that

\begin{align*}
    \sum_{u \in S \cap C} y_u &= \frac{1}{5}\sum_{u \in S \cap C} \tilde{y}_u \\
    &\leq \frac{1}{5}\left(1.01 \sum_{u \in S \cap C} x_u + 10^6q_{L\ref{lemma:accelerated_probabilistic_method}}\log(4N) \right) \\
    &\leq \left(\frac{1.01}{5}\sum_{u \in S \cap C} x_u \right) + \frac{1}{500DEG} \\
    &\leq \left(\frac{1.01}{5}\sum_{u \in S \cap C} \frac{1}{deg(u)} \right) + \frac{1}{500DEG}.
\end{align*}
Using the fact that every node $u \in V(G)$ has clustering degree at most $DEG$, we get 

\begin{align*}
\sum_{u \in S} y_u &= \sum_{C \in \fC \colon S\cap C \neq \emptyset}\sum_{u \in S \cap C} y_u\\
&\leq \sum_{C \in \fC  \colon S\cap C \neq \emptyset} \left( \frac{1.01}{5}\sum_{u \in S \cap C} \frac{1}{deg(u)} + \frac{1}{500DEG}\right) \\
&\leq \frac{1.01}{5}\sum_{u \in S} \frac{1}{\deg(u)} + \frac{1}{500}.
\end{align*}
Now, if $S = IN^*(v)$ for some good vertex $v$, we can use the fact that $\sum_{u \in IN^*(v)} \frac{1}{deg(u)} \leq 4/3$ to conclude that $\sum_{u \in IN^*(v)} y_u \leq 1/3$. Similarly, if $S = OUT(u)$ for some vertex $u \in V(G)$, we can use the fact that $\sum_{w \in OUT(u)} \frac{1}{deg(w)} \leq 1$ to conclude that $\sum_{w \in OUT(u)} y_w \leq 1/4$.

Now, again consider a good vertex $v$. For any $C \in \fC$, it holds that

\begin{align*}
    &\sum_{u \in IN^*(v) \cap C} y_u \\
    &= \frac{1}{5}\sum_{u \in IN^*(v) \cap C} \tilde{y}_u \\
    &\geq \frac{1}{5}\left(0.99 \sum_{u \in IN^*(v) \cap C} x_u - 10^6q_{L\ref{lemma:accelerated_probabilistic_method}}\log(4N) \right) \\
    &\geq \left(\frac{0.99}{5}\sum_{u \in IN^*(v) \cap C} x_u \right) - \frac{1}{500DEG} \\
    &= \left(\frac{0.99}{5}\sum_{u \in IN^*(v) \cap C} \frac{1}{20deg(u)} \right) - \frac{1}{500DEG}.
\end{align*}
Using the fact that every node $u \in V(G)$ has clustering degree at most $DEG$, we get 

\begin{align*}
&\sum_{u \in IN^*(v)} y_u \\&= \sum_{\substack{C \in \fC  \colon \\ IN^*(v) \cap C \neq \emptyset}}\sum_{u \in IN^*(v) \cap C} y_u\\
&\geq \sum_{\substack{C \in \fC  \colon \\ IN^*(v) \cap C \neq \emptyset}} \left( \frac{0.99}{5}\sum_{u \in IN^*(v) \cap C} \frac{1}{20deg(u)} - \frac{1}{500DEG}\right) \\
&\geq \left(\frac{0.99}{5}\sum_{u \in IN^*(v)} \frac{1}{20\deg(u)}\right) - \frac{1}{500} \\
&\geq \frac{1}{3}\frac{1}{20}\frac{0.99}{5} - \frac{1}{500} \geq \frac{1}{1000}.
\end{align*}

\end{proof}

\paragraph{Step 2: Local Rounding}

We next round the $\Omega\left(\frac{1}{DEG\log(n)} \right)$-integral $\vec{y} \in [0,1]^{V(G)}$ to an integral solution $\vec{b} \in \{0,1\}^{V(G)}$ using  \cref{lemma:rounding}.

For a given label assignment $\vec{b} \in \{0,1\}^{V(G)}$, we define the utility function as $$\utility(\vec{b}) = \sum_{\textit{good vertex \,} v} (\deg(v)/2) \cdot \left(\sum_{u\in IN^*(v)} b_u\right),$$ and the cost function as 

\begin{align*}
  & \cost(\vec{b}) = &&\sum_{\textit{good vertex \,} v} \bigg( (\deg(v)/2) \; \cdot \\ & &&\big(
\sum_{u, u' \in IN^*(v)} b_{u}\cdot b_{u'}  +  \sum_{u\in IN^*(v)}\sum_{w \in OUT(u)} b_{u}\cdot b_{w}\big)\bigg).  
\end{align*}

If the label assignment is relaxed to be a fractional assignment $\vec{x}\in [0,1]^{V(G)}$, where intuitively now $x_u$ is the probability of $u$ being marked, the same definitions apply for the utility and cost of this fractional assignment. 
    
    Let $G^2$ denote the graph where any two nodes of distance at most $2$ in $G$ are connected by an edge.
    Note that $\utility(\vec{x})$ is a utility function in the graph $G^2$ and similarly $\cost(\vec{x})$ is a cost function in the graph $G^2$.

\begin{claim}\label[claim]{clm:MIS1a} For the fractional label assignment $\vec{y}\in [0,1]^{V(G)}$ computed during the intra-cluster rounding step we have $\utility(\vec{y})-\cost(\vec{y}) \geq \utility(\vec{y})/3$.
\end{claim}
\begin{proof}
We start with the overall utility and cost across all good vertices:
\fullOnly{
\begin{align*}
    &\utility(\vec{y})-\cost(\vec{y}) \\
    =& \sum_{\textit{good vertex \,} v} \frac{\deg(v)}{2} \cdot \bigg(\sum_{u \in IN^*(v)} y_u - \sum_{\substack{u, u' \\ \in IN^*(v)}} y_{u}\cdot y_{u'} - \sum_{\substack{u \in IN^*(v) \\ w \in OUT(u)}} y_{u}\cdot y_{w}\bigg).
\end{align*}
}
\shortOnly{
\begin{align*}
\utility(\vec{y})-\cost(\vec{y}) 
&= \sum_{\textit{good vertex \,} v} \frac{\deg(v)}{2}  \\
&\quad \cdot \bigg(\sum_{u \in IN^*(v)} y_u - \sum_{\substack{u, u' \\ \in IN^*(v)}} y_{u}\cdot y_{u'} \\
&\qquad - \sum_{\substack{u \in IN^*(v) \\ w \in OUT(u)}} y_{u}\cdot y_{w}\bigg).
\end{align*}
}

Moreover, for any good vertex $v$, let us examine the components within this summation more closely:
\fullOnly{
\begin{align*}
    &\left(\frac{\deg(v)}{2}\right) \cdot \bigg(\sum_{u \in IN^*(v)} y_u - \sum_{\substack{u, u' \\ \in IN^*(v)}} y_{u}\cdot y_{u'} - \sum_{\substack{u \in IN^*(v) \\ w \in OUT(u)}} y_{u}\cdot y_{w}\bigg) \\
    =& \left(\frac{\deg(v)}{2}\right) \cdot \left( \sum_{u \in IN^*(v)} y_u \cdot \left(1 - \sum_{u' \in IN^*(v)} y_{u'} - \sum_{w \in OUT(u)} y_{w}\right) \right).
\end{align*}
}
\shortOnly{
\begin{align*}
&\left(\frac{\deg(v)}{2}\right) \cdot \bigg(\sum_{u \in IN^*(v)} y_u  \\
&\quad - \sum_{\substack{u, u' \\ \in IN^*(v)}} y_{u}\cdot y_{u'} \\
&\quad - \sum_{\substack{u \in IN^*(v) \\ w \in OUT(u)}} y_{u}\cdot y_{w}\bigg) \\
&= \left(\frac{\deg(v)}{2}\right) \cdot \Bigg( \sum_{u \in IN^*(v)} y_u  \\
&\quad \cdot \Big(1 - \sum_{u' \in IN^*(v)} y_{u'} - \sum_{w \in OUT(u)} y_{w}\Big) \Bigg).
\end{align*}

}

Considering the conditions from \cref{cl:good}, where for each $u \in IN^*(v)$ it is known that $\sum_{u' \in IN^*(v)} y_{u'} \leq 1/3$ and $\sum_{w \in OUT(u)} y_w \leq 1/4$, we can further deduce:
\begin{align*}
    &\left(\frac{\deg(v)}{2}\right) \cdot \bigg(\sum_{u \in IN^*(v)} y_u \cdot \left(1 - \frac{1}{3} - \frac{1}{4}\right)\bigg) \\
    \geq & \left(\frac{\deg(v)}{2}\right) \cdot \bigg(\sum_{u \in IN^*(v)} y_u \cdot \frac{5}{12}\bigg).
\end{align*}

Finally, combining these deductions for all good vertices $v$, we obtain
\begin{align*}
    &\utility(\vec{y})-\cost(\vec{y}) \\
    \geq & \sum_{\textit{good vertex \,} v} \left(\frac{\deg(v)}{2}\right) \cdot \bigg(\sum_{u \in IN^*(v)} y_u \cdot \frac{5}{12}\bigg) \geq \utility(\vec{y})/3.
\end{align*}

\end{proof}

Hence, we can apply \cref{lemma:rounding} on these fractional assignments with $\lambda_{min} = \Omega\left( \frac{1}{DEG \log(n)}\right)$. The local rounding procedure runs in $O(\log^2 (DEG\log n))$ rounds in $G^2$, and hence can be simulated with no asymptotic overhead in $G$, and as a result we get an integral label assignment $\vec{b} \in \{0,1\}^{V(G)}$ which satisfies $\utility(\vec{b}) - \cost(\vec{b}) \geq 0.5 (\utility(\vec{y}) - \cost(\vec{y}))$. We know that $Z(\vec{b})=\utility(\vec{b})-\cost(\vec{b}) \geq (1/2) \cdot (\utility(\vec{y})-\cost(\vec{y})).$ Next, we argue that this implies $Z(\vec{b}) \geq |E(G)|/24000$.

\begin{claim}\label[claim]{clm:MIS2a} For the fractional label assignment $\vec{y}\in [0,1]^{V(G)}$ computed during the intra-cluster rounding step we have $Z(\vec{y}) = \utility(\vec{y})-\cost(\vec{y}) \geq |E(G)| /12000$. Hence, for the integral marking assignment $\vec{b}$ we obtain from rounding $\vec{y}$, we have 
\[Z(\vec{b})=\utility(\vec{b})-\cost(\vec{b}) \geq (1/2) \cdot (\utility(\vec{y})-\cost(\vec{y})) \geq |E(G)|/24000.\]
\end{claim}

\begin{proof}From \Cref{clm:MIS1a}, we have $Z(\vec{y}) =\utility(\vec{y})-\cost(\vec{y}) \geq \utility(\vec{z})/3$. Hence, 

\begin{align*} Z(\vec{y}) &\geq \utility(\vec{y})/3 \\
&= 
\sum_{\textit{good vertex \,} v} (\deg(v)/2) \bigg(\sum_{u\in IN^*(v)} y_u/3\bigg)  \\
&\ge \sum_{\textit{good vertex \,} v} (\deg(v)/2) \cdot (1/3000) \geq |E(G)|/12000,
\end{align*}
where we first used \cref{cl:good} that says that $\sum_{u \in IN^*(v)} y_u \ge 1/1000$ and then we used \cref{eq:luby0a} that bounds $\sum_{\textit{good vertex \,} v} \deg(v) \geq |E(G)|/2$. 

Since $Z(\vec{b})=\utility(\vec{b})-\cost(\vec{b}) \geq (1/2) \cdot (\utility(\vec{y})-\cost(\vec{y}))$, the claim follows. 
\end{proof}

\paragraph{Conclusion}

From the rounding procedure described above, which runs in $O(D + \log^2(DEG \log n))$ rounds of the \LOCAL model,  we get an integral marking assignment $\vec{b}$ with the following guarantee: if we add marked nodes $u$ that have no marked out-neighbor to the independent set and remove them along with their neighbors, we remove at least a $1/24000$-fraction of the edges in $G$.

\bibliographystyle{alpha}
\bibliography{ref}

\end{document}